\documentclass[a4paper]{article}

\usepackage[utf8]{inputenc}

\usepackage{amsmath}
\usepackage{amssymb}
\usepackage{latexsym}
\usepackage{cmmib57}
\usepackage{amscd}
\usepackage{graphicx}
\usepackage{float}
\usepackage{epsf}
\usepackage[a4paper]{geometry}

\def\idty{{\leavevmode{\rm 1\ifmmode\mkern -4.8mu\else\kern -.3em\fi
      I}}}
\renewcommand{\Bbb}[1]{\if1#1\idty\else\mathbb{#1}\fi}

\newcommand{\kb}[1]{|#1\rangle\langle#1|}
\newcommand{\KB}[2]{|#1\rangle\langle#2|}
\newcommand{\ket}[1]{|#1\rangle}

\newcommand{\tr}{\operatorname{tr}}

\newcommand{\SP}{\operatorname{span}}

\newcommand{\sign}{\operatorname{sign}}

\newtheorem{thm}{Theorem}[section]
\newtheorem{defi}[thm]{Definition}
\newtheorem{prop}[thm]{Proposition}
\newtheorem{lem}[thm]{Lemma}
\newtheorem{kor}[thm]{Corollary}

\newenvironment{proof}{\par\noindent\textit{Proof.\ }}{\hfill $\Box$ \vspace{1em}}

\newtheorem{aX}{Axiom} 

\usepackage{tabularx}
\usepackage{wrapfig}

\newcolumntype{L}[1]{>{\raggedright\arraybackslash}p{#1}}
\newcolumntype{C}[1]{>{\centering\arraybackslash}p{#1}}
\newcolumntype{R}[1]{>{\raggedleft\arraybackslash}p{#1}}

\newcommand{\Conf}{\mathcal{C}}
\newcommand{\ConfP}{\mathcal{C}_+}
\newcommand{\ConfE}{\hat{\mathcal{C}}}
\newcommand{\Nil}{\mathrm{Nil}}
\newcommand{\PA}{\mathcal{A}_0}
\newcommand{\PAext}{\mathcal{A}}
\newcommand{\PAc}{\overline{\mathcal{A}}}

\title{Controlling a d-level atom in a cavity}

\author{Thomas Hofmann$^\dagger$ and Michael Keyl$^{\dagger\ddagger}$\\[1em]
{\small $^\dagger$ Zentrum Mathematik, M5, Technische Universit{\"a}t M{\"u}nchen,}\\
{\small Boltzmannstrasse 3, 85748 Garching, Germany}\\[1em]
{\small $^\ddagger$ Dahlem Center for Complex Quantum Systems,}\\ {\small Freie Universität Berlin, 14195 Berlin, Germany}\\[1em]
{\small \texttt{thohof@yahoo.de}, \texttt{michael.keyl@tum.de}}
}

\begin{document}

\maketitle

\begin{abstract}
  In this paper we study controllability of a $d$-level atom interacting with the electromagnetic field in a cavity. The system is modelled by an
  ordered graph $\Gamma$. The vertices of $\Gamma$ describe the energy levels and the edges allowed transitions. To each edge of $\Gamma$ we associate
  a harmonic oscillator representing one mode of the electromagnetic field. The dynamics of the system (drift) is given by a natural generalization of
  the Jaynes-Cummings Hamiltonian. If we add in addition sufficient control over the atom, the overall system (atom and em-field) becomes strongly
  controllable, i.e. each unitary on the system Hilbert space can be approximated with arbitrary precision in the strong topology by control
  unitaries. A key role in the proof is played by a topological *-algebra $\PA(\Gamma)$ which is (roughly speaking) a representation of the path algebra of 
  $\Gamma$. It contains crucial structural information about the control problem, and is therefore an important tool for the implementation of control tasks like preparing a particular state from the
  ground state. This is demonstrated by a detailed discussion of different versions of three-level systems. 


\end{abstract}

\section{Introduction}

The goal of quantum control is the systematic manipulation of the dynamical behavior of microsystems like single atoms or molecules in terms of
externally accessible parameters like laser pulses or magnetic fields. It has a wide field of applications, ranging from atomic and molecular physics,
via material science and chemistry, to biophysics and medicine; an overview over recent developments can be found in \cite{glaser2015training}. On the
mathematical side lots of knowledge is gathered about models which are based on finite dimensional Hilbert spaces like finite spin or Fermionic
systems. In particular questions of controllability and simulability are well understood, and can be efficiently solved in terms of Lie-theoretic
methods; cf. \cite{sussmann1972controllability,jurdjevic1972control,brockett1972system,brockett1973lie,jurdjevic1997geometric,zeier2011symmetry,
zimboras2014dynamic} and the references therein.

The situation becomes much more difficult if the system Hilbert space becomes infinite dimensional, since we have to deal with the challenges of
unbounded operators. There are several approaches to handle the problems, at least for large classes of Hamiltonians with pure point spectrum. The
following is a (most likely incomplete) list with corresponding references \cite{brockett2003controllability,rangan2004control,adami2005controllability,
  yuan2007controllability,nersesyan2009growth,bloch2010finite,nersesyan2010global,boscain2012adiabatic,nersesyan2012global,bliss2014quantum,
  KZSH,morancey2014global,boscain2015approximate,boscain2015control,morancey2015simultaneous,paduro2015approximate,paduro2015control,
  chitour2016generic,caponigro2017exact}

In this paper we want to concentrate on simple models for the interaction of light with atoms. The most prominent example is the Jaynes-Cummings-Model
\cite{JC63}, which describes a two-level atom, interacting with one mode of the electromagnetic field in a cavity. It is a very important tool in theory
as well as in experiments since it can be used to model many experimental setups. A single ion in a trap which is placed into an optical cavity and
controlled by external lasers, is a typical example. Mathematically it is interesting since it provides a simple (and tractable) example for
quantum control in infinite dimensions. This was explored in a number of works \cite{brockett2003controllability,rangan2004control,yuan2007controllability,
  bloch2010finite,KZSH}.

Our new contribution to this circle of question is the generalization to an arbitrary finite number of levels. As shown in Fig. \ref{fig:atom} we
describe the system in terms of an ordered graph $\Gamma$ with vertices representing energy levels and edges marking allowed transitions. To each edge
we associate a different mode of the electromagnetic field. This opens the possibility to work with photons of different frequencies, and since the
system is (as we will see) fully controllable, we can manipulate them arbitrarily. Typical examples are swapping two modes, ``joining'' two photons
into one with higher frequency, or generating entanglement between many different modes.

\begin{wrapfigure}{l}{0.5\textwidth}
  \includegraphics[width=0.4\textwidth]{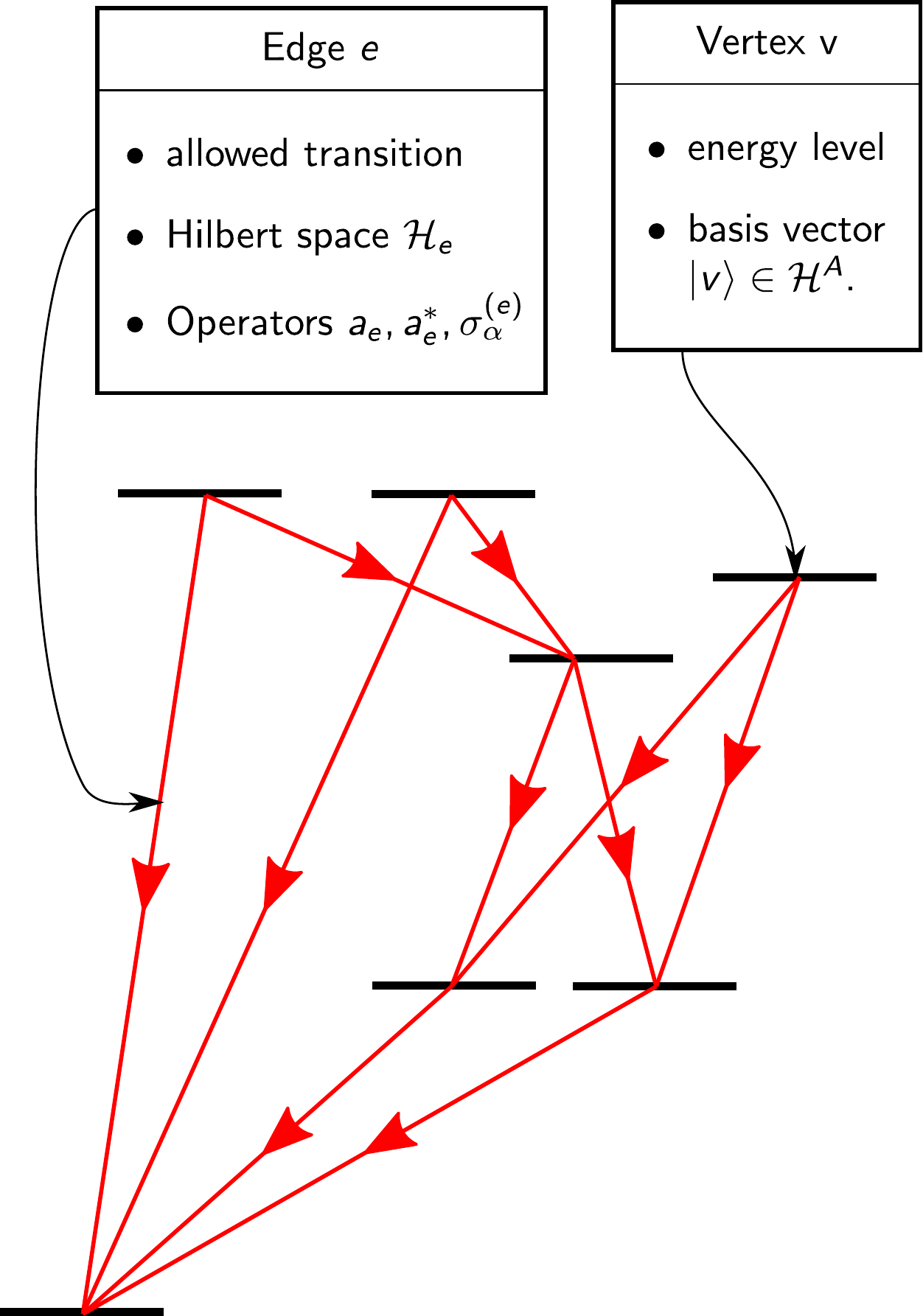}
  \caption{A graph $\Gamma$ describing a $d$-level atom interacting with a
    finite number of modes of the electromagnetic field. \medskip\label{fig:atom}}
\end{wrapfigure}
We treat the problem by generalizing the method from \cite{KZSH}, where we have used symmetry arguments to cut the infinite dynamical setting down
into an increasing sequence of finite dimensional subsystems, which are then discussed with standard methods. In this paper we replace the symmetries by
a representation of a *-algebra associated to the graph $\Gamma$. Since it is generated (roughly speaking) by the path of $\Gamma$ it is called ``path
algebra'' in the following. This concept was taken from quiver theory, where path algebras are an important tool
\cite{boevey1992lectures,savage2006finite}. In our case their relevance does not only arise in the proof of controllability, but also in the
structural analysis of the corresponding control problem. This can can be of great importance for the implementation of explicit control task, like
preparing a particular state from  the ground state, or the development of efficient algorithms for optimal control. We will demonstrate this by the
discussion of different versions of three level systems. From a purely mathematical point of view the path algebra is interesting as well, since its
structure is closely related to the structure of the graph $\Gamma$. We claim that the path algebras associated to two different ordered graphs
$\Gamma$ and $\Gamma'$ are equivalent, iff $\Gamma$ and $\Gamma'$ are equivalent. Note in this context that our setting is slightly different from the one known in quiver theory. Our path algebra is
in particular always infinite dimensional -- even in the most simple case of the connected graph with two vertices and one edge (cf. the discussion on Sect. \ref{sec:example-1:-two}).

The organization of the paper is as follows: In Section \ref{sec:description-problem} we state the main controllability theorem, together with
technical result which are necessary to formulate the statement in the first place (selfadjointness and recurrence properties of certain
operators). This is followed in Sect. \ref{sec:path-algebra} by the definition and detailed discussion of the path algebra. Section
\ref{sec:spectral-analysis} contains some spectral analysis, providing proofs for some of statements from Section
\ref{sec:description-problem}. Technical statements about dynamical groups are contained in Section \ref{sec:dynamical-group} and applied in Section
\ref{sec:full-controlability} to prove controllability. Using the newly introduced language, the work on two-level systems from \cite{KZSH} is
reviewed in Section \ref{sec:example-1:-two}, while a detailed discussion of the three-level case is given in Sections \ref{sec:example-2:-three} and
\ref{sec:example-3:-delta}. The paper closes with an outlook in Section \ref{sec:outlook}.

\section{Description of the problem}
\label{sec:description-problem}

We will describe the atom in terms of a graph such that the vertices become
energy levels and the edges allowed transitions; cf. Fig. \ref{fig:atom}. Therefore, let us
introduce some terminology from graph theory first (cf. \cite{diestel2006graph} for detailed discussion). A graph $\Gamma$ consists of the sets $V(\Gamma)$ of
vertices, $E(\Gamma)$ of edges and the maps 
\begin{gather}
  I : E \rightarrow V \times V, e \mapsto I(e) = \bigl(i(e),t(e)\bigr),\\
  \overline{\phantom{e}}: E \rightarrow E, e \mapsto \overline{e},
\end{gather}
such that for all $e \in E$ the following three conditions hold: 
\begin{equation}
  \overline{e} \neq e,\quad \overline{\overline{e}} = e\quad i(\overline{e}) = t(e), 
\end{equation}
and the relation $t(\overline{e}) = i(e)$ which can be derived from the other three. Hence edges have a direction and always come in pairs: $e$ points
from $i(e)$ to $t(e)$ and $\overline{e}$ the other way round. The pair $g=\{e,\overline{e}\}$ is called a geometric edge and only contains information
about the link between two vertices not about the direction. If we distinguish exactly one edge in each geometric edge we get a directed graph. More
precisely a directed graph is a graph together with a set $E_+ \subset E$ such that $e \in E_+ \Leftrightarrow \overline{e} \in E_- = E \setminus
E_+$. We will call edges in $E_+$ positive and edges in $E_-$ negative. 

\paragraph{The graph}
Now consider an oriented graph $\Gamma$ with finite sets of vertices and edges. We associated a vertex $v \in V(\Gamma)$ to each energy level of our
atom and an edge $e \in E_+ \subset E(\Gamma)$ for each allowed transition. The orientation of the latter is chosen such that $e \in E_+$ always
points from higher to lower energies. From this picture we deduce the following assumption on $\Gamma$ which will hold throughout the paper:
\begin{itemize}
\item 
  \emph{No loops:} No vertex is connected to itself, i.e. there is no edge $e$ with $i(e)=t(e)$.
\item 
  \emph{No double edges:} If $i(e_1) = i(e_2)$ and $t(e_1) = t(e_2)$ hold for two edges $e_1, e_2$ we have $e_1 = e_2$.
\item 
  \emph{Connectedness:} The graph is connected: Each pair of vertices $v_1, v_2$ can be connected by a path $\gamma$, i.e. a sequence $\gamma =
  (e_1,\dots,e_N) \in E(\Gamma)^N$, with $i(e_{j+1}) = t(e_j)$ for all $j=1,\dots,N-1$, and $i(e_1) = v_1$ and $t(e_N) = v_2$.
\item 
  \emph{No ordered cycles:} A cycle is a non-empty path $\gamma=(e_1, \dots, e_N)$ with and $t(e_N) = i(e_1)$ (i.e. a closed path which connects a 
  vertex with itself). We assume that there are no cycles with $e_j \in E_+$ for all $j=1,\dots,N$ or $e_j \in E_-$ for all $j=1,\dots,N$.
\end{itemize}
All four assumption are natural for the physical situation we want to describe, and at the same time they are crucial for the proofs we are going to
present. They allow us in particular to define a partial ordering $\leq$ on $V$ by: $v_1 \geq v_2$ $:\Leftrightarrow$ there is a path $\gamma =
(e_1, \dots, e_N)$ with $e_j \in E_+$ $\forall j=1,\dots,N$ (i.e. an \emph{ordered} path) from $v_1$ to $v_2$ (i.e. $i(e_1) = v_1$ and $t(e_N) =
v_2$). Note that we allow explicitly an empty path\footnote{Please note that we are considering only one empty path for the whole graph, and not one for each vertex as in quiver theory.} $\gamma =
()$ as the only connection from a vertex $v$ to itself - this makes the relation reflexive. Please check yourself that all other conditions for a partial ordering (transitivity, antisymmetry) are
satisfied as well. 

\paragraph{Configurations}
An important concept in this paper are configurations. A configurations is a pair $b = (b_0,\underline{b})$ consisting of a vertex $b_0 \in V(\Gamma)$
-- called the \emph{current level}, and a map $\underline{b}$ from $E_+(\Gamma)$ into the integers $\Bbb{Z}$, which we will call the \emph{number
  map}. Hence the set $\Conf(\Gamma)$ of all configurations is given by
\begin{equation}
   \Conf(\Gamma) = \{ b=(b_0,\underline{b})\,|\, b_0 \in V(\Gamma),\quad \underline{b} : E_+ \rightarrow \Bbb{Z} \}.
 \end{equation}
The basic idea behind this definition is that a configuration describes a state of the system, where the current level represents the state of the atom
and the number map describes the number of photons in each mode (cf. next paragraph). The latter requires that $\underline{b}(e) \geq 0$ holds for all
$e \in E_+$. Each configuration satisfying this requirement is called \emph{regular}. The set of regular configurations is
\begin{equation}
  \ConfP(\Gamma) = \{ b \in \Conf(\Gamma)\,|\, \underline{b}(e) \geq 0\quad \forall e \in E_+\}.
\end{equation}

\begin{wrapfigure}{l}{0.5\textwidth}
  \includegraphics[width=0.4\textwidth]{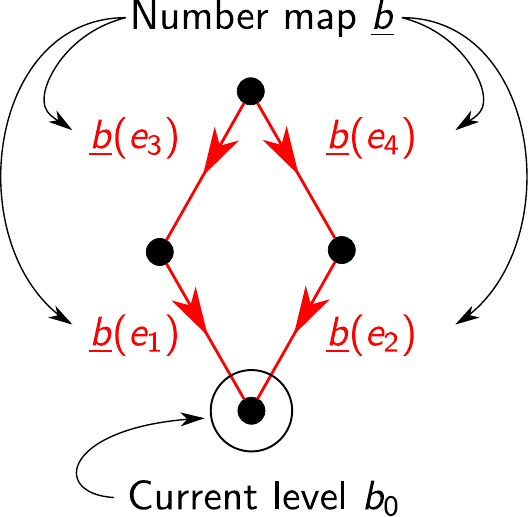}
  \caption{Graphical representation of the configuration $b$ of a graph $\Gamma$ with four vertices and edges
    $E_+(\Gamma)=\{e_1,e_2,e_3,e_4\}$. \label{fig:conf}}
\end{wrapfigure}
\noindent Configurations have a nice graphical representation as shown in Fig. \ref{fig:conf}. Occasionally this will turn out handy, to represent the actions of certain
operators in a graphical way, cf. Sect. \ref{sec:path-algebra}. Finally we introduce the extended configuration set by
\begin{equation}
  \ConfE(\Gamma) = \Conf(\Gamma) \cup \{\Nil\}.
\end{equation}
The extra \emph{nil-configuration} we are adding here is needed later (cf. in particular Sect. \ref{sec:path-algebra}) to serve as the output of some
operation which are otherwise undefined.

\paragraph{The Hilbert space}
The atom is described by the Hilbert space $\mathcal{H}^A= \Bbb{C}^d$, where $d = |V(\Gamma)| \in  \Bbb{N}$
denotes the cardinality of $V(\Gamma)$ and hence the number of energy levels we want to consider. Each allowed transition $e \in E_+$ is connected to a
different mode of the light field described by a Hilbert space $\mathcal{H}^C_e = \mathrm{L}^2(\Bbb{R})$. Hence the overall Hilbert space describing
the atom and the photons interacting with it is 
\begin{equation}\label{eq:17}
  \mathcal{H} = \mathcal{H}^A \otimes \mathcal{H}^C,\quad \mathcal{H}^C = \bigotimes_{e \in E_+} \mathcal{H}^C_e.
\end{equation}
Frequently we will call $\mathcal{H}^A$ and $\mathcal{H}^C$ the atom and cavity Hilbert space respectively. To get a distinguished basis we choose the
canonical basis $\ket{v} \in \mathcal{H}_V = \Bbb{C}^d$, $v \in V(\Gamma)$ and for each $e \in E_+$ the number basis $\ket{n;e} \in \mathcal{H}_e =
\mathrm{L}^2(\Bbb{R})$. Together we have for $b = (b_0,\underline{b}) \in \ConfP(\Gamma)$ 
\begin{equation} \label{eq:1}
  \ket{b} = \ket{b_0;\underline{b}} = \ket{b_0} \otimes \bigotimes_{e \in E_+} \ket{\underline{b}(e),e}.
\end{equation}
Hence the states described by the basis vectors $\ket{b}$ perfectly fit the intuitive interpretation of regular configurations given above. To keep the
notations consistent we define $\ket{b} = 0$ for all non-regular configurations, i.e.
\begin{equation}
  \ket{b} = 0 \quad \forall b \in \ConfE(\Gamma) \setminus \ConfP(\Gamma),\quad\text{in particular}\quad \ket{\Nil} = 0.
\end{equation}
The space of finite linear combinations of basis vectors $\ket{b}$  gives rise to a dense subspace
\begin{equation} \label{eq:23}
  D_\Gamma = \SP\{\ket{b}\,|\, b \in \ConfP(\Gamma)\} \subset \mathcal{H}
\end{equation}
which will serve as the domain of several unbounded operators. 

\paragraph{The operators}
Our next step is to associate certain operators to the vertices and edges of $\Gamma$. To this end let us define for each $e \in E_+$ the subspace
$\mathcal{H}^A_e \subset \mathcal{H}^A$ generated by $\ket{i(e)}, \ket{t(e)}$. We will identify it with $\Bbb{C}^2$ by the map $\varphi_e : \Bbb{C}^2
\rightarrow \mathcal{H}^A_e$ with $\varphi(\ket{0}) = \ket{t(e)}$ and $\varphi(\ket{1}) = \ket{i(e)}$ (i.e. $\varphi$ respects the ``ordering'' $0 <
1$ and $t(e) < i(e)$). With this map we can define 
\begin{equation}
  \mathcal{B}(\Bbb{C}^2) \ni X \mapsto X^{(e)} \in \mathcal{B}(\mathcal{H}^A),\quad \text{with}\quad X^{(e)} v = 
  \begin{cases}
    0 & \text{for $v \in \left(\mathcal{H}^A_e\right)^{\perp}$}\\
    \varphi_e X \varphi_e^{-1} v & \text{for $v \in \mathcal{H}^A_e$},
  \end{cases}
\end{equation}
i.e. $X^{(e)}$ acts as $X$ on $\mathcal{H}^A_e$ and as $0$ otherwise. Of particular importance for the following are the Pauli operators
$\sigma_{\alpha}^{(e)}$, $\alpha =1,\dots,3,\pm$. 

If $Y$ is a (possibly unbounded) operator on $\mathrm{L}^2(\Bbb{R})$ we define $Y_e$ as the operator on $\mathcal{H}^C$ which acts as $Y$ on
$\mathcal{H}^C_e$ and as the identity on all other tensor factors, i.e.
\begin{equation}
  Y_e = Y \otimes \bigotimes_{\stackrel{f \in E_+}{f \neq e}} \Bbb{1}_f 
\end{equation}
where $\Bbb{1}_f$ denotes the unit operator on $\mathcal{H}^C_f$. Of particular importance for us are $a_e, a^*_e$ where $a, a^*$ are the usual
annihilation and creation operators.

Now note that the operators of the form $X^{(e)} \otimes a_f$ or $X^{(e)} \otimes a^*_f$ with an arbitrary $X \in \mathcal{B}(\mathcal{H}^A)$ map the
domain $D_\Gamma \subset H$ into itself, such that $D_\Gamma$ becomes an invariant, dense domain for these operators. Hence we can define for all
$\psi \in D_\Gamma$:
\begin{equation}\label{eq:4}
  H_X \psi = X \otimes \Bbb{1}^C\psi + \sum_{e \in E_+} \left[\omega_{C,e} \Bbb{1}^A \otimes a^*_ea_e \psi + \omega_{I,e} \left(\sigma_+^{(e)} \otimes a_e +
    \sigma_-^{(e)} \otimes a^*_e\right)\right]  
\end{equation}
where $X \in \mathcal{B}(\mathcal{H}^A)$ is an arbitrary, selfadjoint operator, the $\omega_{C,e}, \omega_{I,e}$ are arbitrary real
constants, and  $\Bbb{1}^A$, $\Bbb{1}^C$ are unit operators on the atom and cavity Hilbert spaces, respectively. Operators of this form are the
Hamiltonians we are going to study. Here $X \otimes \Bbb{1}^C$ and $\sum \omega_{A,e} \Bbb{1}^A \otimes a^*_ea_e$ describe the free evolution of the atom
and cavity respectively, while $\sum \omega_{I,e} \left(\sigma_+^{(e)} \otimes a_e + \sigma_-^{(e)} \otimes a^*_e\right)$ is the interaction
term. In other words we have (roughly speaking) for each edge $e \in E_+$ a Jaynes-Cummings type ``sub-Hamiltonian''. Our first main result is the
following:

\begin{thm} \label{thm:1}
  For each $X \in \mathcal{B}(\mathcal{H}^A)$ and for all real constants $\omega_{C,e}, \omega_{I,e}$, $e \in E_+$, the operator $H_X$ defined in
  Eq. (\ref{eq:4}) has the following properties: 
  \begin{enumerate}
  \item 
    \emph{Self-adjointness:} $H_X$ is essentially selfadjoint on the domain $D_\Gamma$ and in abuse of notation we will denote its self-adjoint
    extension by the same symbol. 
  \item 
    \emph{Recurrence:} For all $t_- \in \Bbb{R}$, $t_- \leq 0$ and all strong neighborhoods $V$ of $\exp(i t_- H_X)$ there is a time $t\in \Bbb{R}$, $t_+
    > 0$ with $\exp(i t_+ H_X) \in V$.  
  \end{enumerate}
\end{thm}
 
Self-adjointness guarantees the existence of time-evolution operators $\exp(i t H_X)$ for all $t \in \Bbb{R}$, while recurrence tells us that it is
sufficient to look at positive times, since we can find time-evolutions into the past in the strong closure of time-evolutions into the future.

\paragraph{Control}
We introduce the drift Hamiltonian $H_D$ as a variant of $H_X$ with $X$ diagonal in the basis $\ket{v}$, $v \in V$:
\begin{equation}\label{eq:38}
   H_D \psi = \sum_{e \in E_+} \left[\omega_{A,e} \sigma_3^{(e)} \otimes \Bbb{1}^C  + \omega_{C,e} \Bbb{1}^A \otimes a^*_ea_e \psi + \omega_{I,e}
     \left(\sigma_+^{(e)} \otimes a_e + \sigma_-^{(e)} \otimes a^*_e\right)\right] 
\end{equation}
with another family $\omega_{A,e}$ of (positive) real constants. As control Hamiltonians we consider all possible $\sigma_3^{(e)}$ and
$\sigma_1^{(e)}$ rotations on the atom
\begin{equation}
  X^{(e)} = \sigma_1^{(e)} \otimes \Bbb{1}^C,\quad Y^{(e)} = \sigma_3^{(e)} \otimes \Bbb{1}^C \quad e \in E_+ 
\end{equation}
It is  easy to see (since $\Gamma$ is connected) that the atom alone has to be fully controllable; cf. Lemma \ref{lem:7}. We do not assume, however,
any direct control over the field or the interaction. The control functions are chosen to be piecewise constant. Hence we introduce the space
$\mathcal{P}$ of maps ($\Bbb{R}^{E_+}$ denotes the set of all functions $E_+ \rightarrow \Bbb{R}$)
\begin{equation}
  u : \Bbb{R} \rightarrow \Bbb{R}^{E_+}, t \mapsto u(t) = (u_e(t))_{e \in E_+}\quad\text{with}\quad u_e : \Bbb{R} \rightarrow \Bbb{R}
\end{equation}
such that there are $0 < t_1 < \dots < t_N = T$ and $u^{(j)} \in \Bbb{R}^{E_+}$, $j=1, \dots, N$ with
\begin{equation}
  u(t) = 0\ \forall t \not\in (0,T]\quad \text{and}\quad u(t) = u^{(j)}\ \forall t \in (t_{j-1},t_j]\ \forall j=1,\dots,N.
\end{equation}
Note that the control time $T$ is determined by $u$ via $T = \sup \{ t \in \Bbb{R} \,|\, u_t \not= 0 \}$. Each pair $(u,v) \in \mathcal{P}$ leads to
the time-dependent Hamiltonian $t \mapsto H_{u(t),v(t)}$, $t \in \Bbb{R}$ with
\begin{equation} \label{eq:7}
  H_{x,y} = H_D + \sum_{e \in E_+} \left(x_e X^{(e)} + y_e Y^{(e)}\right) \quad (x,y) \in \Bbb{R}^{E_+} \times \Bbb{R}^{E_+}
\end{equation}
and therefore to the control problem
\begin{equation} \label{eq:5}
  i \frac{d}{dt} U_{u,v}(0,t) \psi = H_{u(t),v(t)} U_{u,v}(0,t) \psi.
\end{equation}
Since $H_{u(t),v(t)}$ is piecewise constant, the unitary time-evolution operator is given as a product of exponentials $\exp\left(i \Delta t_j
  H_{u(t_j),v(t_j)}\right)$; e.g. for $t=T$ and $\Delta t_j = t_j - t_{j-1}$ with $j=1,\dots,N$, $t_0 = 0$ we get:
\begin{equation} \label{eq:6}
  U_{u,v}(0,T) = \exp\left(i \Delta t_N H_{u(t_N),v(t_N)}\right) \dots \exp\left(i \Delta t_1 H_{u(t_1),v(t_1)}(t_1)\right). 
\end{equation}
Our main result shows that the control problem in (\ref{eq:5}) is \emph{strongly controllable} \cite{KZSH}, i.e. that all unitaries $U$ on $\mathcal{H}$
can be realized (up to a phase factor) as the limit of a strongly convergent sequence (or net) of operators $U_{u,v}(T)$.  In other words:

\begin{thm} \label{thm:2}
  The control problem from (\ref{eq:5}) is strongly controllable, i.e. for any unitary $U$ on $\mathcal{H}$ there is a constant phase factor
  $e^{i\alpha} \in \Bbb{C}$ such that the strong closure of the set
  \begin{equation} \label{eq:8}
    \mathcal{M} = \{ U_{u,v}(0,T) \, | \, (u,v) \in \mathcal{P}^2,\ T=T_{\max}(u,v)\},\quad T_{\max}(u,v) = \sup \{ t\in \Bbb{R}\,|\, u(t) \neq 0,\
    v(t) \neq 0 \}
  \end{equation}
  contains $e^{i \alpha} U$.
\end{thm}

Note that this means we only need (complete) control over the atom to gain complete control over the photonic modes in
the cavity.

\section{The path algebra}
\label{sec:path-algebra}

In this section we will study an algebraic representation of the graph which is very important for the analysis of the control problem just
introduced. Note that that some ideas used here are taken from quiver theory \cite{savage2006finite,boevey1992lectures}, however, our setup is slightly different, and in particular more special, since
we have to serve the needs of our control problem. To start we introduce the operation
\begin{equation}
  E(\Gamma) \times \ConfE(\Gamma) \ni (e,b) \mapsto e\cdot b \in \ConfE(\Gamma)
\end{equation}
which is defined as follows:
\begin{enumerate}
\item 
  If $e$ starts at $b_0$ the current level is moved to the end $t(e)$ of $e$ and the number $n(e)$ is incremented (if $e \in E_+$) or decremented (if
  $e \in E_-$). In other words if $b_0 = i(e)$ we have 
  \begin{equation}
    e \cdot b = b',\quad\text{with}\quad b'_0 = t(e),\ \underline{b'}(e') =  \underline{b}(e') + \sign(e)  \delta_{ee'},
  \end{equation}
  where $\delta_{ee}=1$ and $\delta_{ee'}=0$ for $e\neq e'$. The signum $\sign(e)$ of $e \in E(\Gamma)$ is $+1$ for positive edges ($e \in E_+$) and
  $-1$ otherwise ($e \in E_-$). The whole operation is best described graphically as shown in figure \ref{fig:confmap}.
\item 
  If the edge $e$ does not start at the current level of $b$ the latter is mapped to $\Nil$; i.e. $i(e) \neq b_0$ $\Rightarrow$ $e \cdot b = \Nil$. In
  particular we have $e \cdot \Nil = \Nil$ for all edges $e \in E(\Gamma)$. 
\end{enumerate}
Note that $e\cdot b$ is basically only a partially defined operation. For notational purposes we have, however, introduced the nil-configuration to
turn this into a proper operation on the set $\ConfE(\Gamma)$.

\begin{figure}[b]
  \centering
  \includegraphics[width=0.8\textwidth]{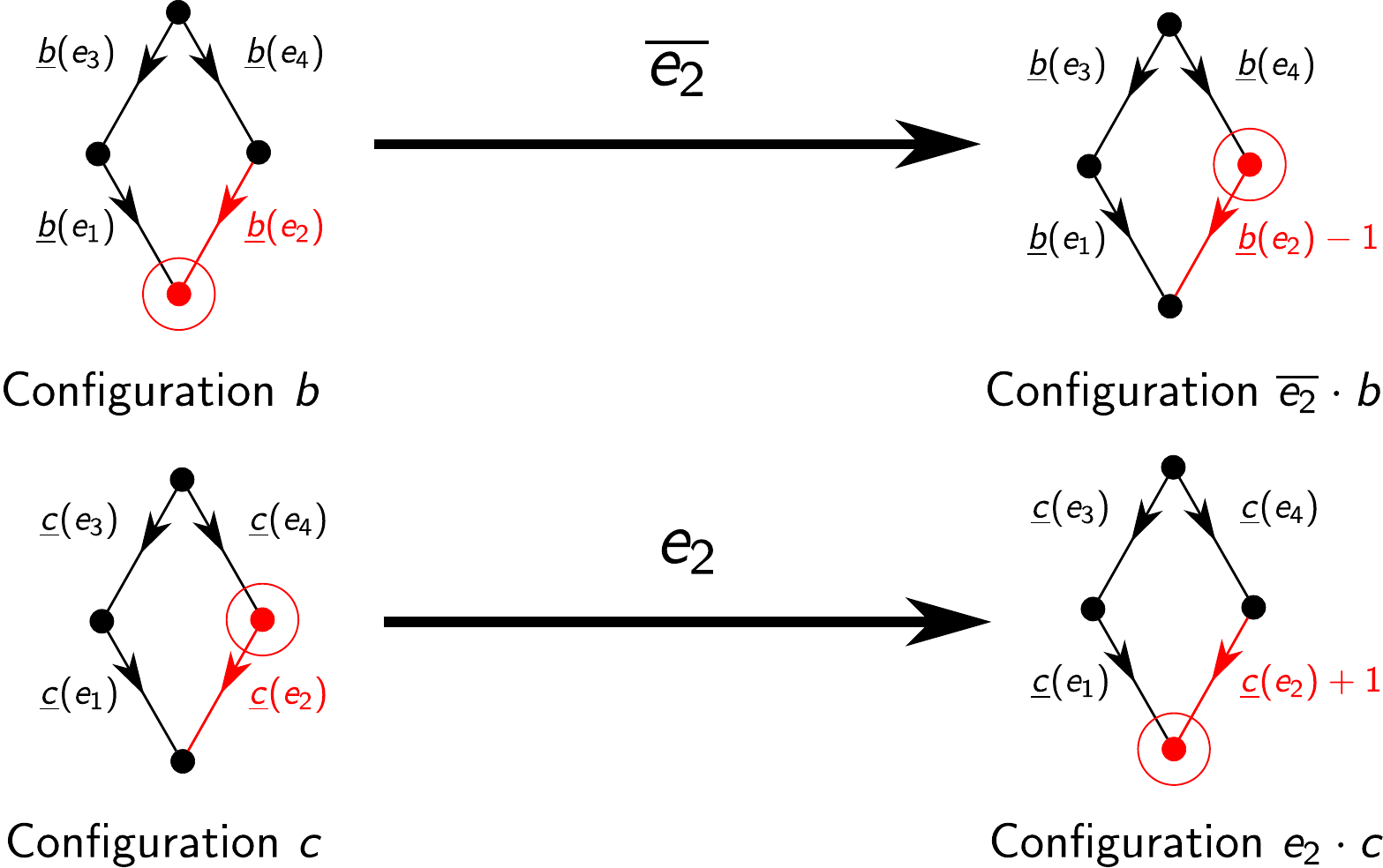}
  \caption{Graphical representation of the map $b \mapsto e \cdot b$ on configurations; cf. also Fig. \ref{fig:conf}. In the first line the
    configuration $b$ is mapped along the \emph{negative} edge $\overline{e_2}$. Hence the photon number $\underline{b}(e_2)$ is decremented. In the second line
    $c$ is mapped along the \emph{positive} edge $e_2$ and therefore $\underline{c}(e_2)$ is incremented. In both cases the current level is shifted
    from the beginning to the end of the current edge -- against the arrow for $\overline{e}_2$ and in the direction of the arrow for
    $e_2$. \label{fig:confmap}} 
\end{figure}

Using the standard basis $\ket{b}$, $b \in \Conf(\Gamma)$ of $\mathcal{H}$, each edge defines a bounded operator by $\ket{b} \mapsto \ket{e\cdot b}$. Note
here that all cases where $e \cdot b$ is not regular leads to $\ket{e \cdot b}=0$. We have in particular $\ket{\Nil} = 0$. Another important case where
the operator just defined gives $0$ arise for $e \in E_-$ with $b_0 = i(e)$, if $\underline{b}(\overline{e}) = 0$, since decrementing $\underline{b}(\overline{e})$ leads to
a non-regular configuration $e \cdot b$. Now we define in addition 
\begin{equation} \label{eq:20}
  \alpha(b,e) =
  \begin{cases}
    \sqrt{\underline{b}(e)+1} &  \text{if $e \in E_+$, $b \in \ConfP(\Gamma)$} \\
    \sqrt{\underline{b}(\overline{e})} & \text{if $e \in E_-$, $b \in \ConfP(\Gamma)$}\\
    0 & \text{if $b$ is not regular}
  \end{cases}
\end{equation}
and the operators
\begin{equation}\label{eq:16}
  A_e : D_\Gamma \rightarrow D_\Gamma \subset \mathcal{H},\quad A_e \ket{b} = \alpha(b,e) \ket{e \cdot b},
\end{equation}
which can alternatively be written as:
\begin{equation} \label{eq:13}
  A_e =
  \begin{cases}
    \sigma_-^{(e)} \otimes a_e & \text{for}\ e \in E_+\\
    \sigma_+^{(\overline{e})} \otimes a_{\overline{e}}^* & \text{for}\ e \in E_-.
  \end{cases}
\end{equation}
The $A_e$ leave the domain $D_\Gamma$ invariant. Therefore arbitrary products and linear combinations of them are well
defined. This leads to

\begin{defi}
  The associative, complex algebra $\PA(\Gamma)$ generated by the family of operators $A_e$, $e \in E(\Gamma)$ is
  called \emph{path algebra}. If we add all operators which are diagonal in the basis $\ket{b}$, $b \in
  \ConfP(\Gamma)$ as generators, we get the \emph{extended path algebra} $\PAext(\Gamma)$. 
\end{defi}

The name of $\PA(\Gamma)$ arises from the fact that a monomial $A_{e_N} \dots A_{e_1}$ of $A$-operators is
nonzero  iff $t(e_j) = i(e_{j+1})$ holds for all $k =1, \dots, N-1$. In other words the collection $\gamma =
(e_1,\dots,e_N)$ has to be a path in $\Gamma$ and elements of $\PA(\Gamma)$ represent in a certain way
``superpositions'' of paths. For a path $\gamma = (e_1, \dots, e_N)$ we will write 
\begin{equation} \label{eq:3}
  \gamma \cdot b = e_N \cdot \dots e_1 \cdot b\quad \text{and}\quad A_\gamma = A_{e_N} \dots A_{e_1} = \alpha(b,\gamma)
  \ket{\gamma \cdot b} 
\end{equation}
with
\begin{equation}
  \alpha(b,\gamma) = \alpha(e_{N-1} \cdot \dots \cdot e_1 \cdot b, e_N) \dots \alpha(e_1 \cdot b, e_2) \alpha(b, e_1).
\end{equation}
In addition we can define the subpath $\gamma_k$ of $\gamma = (e_1, \dots, e_N)$ by
\begin{equation}
  \gamma_k = (e_1, \dots, e_k),\ \text{if}\ k=1,\dots, N; \quad \gamma_0 = (),
\end{equation}
where $()$ denotes the empty path. Using this notation the quantity $\alpha(b,\gamma)$ gets an alternative, recursive definition:
\begin{equation}
  \alpha(b,\gamma_0) = 1,\quad \alpha(b, \gamma_k) = \alpha(\gamma_{k-1},\cdot b,e_k) \alpha(b,\gamma_{k-1}).
\end{equation}
Again, there is a nice graphical representation of the action $b \mapsto \gamma \cdot b$ which is shown in Fig. \ref{fig:confmap2}. By evaluating them
on the basis $\ket{b}$, $b \in \ConfP(\Gamma)$, it easily seen that the $A_\gamma$ form a linearly independent family, which therefore becomes a basis
of $\PA(\Gamma)$.

Path algebras are a well known and important concept in the theory of quivers \cite{boevey1992lectures,savage2006finite}. In that context they are
defined in a more abstract way as the associative algebra over a field $F$ which has (as a vector space) the paths of $\Gamma$ as a basis and with
multiplication given by concatenation of paths (if a path $\gamma_1$ does not end at the vertex where  a second path $\gamma_2$ starts the product
$\gamma_2 \gamma_1$ is zero). The discussion of the previous paragraph clearly shows that $\PA(\Gamma)$ and this abstractly defined path algebra are closely related. We might even think that
$\PA(\Gamma)$ is a representation of the latter (in the case $F=\Bbb{C}$). However, this is not the case since our setup and quiver theory work with different definitions of paths. In our case a path
can consist of positive and negative edges (i.e. we are allowed to move back and forth), while in quiver theory only positive edges are allowed. As a result the abstract path algebra for oriented
graphs $\Gamma$ satsifying the condition from Sect. \ref{sec:description-problem} is always finite dimensional \cite{boevey1992lectures}, while $\PA(\Gamma)$ is always infinite dimensional. A second
more subtle difference arises from the treatment of the empty path. We are using one empty path which can be concatenated with any other path. In quiver theory there is a different empty path for each 
vertex $v$ (which can only be concatenated with path starting or ending at $v$).

The importance of the path algebra for our purposes arise from the fact that all operators $H_X$ from Eq. (\ref{eq:4}) with diagonal $X$ are elements
of $\PAext(\Gamma)$. Furthermore we have the following theorem: 

\begin{figure}[t]
  \centering
  \includegraphics[width=0.9\textwidth]{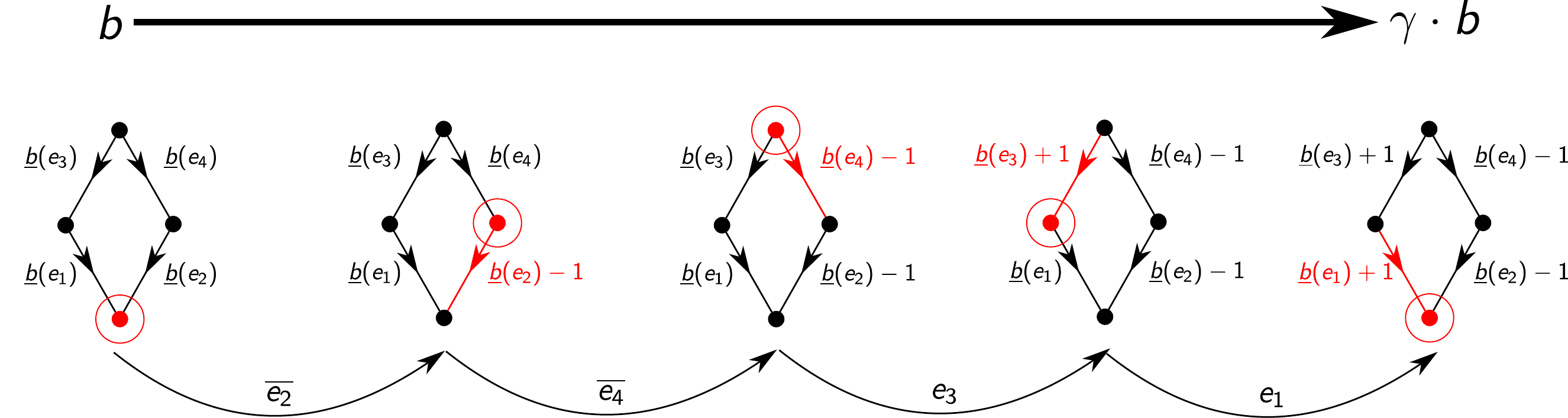}
  \caption{Graphical representation of the map $b \mapsto \gamma \cdot b$ with $\gamma = (\overline{e_2}, \overline{e_4}, e_3, e_1)$; i.e. one cycle
    around the graph $\Gamma$; cf. also Figure \ref{fig:confmap}. \label{fig:confmap2}} 
\end{figure}

\begin{thm} \label{thm:3}
  The Hilbert space $\mathcal{H}$ decomposes into a direct sum $\mathcal{H}=\bigoplus_{n \in \Bbb{N}} \mathcal{H}^{(n)}$ of
  finite dimensional subspaces $\mathcal{H}^{(n)} \subset D_\Gamma \subset \mathcal{H}$ with corresponding projections $P^{(n)}:
  \mathcal{H} \rightarrow \mathcal{H}^{(n)}$ such that
  \begin{enumerate}
  \item 
    $\psi \in D_\Gamma$ iff $P^{(n)} \psi = 0$ for all but a finite number of $n \in \Bbb{N}$.
  \item 
    $\PAext(\Gamma) \mathcal{H}^{(n)} \subset \mathcal{H}^{(n)}$; i.e. the $\mathcal{H}^{(n)}$ are invariant subspaces for the
    extended path algebra $\PAext(\Gamma)$.
  \end{enumerate}
\end{thm}

\begin{proof}
  Consider $b \in \ConfP(\Gamma)$ and a path $\gamma = (e_1,\dots, e_N)$. Then $A_\gamma \ket{b}$ is according to Eq. (\ref{eq:3})
  either a scalar multiple of another basis element (i.e. $\ket{\gamma \cdot b}$) or zero. 
  Hence 
  \begin{equation}\label{eq:28}
    \mathcal{H}^{(b)} = \SP  \{ \ket{\gamma \cdot b}\, | \, \gamma\ \text{path in $\Gamma$}\,\} \subset D_\Gamma
  \end{equation}
  is an invariant subspace of $\PA(\Gamma)$, and since $\PA(\Gamma)$ and $\PAext(\Gamma)$
  differ only by elements which are diagonal in the basis $\ket{b}$ it is an invariant subspace of
  $\PAext(\Gamma)$ as well. Also note that $\ket{\gamma \cdot b} \neq 0$ holds iff
  \begin{equation} \label{eq:15}
    b_0 = i(\gamma) = i(e_1) \quad \text{and}\quad \gamma_k \cdot b \ \text{is regular}\ \forall k = 0, \dots, N,
  \end{equation}
  because, $\ket{\gamma_k \cdot b} = 0$ if $\gamma_k \cdot b$ becomes non-regular (i.e. one of the numbers $\underline{\gamma \cdot b}(e)$ becomes
  negative). This observation motivates the following lemma:

  \begin{lem} \label{lem:14}
    For $b,c \in \ConfP(\Gamma)$ define $b \sim c$ $:\Leftrightarrow \exists$ path $\gamma$ with $c = \gamma \cdot b$, and the regularity condition
    (\ref{eq:15}) holds. $b \sim c$ is an equivalence relation.
  \end{lem}

  \begin{proof}
    The relation is reflexive since $()\cdot b = b$ holds with the empty path $()$. It is transitive since $b \sim c$ and $c \sim d$ implies $c =
    \gamma \cdot b$ and $d = \xi \cdot c$ with two path $\gamma=(e_1,\dots,e_N)$, $\xi = (e_{N+1}, \dots, e_M)$ both satisfying (\ref{eq:15}). Hence
    $d = (\xi\gamma) \cdot b$ with the concatenated path $\xi\gamma = (e_1, \dots, e_M)$, which obviously satisfies (\ref{eq:15}) since $\gamma$ and $\xi$
    do. The relation is symmetric since $c = \gamma \cdot b$ implies $b = \overline{\gamma} \cdot c$ with the reversed path $\overline{\gamma} =
    (\overline{e}_N, \dots, \overline{e}_1)$. It is again easy to see that (\ref{eq:15}) holds with $\overline{\gamma}$ and $c$ iff it holds with
    $\gamma$ and $b$. This concludes the proof of the lemma. 
  \end{proof}

  By construction $b \sim c$ is equivalent to $\ket{c} \in \mathcal{H}^{(b)}$ with $\ket{c} \neq 0$. Hence $\mathcal{H}^{(b)}$ is the linear hull of all basis vectors $\ket{c}$
  belonging to configurations in the equivalence class $[b]$ of $b \in \ConfP(\Gamma)$. This shows that for $b, c \in \ConfP(\Gamma)$ the Hilbert spaces
  $\mathcal{H}^{(b)}$, $\mathcal{H}^{(c)}$ are either identical or orthogonal (since $\ket{b}$, $b \in \ConfP(\Gamma)$ is a complete orthonormal system). 

  The next step is to show that the equivalence classes $[b]$ are finite sets and the Hilbert spaces $\mathcal{H}^{(b)}$ therefore finite dimensional. To
  this end recall from Sect. \ref{sec:description-problem} that there is a partial ordering $\leq$ on $V(\Gamma)$ which is uniquely determined by the
  condition: $t(e) < i(e)$ $\forall e \in E_+$. Since $V(\Gamma)$ is a finite set it contains elements which are minimal with respect to $\leq$,
  i.e. vertices $v \in V(\Gamma)$ such that $w \leq v$ implies $w=v$.  We use this fact to decompose $V(\Gamma)$ into a disjoint union of
  subsets $V_k(\Gamma)$. The latter are recursively defined as follows:
  \begin{enumerate}
  \item 
    $V_{-1}(\Gamma) = \emptyset$.
  \item 
    If $k\geq 0$ the set $V_k(\Gamma)$ consists of the minimal elements in $V(\Gamma) \setminus \bigcup_{j=-1}^{k-1} V_j(\Gamma)$.
  \item 
    The process terminates at $k=M$, with an $M \in \Bbb{N}$, when all of $V(\Gamma)$ is covered, i.e. $V(\Gamma) = \bigcup_{k=0}^M V_k(\Gamma)$ and
    all the $V_k(\Gamma)$ with $k \geq 0$ are non-empty. 
  \end{enumerate}
  The whole procedure is demonstrated in Figure \ref{fig:pseudoenergy}. The sets $V_k(\Gamma)$, $k=1,\dots,M$ are obviously disjoint and cover
  $V(\Gamma)$. Hence we can define functions $h_V: V(\Gamma) \rightarrow \Bbb{N}$ and $h_E : E(\Gamma) \rightarrow \Bbb{Z}$ by
  \begin{equation}
    h_V(v) = k \Leftrightarrow v \in V_k(\Gamma),\quad h_E(e) = h_V(i(e)) - h_V(t(e)).
  \end{equation}
  Both functions together leads to 
  \begin{equation}
    h: \Conf(\Gamma) \rightarrow \Bbb{Z}, b \mapsto h_V(b_0) + \sum_{e \in E_+} \underline{b}(e) h_E(e)
  \end{equation}

  \begin{figure}[t]
    \centering
    \includegraphics[width=0.9\textwidth]{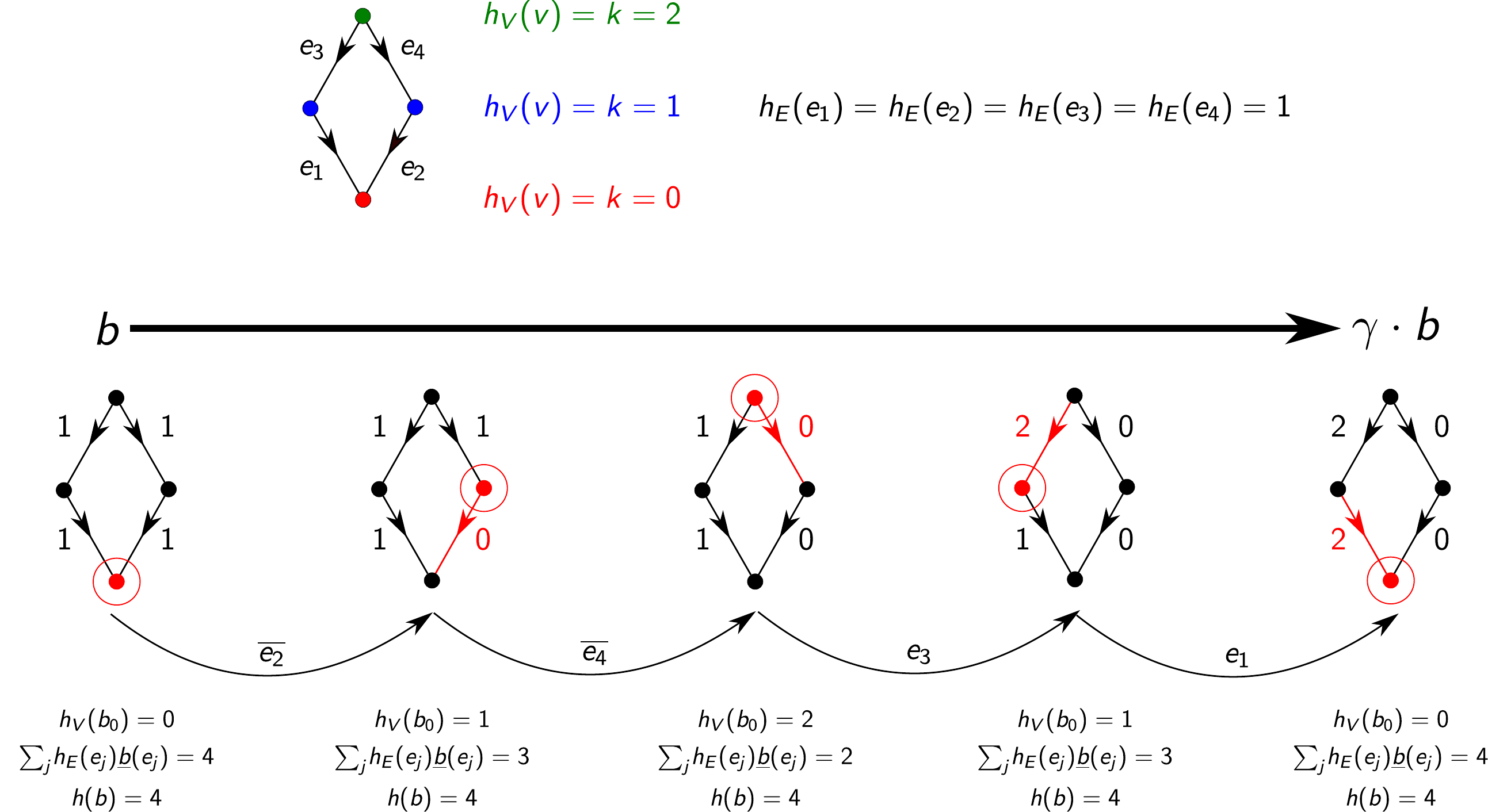}
    \caption{Definition of the pseudo-energy $h: \Conf(\Gamma) \rightarrow \Bbb{Z}$ and its invariance under the partial action $b \mapsto \gamma
      \cdot b$. \label{fig:pseudoenergy}} 
  \end{figure}

  \begin{lem} \label{lem:5}
    The function $h: \Conf(\Gamma) \rightarrow \Bbb{Z}$ just defined has the following properties
    \begin{enumerate}
    \item \label{item:2}
      If $b \in \Conf(\Gamma)$ and $e \in E(\Gamma)$ are chosen such that $e \cdot b \neq \Nil$ we have $h(e \cdot b) = h(b)$; i.e. $h$ is invariant under the
      (partial) action of $E(\Gamma)$ on $\Conf(\Gamma)$.
    \item \label{item:3}
      For each $n \in \Bbb{N}$ the level sets $\{b \in \ConfP(\Gamma)\,|\, h(b) = n\}$ are finite.
    \end{enumerate}
  \end{lem}

  \begin{proof}
    Let $e \cdot b \neq \Nil$ and assume without loss of generality that $e \in E_+$ (the other case can be handled similarly with a sign flip, and by
    using $h_E(\overline{e}) = -h_E(e)$). By definition we have  
    \begin{align}
      h(e \cdot b) &= h_V\bigl((e\cdot b)_0) + \sum_{f \in E_+} (\underline{e \cdot b})(f) h_E(f) \\
      &= h_V(t(e)) + \sum_{f \in E_+} \underline{b}(f) h_E(f) + h_E(e) \\
      &= h_V(t(e)) + \sum_{f \in E_+} \underline{b}(f) h_E(f) + h_V(i(e)) - h_V(t(e)) = h(b)
    \end{align}
    This proves the first statement. To show the second let us introduce the auxiliary function
    \begin{equation}
      \tilde{h}: \ConfP(\Gamma) \rightarrow \Bbb{N}_0,\quad b \mapsto \sum_{e \in E_+} \underline{b}(e).
    \end{equation}
    Obviously we have $0 \leq \tilde{h}(b) \leq h(b)$ for all $b \in \ConfP(\Gamma)$, and therefore
    \begin{equation}
      \{b \in \ConfP(\Gamma)\,|\, h(b) = n\} \subset  \{ b \in \ConfP(\Gamma)\,|\, \tilde{h}(b) \leq n\}.
    \end{equation}
    It is easy to see (e.g. by induction) that the set on the right hand side is finite (only non-negative $\underline{b}(e)$ allowed), which
    concludes the proof. 
  \end{proof}

  Now consider $b, c \in \ConfP(\Gamma)$ with $b \sim c$. Hence there is a path $\gamma$ satisfying $c = \gamma \cdot b$ and the condition in
  (\ref{eq:15}). Hence we can apply item \ref{item:2} of the last lemma recursively to show that $h(c) = h(\gamma \cdot b) = h(b)$. In other words,
  the function $h$ is constant on equivalence classes $[b]$. Finiteness of $[b]$ follows from item \ref{item:3} of Lemma \ref{lem:5}, and this shows that
  the $\mathcal{H}^{(b)}$ are finite dimensional.
  
  Therefore we have constructed a countable family of pairwise orthogonal, finite dimensional Hilbert spaces $\mathcal{H}^{(n)} \subset D_\Gamma \subset
  \mathcal{H}$, $n \in \Bbb{N}$ (where we have applied an arbitrary relabelling of the $\mathcal{H}^{(b)}$ in terms of positive integers -- this is
  obviously possible for any countable family). By construction each basis element $\ket{c}$ is contained in exactly one $\mathcal{H}^{(n)}$. Hence
  $\mathcal{H} = \bigoplus_{n=1}^\infty \mathcal{H}^{(n)}$ as stated. Since the $\mathcal{H}^{(n)}$ are by construction invariant subspaces for the path
  algebra $\PA(\Gamma)$, the second statement of the theorem is proved. The first folloiws immediately from the previous construction and the definition of $D_{\Gamma}$. 
\end{proof}

\section{Spectral analysis}
\label{sec:spectral-analysis}

From Eq. (\ref{eq:13}) it is easy to see that for any $A \in \PAext(\Gamma)$ the adjoint $A^*$ admits $D_\Gamma$ as an invariant domain as
well, and its restriction $A^+$ to $D_\Gamma$ is again an element of $\PAext(\Gamma)$. Hence we have defined a *-operation
$\PAext(\Gamma) \ni A \mapsto A^+ \in \PAext(\Gamma)$ which turns $\PAext$ into a *-algebra. To distinguish
selfadjoint operators on $\mathcal{H}$ from selfadjoint elements in $\PAext(\Gamma)$ we call the latter \emph{formally selfadjoint} (i.e. $A
= A^+$ holds). The difference between the two notions is, however, not too big, since any formally selfadjoint operator $A \in
\PAext(\Gamma)$ is -- as an operator $A$ on $\mathcal{H}$ with domain $D_\Gamma$ -- essentially selfadjoint, as the following proposition
shows. 

\begin{prop} \label{prop:3}
  A formally selfadjoint element $A$ of $\PAext(\Gamma)$ is (as an operator on $\mathcal{H}$) essentially selfadjoint on the domain
  $D_\Gamma$. The selfadjoint extension $\overline{A}$ of $A$ has a pure point spectrum
\end{prop}

\begin{proof}
  We use the subspaces $\mathcal{H}^{(n)}$ from Thm. \ref{thm:3} and define for all $N \in \Bbb{N}$: $\mathcal{K}_N = \bigcup_{k=1}^N \mathcal{H}_k$. The
  corresponding projections $\mathcal{H} \rightarrow \mathcal{K}_N$ are denoted by $Q_N$. The $\mathcal{K}_N$ are finite dimensional, invariant
  subspaces of $\PAext(\Gamma)$ and $D_\Gamma = \bigcup_{N \in \Bbb{N}} \mathcal{K}_N = D_\Gamma$. Hence, with $A_N = Q_N A Q_N$ we can
  find for each $\psi \in D_\Gamma$ an $N \in \Bbb{N}$ with $A \psi = A_N \psi$. Since the $A_N$ are finite rank (and therefore bounded) this implies
  \begin{equation}
    \sum_{k=0}^\infty \frac{\|A^k \psi\|}{k!} = \sum_{k=0}^\infty \frac{\|A^k_N \psi\|}{k!} < \infty.
  \end{equation}
  In other words all elements of $D_\Gamma$ are analytic vectors for $A$ and therefore $A$ is essentially selfadjoint on $D_\Gamma$ by Nelson's
  analytic vector theorem.

  To show the second statement note that each $A_N$ is selfadjoint on the finite dimensional Hilbert space $\mathcal{K}_N$. It therefore admits an
  orthonormal basis of eigenvectors $\phi_k$, $k=1,\dots,\dim \mathcal{H}_N$ satisfying $A_N \phi_k = \lambda_k \phi_k$ with eigenvalues $\lambda_k
  \in \Bbb{R}$. For $M > N$ we have $A_M \phi = A_N \phi$ $\forall \phi \in \mathcal{K}_N$. In other words the $\phi_k$ are eigenvectors of $A_M$, too
  (with the same eigenvalues), and we can extend the basis $\phi_k$, $k=1, \dots, \dim \mathcal{K}_N$ to an eigenbasis $\phi_k$, $k=1,\dots,\dim
  \mathcal{K}_M$ of $A_M$. Obviously the eigenvectors $\phi_k$ of an $A_N$ are eigenvectors of $A$ (note that $\psi_k \in D_\Gamma$ since
  $\mathcal{K}_N \subset D_\Gamma$). Hence, by increasing $N$ arbitrarily large we can construct a complete, orthonormal set of eigenvectors, which
  proves that $A$ has pure point spectrum. 
\end{proof}

Now recall the operators $H_X$ from Eq. (\ref{eq:4}). If $X$ is selfadjoint on $\mathcal{H}^A$ and diagonal in the canonical basis $\ket{v} \in
\mathcal{H}^A$, $v \in V(\Gamma)$, we get $H_X\in \PAext(\Gamma)$. Therefore $H_X$ is essentially selfadjoint on $D_\Gamma$, by
Prop. \ref{prop:3}. If $X$ is selfadjoint but not diagonal, we can still apply Prop. \ref{prop:3}, since $X$ is bounded and therefore relatively
bounded (with an arbitrary relative bound $0 < a < 1$) by any $H_Y$ with diagonal $Y$. Essential selfadjointness of $H_X$ on $D_\Gamma$ then follows
from the Kato-Rellich Theorem \cite[Thm. X.12]{RESI2}. However, the methods used in Prop. \ref{prop:3} and Thm. \ref{thm:3} does not tell us anything
about the eigenvalues. We do not even know (by Prop. \ref{prop:3}) whether $H_X$ (with non-diagonal $X$) has any discrete spectrum. To fill this gap
we will prove that all $H_X$ have compact resolvent (this is not true for all formally selfadjoint elements of $\PAext(\Gamma)$). To this
end we introduce on the domain $D_\Gamma$ the operators 
\begin{equation}
  H_0 = \sum_{e \in E_+} \omega_{C,e} \Bbb{1} \otimes a_e a_e^*,\quad H_I = \sum_{e \in E_+} \omega_{I,e}\left(\sigma_+^{(e)} \otimes a_e + \sigma_-^{(e)} \otimes a^*_e\right)
\end{equation} 
Both are elements of $\PAext(\Gamma)$ and therefore essentially selfadjoint on $D_\Gamma$. At least for $H_0$ this is well known since this
is (up to an additive constant) the Hamiltonian of an $|E_+|$-dimensional harmonic oscillator. We will write $\overline{H_0}$ for its (unique) selfadjoint 
extension and $D_0$ for the domain of the latter. For later use let us recall from Eq. (\ref{eq:13}) that the $A_e$ can be rewritten as $A_e =
\sigma_+^{(e)} \otimes a_e$ for $e \in E_+$ and $A_e = \sigma_-^{(\overline{e})} \otimes a^*_{\overline{e}}$ for $e \in E_-$. Therefore $H_I$ just
becomes the sum over all $A_e$ 
\begin{equation} \label{eq:19}
  H_I = \omega_{I,e} \sum_{e \in E(\Gamma)} A_e,
\end{equation}
with $\omega_{I,e} = \omega_{I,\overline{e}}$ for $e \in E_-(\Gamma)$. We will use this in the next lemma to prove a relative bound on $H_I$ in terms of $H_0$.

\begin{lem} \label{lem:1}
  There are constants $a, \eta > 0$, $a < 1$ such that
  \begin{equation}
    \|H_I \psi\| \leq a \|H_0 \psi\| + \eta \|\psi\| \quad \forall \psi \in D_\Gamma
  \end{equation}
\end{lem}

\begin{proof}
  From Eq. (\ref{eq:19}) and the definition of the $A_e$ in Eq. (\ref{eq:16}) we get
  \begin{equation}
    H_I \ket{b} = \sum_{e \in \mathcal{N}(b_0)} \omega_{I,e} \alpha(b,e) \ket{e\cdot b},
  \end{equation}
  where
  \begin{equation}
    \mathcal{N}(v) = \{ e \in E(\Gamma)\,|\, i(e) = v\}
  \end{equation}
  is the set of all edges starting at the vertex $v \in V(\Gamma)$. With $\psi \in D_\Gamma$:
  \begin{equation}
    \psi = \sum_{b \in \ConfP(\Gamma)} \psi_b \ket{b}\quad \text{with}\quad \psi_b = 0 \ \forall b \not\in \Delta
    \quad\text{and}\quad \Delta \subset
    \ConfP(\Gamma)\ \text{finite}
  \end{equation}
  we get
  \begin{equation}
    H_I \psi = \sum_{b \in \Delta} \sum_{e \in \mathcal{N}(b_0)}  \psi_b \, \omega_{I,e}\,\alpha(b,e) \ket{e\cdot b}. 
  \end{equation}
  With
  \begin{equation}
    \tilde{\Delta} = \{ e\cdot b\,|\, b \in \Delta,\ e \in \mathcal{N}(b_0) \}
  \end{equation}
  and for $c \in \tilde{\Delta}$
  \begin{equation}
    \Sigma_c = \{ (e,b) \, | \, b \in \Delta,\ e \in \mathcal{N}(b_0)\ \text{with}\ e\cdot b = c\}
  \end{equation}
  we can rewrite $H_I \psi$ further as 
  \begin{equation}
    H_I \psi = \sum_{ c \in \tilde{\Delta}} \left(\sum_{(e,b) \in \Sigma_c} \psi_{b} \, \omega_{I,e}\,\alpha(b,e)\right) \ket{c}.
  \end{equation}
  Hence
  \begin{equation} \label{eq:47}
    \|H_I\psi\|^2 \leq \lambda^2 \sum_{c \in \tilde{\Delta}} \left|  \sum_{(e,b) \in \Sigma_c} \psi_{b} \, \alpha(b,e) \right|^2
  \end{equation}
  with $\lambda = \max_{e \in E_+} |\omega_{I,e}|$. The cardinality $|\mathcal{N}(b_0)|$ of $\mathcal{N}(b_0)$ is
  bounded from above by $|E(\Gamma)|$ hence with $N = |\Delta|$ we get $|\Sigma_c| \leq N |E(\Gamma)|$, and therefore  
  \begin{equation}
    \left|  \sum_{(e,b) \in \Sigma_c} \psi_{b} \, \alpha(b,e) \right|^2 \leq N |E(\Gamma)| \sum_{(e,b) \in \Sigma_c} |\psi_{b}|^2 \alpha(b, e)^2.  
  \end{equation}
  Hence, with Eqs. (\ref{eq:20}) and (\ref{eq:47})
  this leads to
  \begin{equation}
    \|H_I\psi\|^2 \leq \lambda^2 N |E(\Gamma)| \sum_{b \in \Delta} \sum_{e \in \mathcal{N}(b_0)} |\psi_b|^2 \bigl(\underline{b}(e)+1\bigr) \leq
    \lambda^2 N |E(\Gamma)| \sum_{b \in \Delta} \sum_{e \in E_+} |\psi_b|^2 \bigl(\underline{b}(e)+1\bigr),
  \end{equation}
  where we have used the fact that only positive terms are added on the right hand side of the second inequality. We compare this to $\|H_0 \psi\|$:
  \begin{align}
    \|H_0 \psi\|^2 &= \sum_{b \in \Delta} \sum_{e \in E_+} \omega_{C,e}^2 |\psi_b|^2 \underline{b}(e)^2\\
    & \geq \mu^2 \sum_{b \in \Delta} \sum_{e \in E_+} |\psi_b|^2 \underline{b}(e)^2
  \end{align}
  with $\mu = \min_{e \in E_+} \omega_{C,e}$, which is strictly positive, since we have assumed $\omega_{C,e} > 0$ for
  all $e \in E_+$.

  The next step is to find constants $a \in (0,1)$, $\beta \geq 0$ such that $\lambda^2 N |E(\Gamma)| (n + 1) \leq \mu^2 (a n + \beta)^2$ holds for
  all $n \in \Bbb{N}_0$. To this end we choose $a \in (0,1)$ arbitrarily and introduce the polynomial
  \begin{equation}
    f(x) = -a^2x^2 + (d - 2a\beta) x + d - \beta^2,\quad x \in \Bbb{R},
  \end{equation}
  with the abbreviation $d=\lambda^2 \mu^{-2} N |E(\Gamma)|$. It has a global maximum at $\hat{x} = (d -2a\beta)/(2a^2)$ and
  an easy calculation shows that $f(\hat{x}) = 0$ holds if we choose $\beta = (d+4a^2)/(4a)$. Hence with this $\beta$ we have
  $f(x) \leq 0$ for all $x \in \Bbb{R}$ and in particular $f\bigl(\underline{b}(e)\bigr) \leq 0$ for all $e \in E_+$ and $b \in
  \ConfP(\Gamma)$. This shows that 
  \begin{equation}
    \lambda^2 |E(\Gamma)| (\underline{b}(e)+1) \leq (a \mu \underline{b}(e) + \mu \beta)^2 \quad \forall b \in \ConfP(\Gamma)\quad \forall e \in E_+
  \end{equation}
  holds with the chosen $a,\beta$ and therefore with $\eta = \mu \beta$
  \begin{multline}
    \| H_I\psi\|^2 \leq \lambda^2 |E(\Gamma)| \sum_{b \in \Delta} \sum_{e \in E_+} |\psi_b|^2 \bigl(\underline{b}(e)+1\bigr)
    \leq \sum_{b \in \Delta} \sum_{e \in E_+} |\psi_b|^2 (a \mu \underline{b}(e) + \eta)^2 \\ \leq \| a H_0 \psi + \eta \psi\|^2.
  \end{multline}
  Taking square roots at both sides and applying the triangle inequality to the right hand side leads to $\|H_I\psi\| \leq
  a \|H_0\psi\| + \eta \|\psi\|$ for all $\psi \in D_\Gamma$ as stated.  
\end{proof}

Now consider $X \in \mathcal{B}(\mathcal{H}^A)$ and $K_X = X \otimes \Bbb{1}^C + \lambda H_I$ with $\lambda \in
\Bbb{R}$. The latter is well defined and symmetric on $D_\Gamma$, hence it is closable with closure $\overline{K_X}$ and
domain $\mathrm{Dom}(\overline{K_X})$. We can use Lemma \ref{lem:1} and the Kato-Rellich Theorem \cite{RESI1} to proof
selfadjointness of the operators $H_X = H_0 + K_X$ introduced in Eq. (\ref{eq:4}).

\begin{lem} \label{lem:2}
  The operator $\overline{K_X}$ is relatively $\overline{H_0}$ bounded with relative bound $a < 1$, i.e. 
  $D_0 \subset \mathrm{Dom}(\overline{K_X})$ and
  \begin{equation}
    \| \overline{K_X} \psi\| \leq a \|\overline{H_0} \psi\| + (\eta+\|X\|) \|\psi\| \quad \forall \psi \in D_0
  \end{equation}
  holds with a constant $\eta>0$.
\end{lem}

\begin{proof}
  Consider $\psi \in D_0$. Since the graph $\mathrm{Gr}(\overline{H_0})$ of $\overline{H_0}$ satisfies
  $\mathrm{Gr}(\overline{H_0}) = \overline{\mathrm{Gr}(H_0)}$, we can find a sequence $\psi_n \in D_\Gamma$, $n \in
  \Bbb{N}$ converging to $\psi$ such that $\lim_{n \rightarrow \infty} H_0 \psi_n = H_0 \psi$. Hence $H_0 \psi_n$, and
  $\psi_n$, $n \in \Bbb{N}$ are Cauchy sequences and due to Lemma \ref{lem:1} $K_X \psi_n$ is a Cauchy
  sequence, too. Hence it converges to a $\phi \in \mathcal{H}$, and due to $\mathrm{Gr}(\overline{K_X}) =
  \overline{\mathrm{Gr}(K_X)}$ we can conclude that $\psi \in \mathrm{Dom}(\overline{K_X})$ with $K_X \psi =
  \phi$. Therefore we have $D_0 \subset \mathrm{Dom}(\overline{K_X})$ as stated. Using again Lemma \ref{lem:1} and
  monotonicity of limits we see that in addition
  \begin{equation}
    \| K_X \psi\| \leq a \|H_0 \psi\| + (\eta+\|X\|) \psi \quad \forall \psi \in D_0
  \end{equation}
  holds for $a, \eta > 0$ and $a < 1$, which concludes the proof.
\end{proof}

From now on we drop the bar over operators which are essentially selfadjoint on the domain $D_\Gamma$, i.e. whenever
necessary the corresponding selfadjoint extension is automatically understood.

\begin{prop}
  The operator $H_X = H_0 + X \otimes \Bbb{1}^C + H_I$ introduced in Eq. (\ref{eq:4}) is selfadjoint on the domain $D_0$
  and bounded from below. $D_\Gamma$ is a core of $H_X$.
\end{prop}

\begin{proof}
  This follows from Lemma \ref{lem:2} and the Kato-Rellich Theorem \cite{RESI1}.
\end{proof}

We have recovered the selfadjointness already proven in Prop. \ref{prop:3}, and in addition we have got a statement about the domain of
selfadjointness. This proves the first statement of Theorem \ref{thm:1}. The main tool towards the second is the following proposition
which states that $H_X$ has compact resolvent. In this context note that $H_0$ has (obviously) a pure point
spectrum with $\ket{b}$ as a complete basis of eigenvectors, i.e. 
\begin{equation} \label{eq:2}
  H_0 \ket{b} = \lambda_b \ket{b} = \left(\sum_{e\in E_+} \omega_{C,e} \, \underline{b}(e)\right) \ket{b}.  
\end{equation}
 By induction we can easily construct an enumeration of $\ConfP(\Gamma)$, i.e. a
bijective map $\Bbb{N}_0 \ni n \mapsto b_n \in \ConfP(\Gamma)$ such that the eigenvalues $\lambda_n =
\lambda_{b_n}$ satisfy $\lambda_n \leq \lambda_{n+1}$. Since $\lim_{n\rightarrow\infty} \lambda_n = \infty$
Thm. XIII.64 of \cite{RESI4} shows that $H_0$ has compact resolvent. In the following we will use Lemma \ref{lem:2} and
the min-max principle to show that $H_X$ shares this property. This leads to 

\begin{prop}
  The operator $H_X = H_0 + X \otimes \Bbb{1}^C + H_I$ has compact resolvent.
\end{prop}

\begin{proof}
  This is a slightly modified version of the proof of Theorem XIII.68 from \cite{RESI4}. We define for $n \in \Bbb{N}$: 
  \begin{equation}
    \mu_n(H_X) = \sup_{\psi_1,\dots,\psi_{n-1}} \inf \{ \langle \phi, H_X \phi\rangle \, | \, \phi \in
    \mathrm{Dom}(H_X),\ \| \phi\|=1,\ \phi \bot \psi_j\ \forall j=1,\dots,n-1\}.
  \end{equation}
  By the min max principle, $\mu_n(H_X)$ is either the $n^{\mathrm{th}}$ eigenvalue (counting multiplicities) or the
  infimum of the essential spectrum. In the latter case we have $\mu_n(H_X) = \mu_k(H_X)$ $\forall k \geq n$.

  Lemma \ref{lem:2} shows together with Thm X.18 of \cite{RESI2} that there are positive constants $a<1, \eta$ such
  that
  \begin{equation}
    |\langle \psi, K_X  \psi\rangle| \leq a \langle \psi, H_0 \psi\rangle + \eta \langle \psi, \psi\rangle
  \end{equation}
  holds for all $\psi \in D_0 = \mathrm{Dom}(H_0)$. Hence we get
  \begin{equation}
    \langle \psi, H_X \psi \rangle = \langle \psi, H_0 \psi\rangle + \langle\psi, K_X\psi\rangle 
    \geq (1-a) \langle \psi, H_0 \psi \rangle - \eta \langle \psi, \psi \rangle.
  \end{equation}
  This shows that $\mu_n(H_X) \geq \alpha \mu_n(H_0) - \eta$ with $\alpha = 1-a > 0$. Since due to Eq. (\ref{eq:2}) we have $\mu_n(H_0) = \lambda_n
  \rightarrow \infty$ for $n \rightarrow \infty$. We get $\lim_{n\rightarrow\infty} \mu_n(H_X) = \infty$. This excludes
  the possibility of a non-empty essential spectrum and therefore the $\mu_n(H_X)$ are the eigenvalues of $H_X$ and they
  converge to $\infty$. Hence the statement follows from Thm. XIII.64 of \cite{RESI4}. 
\end{proof}

Finally, we are ready to proof the recurrence statement of Thm. \ref{thm:1}.

\begin{prop}
  For all $t_- \in \Bbb{R}$, $t_- \leq 0$ and all strong neighborhoods $V$ of $\exp(i t_- H_X)$ in the unitary group
  $\mathrm{U}(\mathcal{H})$ of $\mathcal{H}$, there is a time $t_+ \in \Bbb{R}$, $t_+ > 0$ with $\exp(i t_+ H_X) \in V$.  
\end{prop}

\begin{proof}
  We can assume without loss of generality that $V$ is of the form
  \begin{equation}
    V = \{ U \in \mathrm{U}(\mathcal{H})\, | \, \|(U- \exp(i t_- H_X))\psi_1\| < \epsilon, \dots, \|(U- \exp(i t_-
    H_X))\psi_m\| \leq \epsilon \}, 
  \end{equation}
  (with $\epsilon >0$ and normalized vectors $\psi_1, \dots, \psi_m \in \mathcal{H}$), since neighborhoods of this type
  form a strong neighborhood base of $\exp(i t_- H_X)$. Now let us consider a complete basis $\phi_n$, $n \in \Bbb{N}$
  of eigenvectors of $H_X$ with eigenvalues $\lambda_n = \langle\phi_n, H_X \phi_n\rangle$. Furthermore $N \in \Bbb{N}$
  is chosen such that $\| (\Bbb{1} - P_N) \psi_j\| \leq \epsilon/3$ holds for the projection $P_N$ onto the span of
  $\phi_1, \dots, \phi_N$. Since $P_N\mathcal{H}$ is invariant under $H_X$, the latter defines a one-parameter group of
  unitaries $\tilde{U}_t = \exp(itH_X)|_{P_N\mathcal{H}}$ on $P_N\mathcal{H}$. On a finite dimensional Hilbert space
  (like $P_N\mathcal{H}$) recurrence in the required sense is always satisfied (since a finite number of eigenvalues can be approximated with
  arbitrary precision by rational numbers with common denominator).
  In other words, there is a $t_+ \in
  \Bbb{R}$ such that $\| \tilde{U}_{t_+} - \tilde{U}_{t_-}\| < \epsilon/3$ in the operator norm.  Now the statement follows from 
  \begin{align}
    \|\exp(i t_+ H_X) \psi_j &- \exp(i t_- H_X) \psi_j\|\\
    &\leq \|(\tilde{U}_{t_+} - \tilde{U}_{t_-}) P_N \psi_j\| + \|
    (\exp(it_+ H_X) - \exp(it_- H_X))(\Bbb{1} -
    P_N)\psi_j\| \\
    &\leq \frac{1}{3} + \| \exp(it_+ H_X) \| \| (\Bbb{1} - P_N)\psi_j\| + \| \exp(it_- H_X) \| \| (\Bbb{1} - P_N)\psi_j\| \leq \epsilon .
  \end{align} 
\end{proof}

\section{The dynamical group}
\label{sec:dynamical-group}

An important technical tool for the analysis of controllability in infinite dynamical systems is the dynamical group,
which is defined as follows:
 
\begin{defi}\label{def:1}
  Consider a Hilbert space $\mathcal{K}$ and the selfadjoint (unbounded) operators $H_1,$ $\dots,$ $H_N$. The smallest
  strongly closed (as a subset of  $\mathrm{U}(\mathcal{K})$ equipped with the strong topology) subgroup of
  $\mathrm{U}(\mathcal{K})$ containing the unitaries $\exp(itH_j)$, $j=1,...,N$ for all $t \in \Bbb{R}$ is called the
  \emph{dynamical group} generated by $H_1, \dots, H_N$ and denoted by $\mathcal{G}(H_1,\dots,H_N)$. 
\end{defi}
 
Using Eq. (\ref{eq:6}) we see that the strong closure of the set $\mathcal{M}$ in Thm. \ref{thm:2} coincides with the
dynamical group $\mathcal{G}\left(H_{x,y};\, (x,y) \in \Bbb{R}^{2|E_+|}\right)$, where $H_{x,y}$, $x,y \in
\Bbb{R}^{|E_+|}$ denotes the family of Hamiltonians from Eq. (\ref{eq:7}). Note in this context that -- although only
positive times are allowed in the definition  of $\mathcal{M}$ -- the \emph{strong closure} of $\mathcal{M}$ contains,
due to the recurrence property from Thm. \ref{thm:1}, the unitaries $\exp(i t H_{x,y})$ for negative $t$, as well. Hence
the control system (\ref{eq:5}) is strongly controllable iff
\begin{equation} \label{eq:18}
  \mathcal{G}\left(\Bbb{1}, H_{x,y};\, (x,y) \in \Bbb{R}^{E_+} \times \Bbb{R}^{E_+}\right) = \mathrm{U}(\mathcal{H})
\end{equation}
holds. Note that we have added the unit operator $\Bbb{1}$ as a control Hamiltonian, to handle the fact that the strong closure of $\mathcal{M}$ has
to contain a unitary $U$ on $\mathcal{H}$ only up to a constant phase factor $e^{i\alpha}$. 

The main task of this section are several lemmata which simplify calculations with dynamical groups significantly. The first is a simple application of
Trotter's product formula.

\begin{lem} \label{lem:3}
  Consider a separable Hilbert space $\mathcal{K}$ and two selfadjoint operators $H_1$, $H_2$ with domains
  $\mathrm{dom}(H_1), \mathrm{dom}(H_2) \subset \mathcal{K}$. If $H_1 + H_2$ is essentially selfadjoint on
  $\mathrm{dom}(H_1) \cap \mathrm{dom}(H_2)$ we have $\exp\left(it\overline{H_1+H_2}\right) \in
  \mathcal{G}(H_1,H_2)$.  
\end{lem}

\begin{proof}
  By the Trotter product formula we have
  \begin{equation}
    s-\lim_{n\rightarrow\infty} \exp(itH_1/n)\exp(itH_2/n) = \exp\left(it\overline{H_1+H_2}\right) =: U\quad \forall t \in \Bbb{R}
  \end{equation}
  Hence for each strong neighbourhood $\mathcal{N}$ of $U$ there is a $n \in \Bbb{N}$ such that the product $U_{1,n}
  U_{2,n}$ of $U_{1,n}= \exp(itH_1/n)$ and $U_{2,n}=\exp(itH_2/n)$ is in $\mathcal{N}$. However, the operator
  $U_{1,n}U_{2,n}$ is an element of $\mathcal{G}(H_1,H_2)$. Since the latter is strongly closed by assumption the
  statement follows. 
\end{proof}

We need a similar result concerning commutators of anti-selfadjoint operators $i\overline{X}$, $i\overline{Y}$. This is, however, more difficult to
achieve in general. For our purposes, however it is sufficient to consider the case where $X, Y$ are formally selfadjoint elements of the extended path
algebra. By Prop. \ref{prop:3} such operators are essentially selfadjoint on $D_\Gamma$ and therefore admit unique selfadjoint extensions
$\overline{X},\overline{Y}$. Moreover, by Thm. \ref{thm:3} the Hilbert space can be decomposed into a direct sum $\mathcal{H} = \bigoplus_{k=1}^\infty
\mathcal{H}_k$ of finite dimensional $\PAext(\Gamma)$-invariant subspaces. This implies that $X,Y$ are block diagonal of the form
\begin{equation}
  X\psi = \sum_{k=0}^\infty X_k \psi, \quad Y\psi = \sum_{k = 1}^\infty Y_k \psi \quad \forall \psi \in D_\Gamma
\end{equation}
with sequences $X_k$, $Y_k$, $k \in \Bbb{N}$ of selfadjoint (and bounded) operators on the finite dimensional spaces $\mathcal{H}_k$. This case was
studied in detail in \cite{KZSH} such that we can directly apply Thm. 2.1 of this previous work to get:

\begin{lem} \label{lem:6}
  Consider two formally selfadjoint elements $X, Y \in \PAext(\Gamma)$ with selfadjoint extensions $\overline{X}$, $\overline{Y}$.  The
  commutator $i[X,Y]$ is in $\PAext(\Gamma)$ again, and essentially selfadjoint on $D_\Gamma$. The unitaries $\exp(t \overline{[X,Y]})$ are
  elements of $\mathcal{G}(X,Y)$  
\end{lem}

\begin{proof}
  Follows from \cite[Thm. 2.1]{KZSH}
\end{proof}

The reduction to an increasing sequence of finite dimensional problems (as in the last lemma) is often useful. The next result provides a general
recipe for this type of approximations.

\begin{lem} \label{lem:8}
  Consider a separable Hilbert space $\mathcal{H}$ and an increasing sequence $(\mathcal{K}_n)_{n \in \Bbb{N}}$ of finite dimensional subspaces such
  that $\bigcup_n \mathcal{K}_n$ is dense in $\mathcal{K}$. For each $n \in \Bbb{N}$ define
  \begin{equation}
    \mathcal{U}_n = \{ U \in \mathrm{U}(\mathcal{K})\,|\, U \psi = \psi \ \forall \psi \in \mathcal{K}_n^{\perp}\},
  \end{equation}
  where $\mathcal{K}_n^{\perp}$ denotes the orthocomplement of $\mathcal{K}_n$ in $\mathcal{K}$ (i.e. $\mathcal{U}_n$ consists of all unitaries which
  act trivially on $\mathcal{K}_n^{\perp}$). The strong closure of $\bigcup_n \mathcal{U}_n$ coincides with $\mathrm{U}(\mathcal{K})$.
\end{lem}

\begin{proof}
  This follows immediately from Lemma 5.4 of \cite{KZSH}.
\end{proof}

In the next lemma we will use this result to prove a statement about dynamical groups on overlapping tensor products; cf. also \cite{HeiMa}

\begin{lem} \label{lem:10}
  Consider three Hilbert spaces $\mathcal{K}_j$, $j=1,\dots,3$ with $\dim \mathcal{K}_j \geq 2$ (can be infinite) and self adjoint operators $H_1,
  \dots, H_N$ on $\mathcal{K}_1 \otimes \mathcal{K}_2$ and $K_1, \dots, K_M$ on $\mathcal{K}_2 \otimes \mathcal{K}_3$. Assume that
  $\mathcal{G}(\Bbb{1}, H_1, \dots, H_N) = \mathrm{U}(\mathcal{K}_1 \otimes \mathcal{K}_2)$ and $\mathcal{G}(\Bbb{1}, K_1, \dots, K_M) =
  \mathrm{U}(\mathcal{K}_2 \otimes \mathcal{K}_3)$ holds. Then we have
  \begin{equation}
    \mathcal{G}(\Bbb{1}, H_1 \otimes \Bbb{1}, \dots, H_N \otimes \Bbb{1}, \Bbb{1} \otimes K_1, \dots, \Bbb{1} \otimes K_M) = \mathrm{U}(\mathcal{K}_1
    \otimes \mathcal{K}_2 \otimes \mathcal{K}_3)
  \end{equation}
\end{lem}

\begin{proof}
  By assumption we have
  \begin{equation}
      \mathcal{G}(\Bbb{1}, H_1 \otimes \Bbb{1}, \dots, H_N \otimes \Bbb{1}) = \mathrm{U}(\mathcal{K}_1 \otimes \mathcal{K}_2) \otimes \Bbb{1}, \quad \mathcal{G}(\Bbb{1}, \Bbb{1} \otimes 
      K_1, \dots, \Bbb{1} \otimes K_M) = \Bbb{1} \otimes \mathrm{U}(\mathcal{K}_2 \otimes \mathcal{K}_3)\,.
  \end{equation}
  Hence it is sufficient to show that the smallest, strongly closed subgroup of $\mathrm{U}(\mathcal{K}_1 \otimes \mathcal{K}_2
  \otimes \mathcal{K}_3)$ containing $\mathrm{U}(\mathcal{K}_1 \otimes \mathcal{K}_2) \otimes \Bbb{1}$ and $\Bbb{1} \otimes \mathrm{U}(\mathcal{K}_2
  \otimes \mathcal{K}_3)$ is $\mathrm{U}(\mathcal{K}_1 \otimes \mathcal{K}_2 \otimes \mathcal{K}_3)$ itself. If the Hilbert spaces are finite
  dimensional we can check equivalently, whether the real Liealgebras $\mathfrak{u}(\mathcal{K}_1 \otimes \mathcal{K}_2) \otimes \Bbb{1}$ and
  $\Bbb{1} \otimes \mathfrak{u}(\mathcal{K}_2 \otimes \mathcal{K}_3)$ together generate $\mathfrak{u}(\mathcal{K}_1 \otimes \mathcal{K}_2 \otimes
  \mathcal{K}_3)$, where $\mathfrak{u}(\,\cdot\,)$ denote the real Liealgebra of anti-selfadjoint operators on the given Hilbert space. For calculations
  of commutators it is easier to look at the complexification of $\mathfrak{u}(\,\cdot\,)$, i.e. the complex Liealgebra $\mathfrak{gl}(\,\cdot\,)$ of
  all operators, and therefore we have to check that the smallest Liesubalgebra of $\mathfrak{gl}(\mathcal{K}_1 \otimes \mathcal{K}_2 \otimes
  \mathcal{K}_3)$ containing $\mathfrak{gl}(\mathcal{K}_1 \otimes \mathcal{K}_2) \otimes \Bbb{1}$ and $\Bbb{1} \otimes \mathfrak{gl}(\mathcal{K}_2
  \otimes \mathcal{K}_3)$ is $\mathfrak{gl}(\mathcal{K}_1 \otimes \mathcal{K}_2 \otimes \mathcal{K}_3)$ itself. This is easily done by looking at
  commutators of the form
  \begin{multline} \label{eq:10}
    \left[\KB{e^{(1)}_{j_1} \otimes e^{(2)}_{j_2}}{e^{(1)}_{k_1} \otimes e^{(2)}_{k_2}} \otimes \Bbb{1}, \Bbb{1} \otimes \KB{e^{(2)}_{l_2} \otimes
        e^{(3)}_{l_3}}{e^{(2)}_{m_2} \otimes e^{(3)}_{m_3}}\right] = \\
    \delta_{k_2l_2} \KB{e_{j_1} \otimes e_{j_2} \otimes e_{l_3}}{e_{k_1} \otimes e_{m_2} \otimes e_{m_3}} - \delta_{m_2j_2} \KB{e_{j_1} \otimes e_{l_2}
      \otimes e_{l_3}}{e_{k_1} \otimes e_{k_2} \otimes e_{m_3}}
  \end{multline}
  where $e^{(j)}_k$, $j=1,\dots,3$, $k=1,\dots,\dim(\mathcal{K}_j)$ denote orthonormal bases of the Hilbert spaces $\mathcal{K}_j$. It is easy to see
  that we can generate all operators $\KB{e^{(1)}_{k_1}\otimes e^{(2)}_{k_2} \otimes ^{(3)}_{k_3}}{e^{(1)}_{m_1} \otimes e^{(2)}_{m_2} \otimes
    e^{(3)}_{m_3}}$ with commutators from (\ref{eq:10}) and with linear combinations of them we can get any operator on $\mathcal{K}_1 \otimes
  \mathcal{K}_2 \otimes \mathcal{K}_3$. Hence, by applying the reasoning from above the statement follows.

  Now assume all the Hilbert spaces are infinite dimensional (but separable). For each $j=1,2,3$ we choose a strictly increasing sequence $P^{(j)}_n$
  of finite dimensional, orthonormal projections on $\mathcal{K}_j$, converging strongly to $\Bbb{1}$ and define for $n_1, n_2, n_3 \in \Bbb{N}$
  \begin{equation}
    \mathcal{U}_{n_1n_2n_3} = \{ U \in \mathrm{U}(\mathcal{K}_1 \otimes \mathcal{K}_2 \otimes \mathcal{K}_3) \, | \, U \psi = \psi \ \forall \psi \
    \text{with}\ P_{n_1} \otimes P_{n_2} \otimes P_{n_3} \psi = 0 \} .
  \end{equation}
  In other words, $\mathcal{U}_{n_1n_2n_3}$ consists of all unitaries on $\mathcal{K}_1 \otimes \mathcal{K}_2 \otimes \mathcal{K}_3$ acting trivially
  on the orthocomplement of $ \bigotimes_{j=1}^3 [P_{n_j} \mathcal{K}_j]$. Similarly we define $\mathcal{U}_{n_1n_2} \subset \mathrm{U}(\mathcal{K}_1
  \otimes \mathcal{K}_2)$ and $\mathcal{U}_{n_2n_3} \subset \mathrm{U}(\mathcal{K}_2 \otimes \mathcal{K}_3)$. All these groups are finite dimensional
  Lie groups (since the projections $P_j$ are finite dimensional) and therefore we can apply the result from the last paragraph to conclude that
  $\mathcal{U}_{n_1n_2n_3}$ is the smallest Liegroup (and therefore the smallest (strongly) closed group as well) which contains $\mathcal{U}_{n_1n_2}
  \otimes \Bbb{1}$ and $\Bbb{1} \otimes \mathcal{U}_{n_2n_3}$. Since the subspaces $\bigotimes_{j=1}^3 [P_{n_j} \mathcal{K}_j]$ exhaust in the limit
  the whole Hilbert space we can apply Lemma \ref{lem:8} and the statement follows. If only one or two of the Hilbert spaces are infinite dimensional,
  we can proceed in the same way by exhausting only one (or two) Hilbert spaces with finite dimensional subspaces. This concludes the proof.
\end{proof}

\section{Full controllability}
\label{sec:full-controlability}

We are now prepared to provide the full proof of Thm. \ref{thm:2}. To this end we will use Eq. (\ref{eq:18}) and the discussion in
Sect. \ref{sec:dynamical-group}. In the first step we will simplify the set of generators of the dynamical group on the left
hand side of Eq. (\ref{eq:18}).

\begin{lem}\label{lem:4}
  The control problem (\ref{eq:5}) is strongly controllable (in the sense of Theorem \ref{thm:2}) iff $\mathcal{G}(\Bbb{1}, H_D,
  X^{(e)}, Y^{(e)}; e \in E_+) = \mathrm{U}(\mathcal{H})$ holds. 
\end{lem}

\begin{proof}
  By Thm.\ \ref{thm:1} all the operators $H_{x,y}$ are essentially selfadjoint on the domain $D_\Gamma$. The same is true for $H_D$, $X^{(e)}$ and
  $Y^{(e)}$ (the latter two are even bounded). Since all the $H_{x,y}$ are linear combinations of $H_D$, $X^{(e)}$, $Y^{(e)}$ and vice versa the
  statement follows from Lemma \ref{lem:3} and Eq. (\ref{eq:18}).
\end{proof}

The next step concentrates on the group generated by the bounded operators $X^{(e)}$, $Y^{(e)}$. Here we can easily use Lie-algebraic methods. 

\begin{lem} \label{lem:7}
  $\mathcal{G}(X^{(e)},Y^{(e)};e\in E_+) = \mathrm{SU}(\mathcal{H}^A) \otimes \Bbb{1} \subset \mathrm{U}(\mathcal{H})$.
\end{lem}

\begin{proof}
  The operators $X^{(e)}, Y^{(e)}$ are of the form $X^{(e)} = \sigma_1^{(e)} \otimes \Bbb{1}^C$, $Y^{(e)} =
  \sigma_3^{(e)} \otimes \Bbb{1}^C$, with $\sigma_{1,3}^{(e)}$ acting as $\sigma_{1/3}$ on $\mathcal{H}^A_e =
  \SP\{i(e), t(e)\} \cong \Bbb{C}^2$. Hence it is sufficient to show that 
  \begin{equation} \label{eq:14}
    \mathcal{G}(\sigma^{(e)}_1, \sigma^{(e)}_3; e \in E_+) = \mathrm{SU}(\mathcal{H}^A)
  \end{equation}
 holds. However, the Hilbert space $\mathcal{H}^A \cong \Bbb{C}^{|V(\Gamma)|}$
  is finite dimensional, such that  $\mathcal{G}(\sigma^{(e)}_1, \sigma^{(e)}_3;e \in E_+)$ just becomes the smallest
  Lie subgroup of $\mathrm{SU}(\mathcal{H}^A)$ containing all the operators $\exp(it\sigma^{(e)}_{1,3})$ for some $t \in
  \Bbb{R}$. To show that (\ref{eq:14}) holds it is therefore sufficient to prove that the complex Liealgebra $\langle
  \sigma^{(e)}_{1,3}; e \in E_+\rangle_{\mathrm{Lie},\Bbb{C}}$ coincides with $\mathfrak{sl}(\mathcal{H}^A)$, i.e. the
  trace-free matrices on $\mathcal{H}^A$. To this end we will proceed as follows:

  Firstly we can assume that $\Gamma$ is a tree graph, since we can replace a general $\Gamma$ with a spanning tree $\Sigma$ which
  satisfies 
  \begin{equation}
    \langle \sigma^{(e)}_{1,3}; e \in E_+(\Sigma)\rangle_{\mathrm{Lie},\Bbb{C}} \subset \langle \sigma^{(e)}_{1,3}; e \in
    E_+(\Gamma)\rangle_{\mathrm{Lie},\Bbb{C}}.
  \end{equation}
  Secondly, the statement is obviously true for the fully connected graph with two vertices, since the corresponding Hilbert space is two-dimensional
  and the complex Lie-algebra generated by $\sigma_1, \sigma_3$ coincides with $\mathrm{sl}(2,\Bbb{C})$. 

  Finally, the general case follows by induction. Hence, assume that we have proven the result for a tree graph $\Gamma$. In addition consider another
  tree $\Sigma$ with exactly one vertex $v_0$ (and one geometric edge) more than $\Gamma$, i.e. $V(\Sigma) = V(\Gamma) \cup \{ v_0 \}$. The one edge
  $e_0 \in E_+(\Sigma)$ we have to add is of the form $i(e_0) = v_1$, $t(e_0) = v_0$ with $v_1 \in V(\Gamma)$. The operators $\sigma^{(e_0)}_{1,3}$ 
  generate all linear combination of the operators $\KB{v_0}{v_1}$, $\KB{v_1}{v_0}$ and $\kb{v_0} - \kb{v_1}$. The space $\langle
  \sigma^{(e)}_{1,3}; e \in E_+\rangle_{\mathrm{Lie},\Bbb{C}} = \mathfrak{sl}(|V(\Gamma)|,\Bbb{C})$ is on the other hand generated (as a vector space)
  by operators $\KB{v}{w}$, $\KB{w}{v}$, $\kb{v}-\kb{w}$, $v, w \in V(\Gamma)$. Using commutators like $[\KB{v}{v_1},\KB{v_1}{v_0}]$ we can produce
  all operators $\KB{x}{y}$, $\KB{x}{y}$, $\kb{x}-\kb{y}$, $x, y \in V(\Sigma)$. But this set spans (as a vector space) the Lie-algebra
  $\mathfrak{sl}(|V(\Sigma)|,\Bbb{C})$, which concludes the proof. 
\end{proof}

To get a simpler set of generators we split the ``interaction Hamiltonian'' $H_I$ up into its summands
\begin{equation}
  H_I = \sum_{e \in E_+} \omega_{I,e} Z^{(e)},\quad Z^{(e)} = \sigma_+^{(e)} \otimes a_e + \sigma_-^{(e)} \otimes a^*_e = A(e) + A(\overline{e}) 
\end{equation}
and use Lemma \ref{lem:6} to reexpress the $Z^{(e)}$ in terms of double commutators.

\begin{lem} \label{lem:11}
  $\Gamma(X^{(e)}, Y^{(e)}, Z^{(e)}; e \in E_+) \subset \mathcal{G}(H_D, X^{(e)}, Y^{(e)}; e \in E_+)$
\end{lem}

\begin{proof}
  According to Lemma \ref{lem:7} we have
  \begin{equation}
    \mathcal{G}(H_D, X^{(e)}, Y^{(e)}; e \in E_+) = \mathcal{G}(H_D, K \otimes \Bbb{1}^C; K \in \mathfrak{su}(\mathcal{H}^A)).
  \end{equation}
  Furthermore, the tensor product $K \otimes \Bbb{1}^C$ is an element of the extended path-algebra $\PAext(\Gamma)$, provided $i K \in
  \mathfrak{su}(\mathcal{H}^A)$ is diagonal in the canonical basis $\ket{v}$, $v \in V(\Gamma)$. Since $H_D \in \PAext(\Gamma)$ holds as
  well, the one-parameter subgroups $\exp(tQ)$ generated by real linear combinations $Q$ of (repeated) commutators of $iH_D$ and diagonal $iK \otimes
  \Bbb{1}^C$ are subgroups of $\mathcal{G}(H_D, K \otimes \Bbb{1}^C; K \in \mathfrak{su}(\mathcal{H}^A))$ and therefore also of $\mathcal{G}(H_D,
  X^{(e)}, Y^{(e)}; e \in E_+)$. 

  Now consider $K_v = \kb{v} - |V(\Gamma)|^{-1} \Bbb{1}^A$. Obviously $K_v$ is trace-free, selfadjoint and diagonal. Hence it satisfies the requirements
  of the last paragraph. Commutators with $i K_v \otimes \Bbb{1}^C$ equals commutators with $i \kb{v} \otimes \Bbb{1}^C$. Therefore we get for the
  operators $A(e)$ from Eq. (\ref{eq:16}) with $e \in E_+$ and $\psi \in D_\Gamma$
  \begin{equation}
    [ K_v \otimes \Bbb{1}^C, A(e)] \psi = [ \kb{v} \otimes \Bbb{1}^C, \KB{t(e)}{i(e)} \otimes a_e] \psi = \delta_{v,t(e)} A_e \psi -
    \delta_{v,i(e)} A_e \psi,   
  \end{equation}
  where $\delta_{v,w}=1$ for $v, w \in V(\Gamma)$ iff $v=w$ holds. For $e \in E_-$ we get similarly:
  \begin{equation}
    [ K_v \otimes \Bbb{1}^C, A(e)] \psi = [ \kb{v} \otimes \Bbb{1}^C, \KB{t(e)}{i(e)} \otimes a_{\overline{e}}^*] \psi = \delta_{v,t(e)} A_e \psi -
    \delta_{v,i(e)} A_e \psi. 
  \end{equation}
  Now recall from Eq. (\ref{eq:19}) that we can write $H_I$ as a sum of all $A_e$. In other words we get for $H_D$ 
  \begin{equation}
    H_D = \sum_{f \in E_+} \left[\omega_{A,f} \sigma_3^{(f)} \otimes \Bbb{1}^C  + \omega_{C,f} \Bbb{1}^A \otimes a^*a\right] + \sum_{f \in E(\Gamma)}
    \omega_{I,f} A_f  
  \end{equation}
  and therefore with $\psi \in D_\Gamma$
  \begin{align}
    [K_v \otimes \Bbb{1}^C, H_D] \psi &= [K_v \otimes \Bbb{1}^C, H_I] \psi =  \omega_{I,f} \sum_{t(f)=v} A_f \psi - \omega_{I,f} \sum_{i(f)=v} A_f \psi
    \\ &= \sum_{t(f)=v} \omega_{I,f} (A_f - A_{\overline{f}})\psi .
  \end{align}
  Now consider $e \in E_+$ with $i(e)=v$ and $t(e) = w$. Another commutator leads to
  \begin{align}
    [K_w &\otimes \Bbb{1}^C, [K_v \otimes \Bbb{1}^C, H_D]] \psi  \\ 
    &=  \sum_{t(f)=v} \omega_{I,f}\bigl([K_W, A_f] \psi - [K_W, A_{\overline{f}}] \psi\bigr)  \\
    &= \sum_{t(f) = v} \omega_{I,f} \left(\delta_{w,t(f)} A_f \psi - \delta_{w,i(f)} A_f \psi\right) - \sum_{t(f) = v} \omega_{I,f}
    \left(\delta_{w,t(\overline{f})} A_{\overline{f}} \psi - \delta_{w,i(\overline{f})} A_{\overline{f}} \psi\right) \\
    &= - \omega_{I,e} (A_e + A_{\overline{e}}) = - \omega_{I,e} Z^{(e)}.
  \end{align}
  With the reasoning from the last paragraph, the statement follows.
\end{proof}

The last result shows that it is sufficient to show that $\Gamma(\Bbb{1}, X^{(e)}, Y^{(e)}, Z^{(e)}; e \in E_+) = \mathrm{U}(\mathcal{H})$  holds. We
will do this by induction on the set of edges. The first step is to look at the dynamical group which is generated by the operators $X^{(e)}, Y^{(e)},
Z^{(e)}$ for one given edge $e$. To formulate the result we need some additional notation. This includes in particular the Hilbert spaces
$\mathcal{H}_e$ and $\hat{\mathcal{H}}_e$ given by
\begin{equation} \label{eq:11}
  \mathcal{H}_e = \mathcal{H}^A \otimes \mathcal{H}^C_e,\quad \hat{\mathcal{H}}_e = \bigotimes_{f \neq e} \mathcal{H}^C_f,\quad \mathcal{H} \cong
  \mathcal{H}_e \otimes \hat{\mathcal{H}}_e.
\end{equation}
The Hilbert space $\mathcal{H}_e$ contains the subspace
\begin{equation}
  \mathcal{K}_e = \mathcal{H}^A_e \otimes \mathcal{H}^C_e \subset \mathcal{H}_e ,\quad \mathcal{H}^A_e = \SP \{\ket{i(e)}, \ket{t(e)}\} \subset
  \mathcal{H}^A. 
\end{equation}
A unitary $U$ on $\mathcal{K}_e$ can be extended to $\mathcal{H}_e$ by
\begin{equation}
  \tilde{U} \psi = U \psi \quad \text{if $\psi \in \mathcal{K}_e$}\quad \tilde{U}\psi = \psi \quad\text{if $\psi \in \mathcal{K}_e^{\perp}$},
\end{equation}
where $\mathcal{K}_e^{\perp}$ denotes the orthocomplement of $\mathcal{K}_e$ in $\mathcal{H}_e$. In a second step we can extend $\tilde{U}$ to
$\mathcal{H}$ by adjoining a unit operator on $\hat{\mathcal{H}}_e$:
\begin{equation} \label{eq:9}
  \Lambda_e(U) = \tilde{U} \otimes \Bbb{1} = \tilde{U} \otimes \bigotimes_{f\neq e} \Bbb{1}_f 
\end{equation}
with the unit operators $\Bbb{1}_e$ on $\mathcal{H}^C_e$. Using this notations we can reformulate Thm. 3.2. from \cite{KZSH} as follows:

\begin{lem} \label{lem:9}
  For a fixed edge $e \in E_+$ we have
  \begin{equation}
    \mathcal{G}(\Bbb{1}, X^{(e)}, Y^{(e)}, Z^{(e)}) = \{ \Lambda_e(U)\,|\, U \in \mathrm{U}(\mathcal{K}_e)\}.
  \end{equation}
\end{lem}

\begin{proof}
  This follows immediately from Thm. 3.2. of \cite{KZSH}.
\end{proof}

Now we can combine this result with Lemma \ref{lem:7} to include more generators

\begin{lem} \label{lem:12}
  For a fixed edge $e \in E_+$ we have
  \begin{equation}
    \mathcal{G}(\Bbb{1}, Z^{(e)}, X^{(f)}, Y^{(f)}; f \in E_+) = \{ U \otimes \Bbb{1} \, | \, U \in \mathrm{U}(\mathcal{H}_e)\}
  \end{equation}
  where $\Bbb{1}$ denotes the unit operator on $\hat{\mathcal{H}}_e$; cf. Eq. (\ref{eq:9}).
\end{lem}

\begin{proof}
  According to Lemmata \ref{lem:7} and \ref{lem:9} we have to show that the smallest, strongly closed subgroup of $\mathrm{U}(\mathcal{H}_e)  \otimes
  \Bbb{1}$ containing $\mathrm{U}(\mathcal{H}^A) \otimes \Bbb{1}$ and $\mathcal{U} = \{ \Lambda_e(U)\,|\, U \in \mathrm{U}(\mathcal{K}_e)\}$ is
  $\mathrm{U}(\mathcal{H}_e)  \otimes \Bbb{1}$ itself. We will do this with the same strategy as in the proof of Lemma \ref{lem:10}: We break
  the task up into a series of finite dimensional problems and then we apply Lemma \ref{lem:8}.

  Hence consider for each $N \in \Bbb{N}$ the projections $P_N = \sum_{n=0}^N \kb{n;e}$ from $\mathcal{H}^C_e=\mathrm{L}^2(\Bbb{R})$ onto
  $P_N \mathcal{H}^C_e = \SP\{\ket{n;e}\,|\, n < N\}$. Here $\ket{n,e}$ denotes the number basis (i.e. Hermite functions); cf.\ the notations
  introduced in Sect. \ref{sec:description-problem}. Similarly we define 
  \begin{equation}
    Q_N = \sum_{v =i(e),t(e)} \sum_{n=0}^N \kb{v} \otimes \kb{n;e}
  \end{equation}
  which is the projection onto $\mathcal{H}^A_e \otimes P_N\mathcal{H}^C_e \subset \mathcal{K}_e$. Now we can define
  \begin{equation}
    \mathcal{U}_N = \{ \Lambda_e(U)\,|\, U \in \mathrm{U}(\mathcal{K}_e)\ \text{with:}\ Q_N \psi = 0 \Rightarrow U \psi = \psi \ \forall \psi
    \in \mathcal{K}_e\}. 
  \end{equation}
  The $\mathcal{U}_N$ are (as well as $\mathrm{U}(\mathcal{H}^A)$) finite dimensional, hence we can look at the complex Liealgebras
  \begin{equation}
    \mathfrak{gl}(\mathcal{H}^A) \otimes \Bbb{1} \subset \mathfrak{gl}(\mathcal{H}^A \otimes P_N
    \mathcal{H}^C_e)\quad \text{and}\quad \mathfrak{gl}(Q_N \mathcal{K}_e) \subset \mathfrak{gl}(\mathcal{H}^A \otimes P_N
    \mathcal{H}^C_e) 
  \end{equation}
  and show that both together generate $\mathfrak{gl}(\mathcal{H}^A \otimes P_N \mathcal{H}^C_e)$; cf. the proof of Lemma \ref{lem:10}. With the bases
  $\ket{v} \in \mathcal{H}^A$, $v \in V(\Gamma)$ and $\ket{n;e} \in P_N \mathcal{H}^C_e$, $n=0,\dots,N$ we have to look at commutators
  \begin{equation}
    \bigl[ \KB{v}{w} \otimes \Bbb{1}, \KB{x}{y} \otimes \KB{n;e}{m;e} \bigr] = \delta_{wx} \KB{v}{y} \otimes \KB{n;e}{m;e} + \delta_{vy} \KB{x}{w} \otimes
    \KB{n;e}{m;e},
  \end{equation}
  where $x,y \in \{i(e),t(e)\}$. It is easy to see that we can express all operators $\KB{v}{w}\otimes\KB{n;e}{m;e}$ in terms of such commutators and all 
  elements in $\mathfrak{gl}(\mathcal{H}^A \otimes P_N \mathcal{H}^C_e)$ in terms of linear combinations of them. Hence, with the reasoning from Lemma
  \ref{lem:10} we see that $\mathrm{U}(\mathcal{H}^A) \otimes \Bbb{1}$ and $\mathcal{U}_N$ generate $\mathrm{U}(\mathcal{H}^A \otimes P_N
  \mathcal{H}_e)$. Since the $P_N$ form a strictly increasing sequence of orthonormal projections converging strongly to $\Bbb{1}$, the statement
  follows from Lemma \ref{lem:8}.
\end{proof}

This lemma finally allows us to analyze the structure of the dynamical group $\Gamma(X^{(e)}, Y^{(e)}, Z^{(e)}; e \in E_+)$:

\begin{prop} \label{prop:1}
  $\Gamma(X^{(e)}, Y^{(e)}, Z^{(e)}; e \in E_+) = \mathrm{U}(\mathcal{H})$
\end{prop}

\begin{proof}
  Consider a nonempty set $\Delta \subset E_+$ and define
  \begin{equation}
    \mathcal{H}_\Delta = \mathcal{H}^A \otimes \bigotimes_{e \in \Delta} \mathcal{H}^C_e,\quad \hat{\mathcal{H}}_\Delta = \bigotimes_{e \not\in \Delta}
    \mathcal{H}^C_e. 
  \end{equation}
  If $\Delta = E_+$ we have $\mathcal{H}_\Delta = \mathcal{H}$ and for notational consistency we define in addition $\hat{\mathcal{H}}_{E_+}
  =\Bbb{C}$. Note that this is a natural extension of the notation from Eq. (\ref{eq:11}) since we have $\mathcal{H}_e = \mathcal{H}_{\{e\}}$. Now 
  assume that $\Delta$ satisfies
  \begin{equation} \label{eq:12}
    \mathcal{G}(\Bbb{1}, X^{(e)},Y^{(e)}, Z^{(e)}; e \in \Delta) = \mathrm{U}(\mathcal{H}_\Delta) \otimes \Bbb{1}.
  \end{equation}
  If this holds for $\Delta = E_+$ the proposition is proved. Hence assume $\Delta \neq E_+$ with $f \not\in
  \Delta$. According to Lemma \ref{lem:12} we have $\mathcal{G}(\Bbb{1}, X^{(f)}, Y^{(f)}, Z^{(f)}) = \mathrm{U}(\mathcal{H}_{\{f\}}) \otimes
  \Bbb{1}$. Hence by Eq. (\ref{eq:12}) the groups $\mathcal{G}(\Bbb{1}, X^{(e)},Y^{(e)}, Z^{(e)}; e \in \Delta)$ and $\mathcal{G}(\Bbb{1}, X^{(f)},
  Y^{(f)}, Z^{(f)})$ satisfy the assumptions of Lemma \ref{lem:10}, i.e. they ``overlap'' on the tensor factor $\mathcal{H}^A$. Applying Lemma
  \ref{lem:10} we therefore find that Eq. (\ref{eq:12}) holds with $\Delta$ replaced by $\Delta \cup \{ f \}$. Now we use Lemma \ref{lem:12} again to
  see that $\Delta = \{e\}$ with a fixed but arbitrary $e \in E_+$ satisfies Eq. (\ref{eq:12}), and apply the previous induction argument until $\Delta
  = E_+$ is reached. This concludes the proof. 
\end{proof}

With this proposition at hand Theorem \ref{thm:2} follows from Lemma \ref{lem:4} and Lemma \ref{lem:11}.

\section{Example 1:  Two levels}
\label{sec:example-1:-two}

Let us consider now the fully connected graph $\Gamma=K_2$ with two vertices (and one edge) representing a two-level atom interacting with one
mode. The Hilbert space of the systems becomes $\mathcal{H}=\Bbb{C}^2 \otimes \mathrm{L}(\Bbb{R})$ and the operators $X^{(e)}, Y^{(e)}, Z^{(e)}$ are
(dropping the now redundant superscript $e$): 
\begin{equation}
  X = \sigma_1 \otimes \Bbb{1},\quad Y = \sigma_3 \otimes \Bbb{1},\quad Z = \sigma_+ \otimes a + \sigma_- \otimes a^*,
\end{equation}
which leads to the drift Hamiltonian
\begin{equation}
  H_D = \omega_A \sigma_3 \otimes \Bbb{1} + \omega_C \Bbb{1} \otimes a^*a + \omega_I \left(\sigma_+ \otimes a + \sigma_- \otimes a^*\right), 
\end{equation}
i.e. drift is described by the Jaynes-Cummings Hamiltonian. Theorem \ref{thm:2} now tells us that the control problem
\begin{equation}
  i \frac{d}{dt} U_{u,v}(0,t) \psi = H_D U_{u,v}(0,t) \psi + u(t) X U_{u,v}(0,t) \psi + v(t) Y U_{u,v}(0,t) \psi
\end{equation}
with piecewise constant control functions is strongly controllable. This is closely related to a result from \cite{KZSH} where control without drift
is considered. In other words 
\begin{equation} \label{eq:21}
  i \frac{d}{dt} U_{u,v,w}(0,t) \psi = u(t) X U_{u,v,w}(0,t) \psi + v(t) Y U_{u,v,w}(0,t) + w(t) Z U_{u,v,w}(0,t)\psi
\end{equation}
is strongly controllable, too (again with piecewise constant $u,v,w$. This is equivalent to the statement $\mathcal{G}(\Bbb{1},X,Y,Z) =
\mathrm{U}(\mathcal{H})$ which we have already used within the proof of Theorem \ref{thm:2} (cf. Lemma \ref{lem:9}). Both systems are closely related,
since we can generate the generator $Z$ by linear combinations and repeated commutators of $Y$ and $H_D$; cf. Lemma \ref{lem:11}. Therefore we will
concentrate for the rest of this section on (\ref{eq:21}).

To get more insight into the way how a concrete control task has to be done, we will look at the problem of transforming an arbitrary pure state
$\kb{\psi_i}$ into an arbitrary final state $\kb{\psi_f}$, by appropriately choosing the control functions $u, v, w$. Here we have chosen $\psi_i,
\psi_f \in \mathcal{H}$ with $\|\psi_i\|=\|\psi_f\|=1$. Note that this is obviously possible since we can approximate (strongly) an arbitrary
unitary; i.e. our system is not only strongly controllable but also (approximately) pure state controllable (cf. \cite{KZSH}). Up to a large degree we
only have to review the work done in \cite{KZSH}. Therefore another task of this section is to show how this previous work fits into our current
analysis.

As a first step let us have a look at the canonical basis $\ket{b}$, $b \in \ConfP(\Gamma)$. For $\Gamma = K_2$ it takes the simple form $\ket{\nu} \otimes
\ket{n} \in \mathcal{H}$ with $\ket{\nu} \in \Bbb{C}^2$, $\nu=0,1$ the canonical basis and $\ket{n} \in \mathrm{L}^2(\Bbb{R})$, $n \in \Bbb{N}$ the
Hermite functions. We relabel the basis vectors according to  
\begin{equation}
  \ket{\mu;\nu} = \ket{\nu} \otimes \ket{\mu - \nu},\quad \mu \in \Bbb{N},\ \nu=0,1;\quad \ket{0,0} = \ket{0}\otimes \ket{0}.
\end{equation}
This relabelling is particular useful if we look at the action of the operators $A_e, A_e^+$, $e \in E_+(K_2)$ from Eq. (\ref{eq:13}). By dropping
again the redundant label $e$, we get
\begin{alignat}{2}
  A \ket{\mu;0} &= \sqrt{\mu} \ket{\mu;1} \quad &A \ket{\mu;1} &= 0\\
  A^+ \ket{\mu;0} &= 0 \quad &A^+ \ket{\mu;1} &= \sqrt{\mu} \ket{\mu;0}.
\end{alignat}
Since $\PAext(K_2)$ is generated by $A, A^+$ and all operators diagonal in the basis $\ket{\mu;\nu}$ we immediately see that the subspaces
$\mathcal{H}^{(\mu)} \subset \mathcal{H}$ given by 
\begin{equation} \label{eq:22}
  \mathcal{H}^{(\mu)} = \SP \{ \ket{\mu;0}, \ket{\mu; 1}\}\ \text{if $\mu > 0$}\quad \text{and}\quad \mathcal{H}^{(0)} = \Bbb{C} \ket{0;0},
\end{equation}
are invariant for $\PAext(K_2)$. Obviously the infinite direct sum of the $\mathcal{H}^{(\mu)}$ exhaust the whole Hilbert space
$\mathcal{H}$, i.e.
\begin{equation} \label{eq:27}
  \mathcal{H} = \bigoplus_{\mu=0}^\infty \mathcal{H}^{(\mu)},
\end{equation}
with convergence in norm. Hence, we have recovered the direct sum decomposition from Thm. \ref{thm:3}. 

Now the natural question is, whether $\PAext(K_2)$ contains all operators which are block diagonal in the decomposition (\ref{eq:22}). To answer this
question let us first define ``block diagonal'' in a rigorous way.

\begin{defi} \label{def:2}
  Consider a separable Hilbert space $\mathcal{K}$, a finite or countably infinite index set $I$, a sequence $(E^{(\mu)})_{\mu \in I}$ of
  orthonormal projections on $\mathcal{K}$ satisfying $\sum_\mu E^{(\mu)} = \Bbb{1}$ (converging strongly if $I$ is infinite) and the dense domain
  \begin{equation}
    D = \{ \psi \in \mathcal{K}\, | \, \exists K \in \Bbb{N}\, \forall \mu > K\, : \ E^{(\mu)} \psi = 0 \}.
  \end{equation}
  A (not necessarily bounded) operator $A: D \rightarrow D$ is called \emph{block diagonal} (with respect to the sequence $E^{(\mu)}$),
  if\footnote{Note that the sum in Eq. (\ref{eq:25}) is finite due to the definition of $D$.}
  \begin{equation} \label{eq:25}
    A \psi = \sum_{\mu \in I} A^{(\mu)} \psi, \quad \forall \psi \in D \quad \text{with}\quad A^{(\mu)} = E^{(\mu)} A E^{(\mu)}
  \end{equation}
  holds with bounded operators $A^{(\mu)}$ on $\mathcal{K}^{(\mu)} = E^{(\mu)}\mathcal{K}$.
\end{defi}

We apply this definition to $\mathcal{K} = \mathcal{H}$ and the projections $E^{(\mu)}$ onto the subspaces $\mathcal{H}^{(\mu)}$. Obviously the domain
$D$ becomes $D_{K_2}$ and we can define
\begin{equation} \label{eq:33}
  \PAc(K_2) = \{ A : D_{K_2} \rightarrow D_{K_2}\,|\, A\ \text{is linear and block diagonal}\}
\end{equation}
Now we can restate our question from above as: Does $\PAext(K_2) = \PAc(K_2)$ hold? The answer is: no but almost. To make this clearer note first that
$\PAc(K_2)$ is an associative, complex algebra under operator products and even a *-algebra with $\PAc(K_2) \ni A \mapsto A^+ \in \PAc(K_2)$ given by
\begin{equation}
  A^+\psi = \sum_{\mu \in \Bbb{N}_0} (A^{(\mu)})^* \psi \quad \forall \psi \in D_{K_2}.
\end{equation}
Furthermore we can equip $\PAc(K_2)$ with a family of seminorms
\begin{equation} \label{eq:26}
  \PAc(K_2) \ni A \mapsto \|A\|^{(\mu)} = \| A^{(\mu)}\| \in \Bbb{R} \quad \mu \in \Bbb{N}_0.
\end{equation}
It is easy to see (please check) that $\PAext(K_2)$ becomes with this family a Frechet space. In this topology $\PAext(K_2)$ is a dense subspace of
$\PAc(K_2)$. 

In order to prove the last statement we will use a stronger result, already shown in \cite{KZSH}. It requires some additional notations
\begin{align}
  \label{eq:34} \mathrm{U}(K_2) &= \{ U \in \PAc(K_2)\,|\, U^+U = UU^+ = \Bbb{1} \}\\
  \mathrm{SU}(K_2) &= \{ U \in \mathrm{U}(K_2)\,|\, \det U^{(\mu)} = 1 \ \forall \mu \in \Bbb{N}_0\}\\
  \mathfrak{u}(K_2) &= \{ A \in \PAc(K_2)\, | \, A^+ = - A \}\\
  \mathfrak{su}(K_2) &= \{ A \in \mathfrak{u}(K_2)\,|\, \tr A^{(\mu)} = 0 \ \forall \mu \in \Bbb{N}_0\}\\
  \label{eq:35}\mathfrak{sl}(K_2) &= \{ A \in \PAc(K_2)\,|\, \tr A^{(\mu)}= 0 \ \forall \mu \in \Bbb{N}_0 \}
\end{align}
Note here that the subspaces $\mathcal{H}^{(\mu)}$ are finite dimensional. Hence no problems with the definitions of $\tr$ and $\det$
arises. Furthermore, by restricting it to $D_{K_2}$, we have considered the unit operator $\Bbb{1}$ as an element of $\PAc(K_2)$. 

As an associative algebra $\PAc(K_2)$ becomes a complex Liealgebra if we equip it with the operator commutator as the Liebracket. The subspaces
$\mathfrak{u}(K_2)$ and $\mathfrak{su}(K_2)$ are real Lie-subalgebras of $\PAc(K_2)$ and $\mathfrak{sl}(K_2)$ is the complexification of
$\mathfrak{su}(K_2)$. Furthermore by applying Prop. \ref{prop:3} (or more precisely a slight generalization of it) we see that all formally
selfadjoint elements $A$ of $\PAc(K_2)$ are essentially selfadjoint on $D_{K_2}$. Hence, by using their closures $\overline{A}$, we get an exponential
map 
\begin{equation} \label{eq:24}
  \mathfrak{u}(K_2) \ni A \mapsto \exp(\overline{A}) \in \mathrm{U}(K_2).
\end{equation}
In this way $\mathfrak{u}(K_2)$ becomes the Lie algebra of the Frechet-Lie group $\mathrm{U}(K_2)$. Similarly, $\mathfrak{su}(K_2)$ is the Lie algebra
of $\mathrm{SU}(K_2)$. Also note that $\mathrm{U}(K_2)$ and $\mathrm{SU}(K_2)$ are strongly and weakly closed subgroups of the unitary group
$\mathrm{U}(\mathcal{H})$ of $\mathcal{H}$.

Now, let us return to the operators $Y, Z$. Obviously, they are block diagonal and the blocks are trace free. Hence $iY, iZ \in \mathfrak{su}(K_2)$
and we can ask for the Lie subalgebra $\langle iY, i Z\rangle_{\Bbb{R}, \mathrm{Lie}} \subset \mathfrak{su}(K_2)$ generated by them. According to
\cite{KZSH} it has the following structure:

\begin{lem} \label{lem:13}
  For all $K \in \Bbb{N}$ and all tuples $(\tilde{A}^{(\mu)})_{\mu < K}$ with $\tilde{A}^{(\mu)} \in \mathcal{B}(\mathcal{H}^{(\mu)})$ selfadjoint,
  there is a $A \in \langle iY, i Z\rangle_{\Bbb{R}, \mathrm{Lie}} \subset \mathfrak{su}(K_2)$ such that $A^{(\mu)} = i \tilde{A}^{(\mu)}$ holds for
  all $\mu < K$. 
\end{lem}

Given the definition of the topology of $\PAc$ in Eq. (\ref{eq:26}) we can immediately rephrase the result as: The \emph{complex} Lie algebra $\langle
iY, i Z\rangle_{\Bbb{C}, \mathrm{Lie}}$ is dense in $\mathfrak{su}(K_2)$, and since $\langle iY, i Z\rangle_{\Bbb{C}, \mathrm{Lie}}$ is a subspace of
$\PAext(K_2)$ we get:

\begin{prop}\label{prop:4}
  The extended path algebra $\PAext(K_2)$ is a dense subspace of $\PAc$.
\end{prop}

This clarifies the role of $\PAc$ and shows in addition that the subspaces $\mathcal{K}^{(\mu)}$ are the \emph{minimal} invariant subspaces of
$\PAext(K_2)$. Hence the decomposition of $\mathcal{H}$ from (\ref{eq:27}) is uniquely determined by $\PAext(K_2)$ only, and since $\PAext(K_2)$ is
determined by the graph $K_2$ all spaces just introduced (i.e. $\PAc(K_2)$, $\mathrm{U}(K_2)$, etc.) only depend on the graph $K_2$ and not on an
arbitrarily chosen sequence of projections. This justifies in retrospect the notations already used. 

The group $\mathrm{U}(K_2)$ can be introduced alternatively as the set of all unitaries commuting with the operator  $Q \in \PAext(K_2)$ given by
$Q\ket{\mu,\nu} = \mu \ket{\mu,\nu}$ (or more precisely with the selfadjoint extension of $Q$). Hence $\mathrm{U}(K_2)$ is identical with $\mathrm{U}(Q)$
from\footnote{In \cite{KZSH} the operator $Q$ was called $X$ which is, however, already used otherwise in this paper.} \cite{KZSH}. Similarly
$\mathrm{SU}(K_2)$, $\mathfrak{u}(K_2)$ and $\mathfrak{su}(K_2)$ are identical with $\mathrm{SU}(Q)$, $\mathfrak{u}(Q)$ and $\mathfrak{su}(Q)$. This
connects the old symmetry based setting with the path algebra approach introduced in this paper. Along these lines we can reuse a result from
\cite{KZSH} which clarifies the structure of control problem (\ref{eq:21}):

\begin{prop} \label{prop:2}
  $\mathcal{G}(Y,Z) = \mathrm{SU}(K_2)$
\end{prop}

\begin{proof}
  This follows from Lemma \ref{lem:13} and properties of the exponential map from (\ref{eq:24}); cf. \cite{KZSH} for details.
\end{proof}

In other words, all unitaries in the path algebra (with blocks of determinant $1$) can be implemented (approximately) by only using the Hamiltonians $Y$ and $Z$. To calculate the
corresponding control functions we can cut off the direct sum (\ref{eq:27}) at any index $\mu$ (depending on the accuracy we require) and end up with a
finite dimensional problem. Since the truncated operators $Y, Z$ can be represented by sparse matrices the corresponding optimization can be done
efficiently even for high dimensions. The only remaining problem is, how to implement an arbitrary unitary, or a little bit easier, how to prepare an
arbitrary state from the ground state $\ket{0;0}$. Obviously the ``symmetry breaking'' operator $X$ has to be involved here (``symmetry breaking'' now
should read: not in the path algebra). For the state preparation problem a general algorithm was used in \cite{KZSH} (which was in fact used already in
a number of older papers e.g. \cite{brockett2003controllability,rangan2004control,yuan2007controllability,bloch2010finite}). 

\begin{wrapfigure}{l}{0.45\textwidth}
  \includegraphics[width=0.4\textwidth]{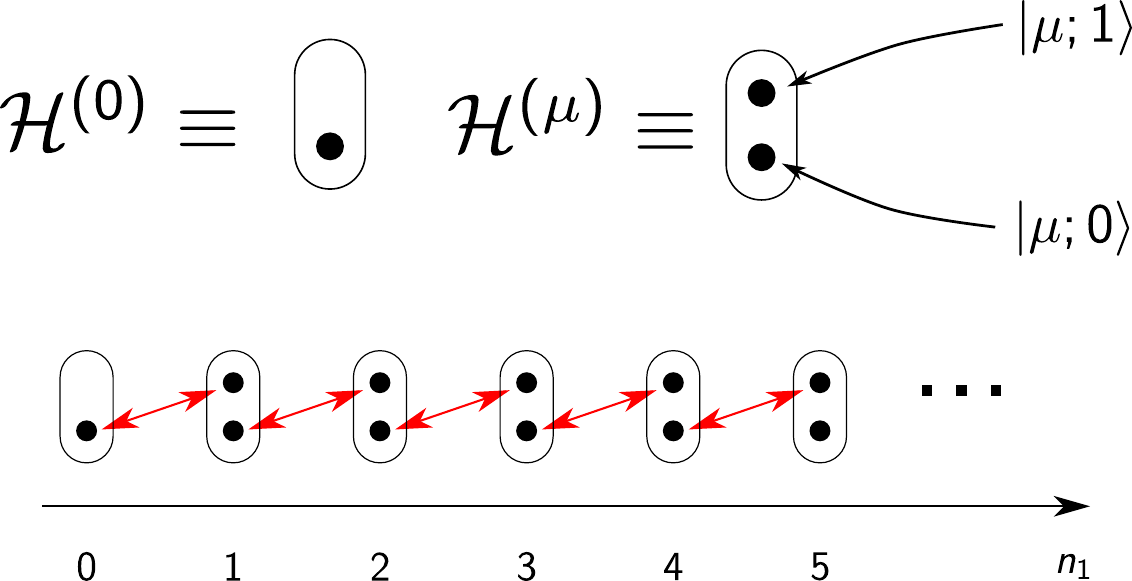}
  \caption{The decomposition of $\mathcal{H}$ into invariant subspaces. The boxes with two (or one) dots represent the subspaces
    $\mathcal{H}^{(\mu)}$, while the dots itself depict the basic vectors. The red arrows indicate the action of $\exp(i\pi X)$. \label{fig:2ddecomp}} 
\end{wrapfigure}
We consider a vector $\psi \in D_{K_2}$. Each such $\psi$ admits a constant $K \in \Bbb{N}$ such that $E^{(\mu)} \psi = 0$ holds for all $\mu > K$ and
$E^{(K)}\psi \neq 0$; cf. the projections $E^{(\mu)}$ introduced in Definition \ref{def:2}. 
Our goal is to transform $\psi$ into $e^{i \alpha} \ket{0;0}$ $\alpha \in
\Bbb{R}$ arbitrary, by using only unitaries of the form $\exp(i t_x X)$, $\exp(i t_y Y)$ and $\exp(i t_Z Z)$ with appropriate $t_x, t_y, t_z \in
\Bbb{R}^+$. Using Proposition \ref{prop:2} and arguments from the last paragraph we assume further that there is a fast algorithm to express (at least
approximately and with arbitrary good accuracy) any $U \in \mathrm{SU}(K_2)$ as a product of $\exp(i t_y Y)$ and $\exp(i t_Z Z)$. We do not care how
this is done explicitly, such that we are looking at a sequence $U_1, \dots, U_N$ of unitaries consisting of elements from $\mathrm{U}(K_2)$ and
$\exp(i t_x X)$. For the latter we only look at $t_x = \pi$ which produces a flip of $\ket{\mu,0}$ and $\ket{\mu+1,1}$. This is indicated by the red
arrows in Fig. \ref{fig:2ddecomp}. Note that we have to flip all pairs of vectors simultaneously, while the elements of $\mathrm{U}(K_2)$ can
manipulate each $\mathcal{H}^{(\mu)}$ (i.e. the boxes in Fig. \ref{fig:2ddecomp}) individually. With this prerequisites we can proceed as follows:

\begin{enumerate}
\item \label{item:1}
  Apply a unitary $U_1 \in \mathrm{SU}(K_2)$ to $\psi$ such that $\langle \mu; 0| U_1 \psi\rangle = 0$ holds for all $\mu > 0$. In other words we
  rotate the vectors $E^{(\mu)} \psi \in \mathcal{H}^{(\mu)} \cong \Bbb{C}^2$ with $0 < \mu \leq K$ towards $\ket{\mu, 1}$ until the $\ket{\mu;0}$
  components become zero. This is always possible, due to the block diagonal structure of $\mathrm{U}(K_2)$.
\item 
  Apply $U_2 = \exp(i \pi X)$ to $\psi_1 = U_1\psi$. This flips $\ket{K-1;0}$ and $\ket{K,1}$. Hence, since the overlap of $\psi_1$ with $\ket{K-1;0}$
  is zero by step \ref{item:1} the resulting vector $\psi_2 = U_2 \psi_1$ satisfies the initial assumptions with $K$ decremented by $1$.
\item 
  We continue this procedure $K-2$ times to get a vector $\psi_{2K}$ which overlaps only with $\mathcal{H}^{(0)}$ and $\mathcal{H}^{(1)}$.
\item 
  We apply $U_{2K+1} \in \mathrm{SU}(K_2)$ to $\psi_{2K}$ such that $E^{(1)} \psi_{2K}$ is rotated towards $\ket{1,1}$. Hence, the only non-zero
  components of $\psi_{2K+1} = U_{2K+1} \psi_{2K}$ are $\langle 0; 0| \psi_{2K+1}\rangle$ and $\langle 1;1 | \psi_{2K+1}\rangle$. Or in other words
  $\psi_{2K+1} = \tilde{\psi}_{2K+1} \otimes \ket{0} \in \Bbb{C}^2 \otimes \mathrm{L}^2(\Bbb{R}) = \mathcal{H}$. 
\item 
  $X$ and $Y$ operate on $\mathcal{H}$ as $\sigma_1 \otimes \Bbb{1}$ and $\sigma_3 \otimes \Bbb{1}$. Hence there is a combination $U_{2K+2}$ of $X$
  and $X$ rotation which transforms $\psi_{2K+1} = \tilde{\psi}_{2K+1} \otimes \ket{0}$ into $e^{i \alpha} \ket{0;0}$ as required.
\end{enumerate}

If $\psi$ is an arbitrary vector in $\mathcal{H}$ we can find for all $\epsilon>0$ a $K \in \Bbb{N}$ such that $\|\psi^{[K]} - \psi\| < \epsilon$
holds with $\psi^{[K]} = \sum_{\mu < K} E^{(\mu)} \psi$. Hence, by applying the algorithm just described to $\psi^{[K]}$ we get a family of unitaries
$U_1, \dots, U_{2K+2}$ which transforms $\psi$ into a final vector $\psi_f = U_{2K+2} \dots U_1 \psi$ with $\| \psi_f - e^{i \alpha}\ket{0;0}\| <
\epsilon$. In other words we can transform any pure state $\kb{\psi}$, $\psi \in \mathcal{H}$ approximately, but with arbitrary precision into
$\kb{0;0}$. By unitarity we can reverse the procedure to reach any state $\kb{\psi}$ up to a an arbitrary small error from the ground state
$\kb{0;0}$. Finally if we want to relate two pure states $\kb{\psi}$ and $\kb{\phi}$ we can stack two sequences of unitaries together: We start by 
transforming $\kb{\psi}$ (approximately) into $\kb{0;0}$ and then we transform $\kb{0;0}$ into $\kb{\phi}$ -- again with an arbitrary error
$\epsilon>0$.  The given procedure is in general very far from being optimal, but it explains how a state preparation can be done (at least in
principle) within the given setup.

This completes the discussion of two levels. We have connected the previous work from \cite{KZSH} to our current setup and seen that the path algebra
basically replaces the symmetry arguments from \cite{KZSH}. We will use this idea as a guide to study 3-level systems and to rediscuss the state
preparation problem.

\section{Example 2: Three level atoms}
\label{sec:example-2:-three}

Our next goal is to translate our discussion from the last section to 3-level atoms. Note  that parts of the material from this and the next section
can also be found in \cite{HofMa}. In contrast to two levels, the structure is already rich enough to indicate what we can expect from the general
case. A short inspection shows that only the four different graphs shown in Figure \ref{fig:3-graphs} satisfy the conditions from Section
\ref{sec:description-problem}. The cases $\Gamma_C$ (``Cascade''), $\Gamma_V$ (``V-shaped'') and  $\Gamma_\Lambda$ (``$\Lambda$-shaped'') are tree
graphs and treated in this section. The ``$\Delta$-configuration'' $\Gamma_\Delta$ contains a cycle which makes its discussion more difficult. It is
postponed therefore to the next section. 

\begin{figure}[h]
  \centering
  \bigskip
  \includegraphics[width=0.8\textwidth]{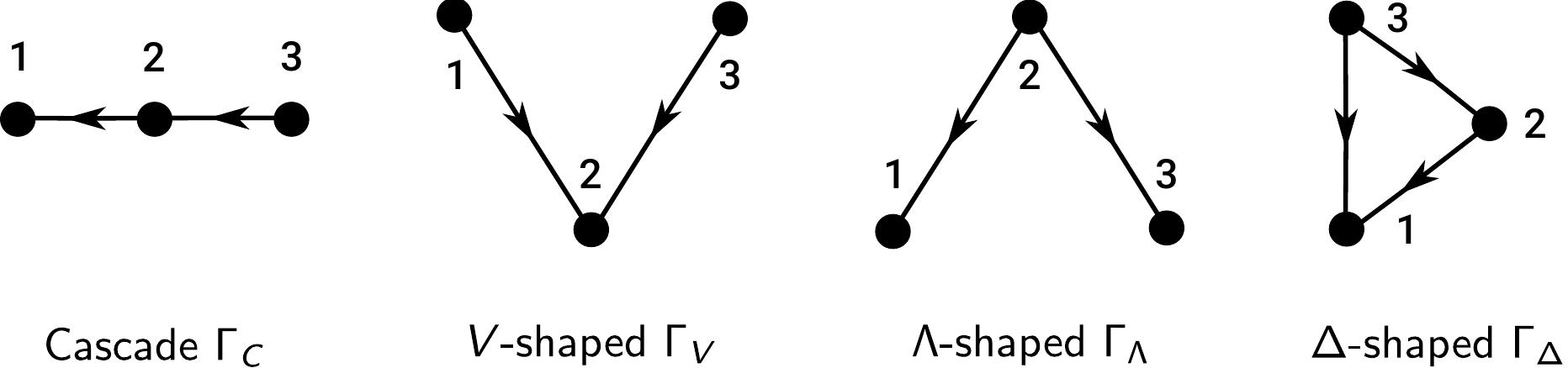}
  \caption{All ordered 3-graphs satisfying the condition from Sect. \ref{sec:description-problem}. \label{fig:3-graphs}}
\end{figure}

As graphs without orientation $\Gamma_C, \Gamma_V$ and $\Gamma_\Lambda$ are identical. In all three cases we can write
\begin{equation}
  V(\Gamma_{\#}) = \{1,2,3\}\quad\text{and}\quad E(\Gamma_{\#}) = \{(1,2), (2,1), (2,3), (3,2)\}\quad\text{where}\quad \# =
  C,V,\Lambda,
\end{equation}
and inversion of edges $e \mapsto \overline{e}$ is given by the map $E(\Gamma_{\#})\ni (a,b) \mapsto \overline{(a,b)} = (b,a) \in E(\Gamma_{\#})$. The
distinction between the $\Gamma_{\#}$ arises from different choices for $E_+(\Gamma_{\#})$: We have  
\begin{equation}
  E_+(\Gamma_C) = \{(2,1), (3,2)\},\quad E_+(\Gamma_V) = \{(1,2), (3,2)\},\quad E_+(\Gamma_\Lambda) = \{(2,1), (2,3)\}.
\end{equation}
Apparently there is a fourth possibility $E_+(\Gamma_{\tilde{C}}) = \{(1,2), (2,3)\}$, but this is just the cascade reversed. In other words it arises
from $\Gamma_C$ by exchanging the vertices $1$ and $3$. Therefore it does not lead to a new system and it is omitted.

\paragraph{The control problem}
Due to these similarities the control problems associated to these four graphs are closely related. The Hilbert space is the same for all cases:
$\mathcal{H} = \Bbb{C}^3 \otimes \mathrm{L}^2(\Bbb{R})^2$, and the canonical basis $\ket{b}$, $b \in \ConfP(\Gamma_{\#})$ becomes 
\begin{equation} \label{eq:32}
  \ket{j} \otimes \ket{n_1} \otimes \ket{n_2} =\ket{j;n_1,n_2} \in \Bbb{C}^3 \otimes \mathrm{L}^2(\Bbb{R})^2,\quad j \in \{1,2,3\},\ n_1,n_2 \in
  \Bbb{N}_0, 
\end{equation}
where $\ket{j} \in \Bbb{C}^3$ denotes the canonical basis and $\ket{n_1},\ket{n_2} \in \mathrm{L}^2(\Bbb{R})$ is the number basis. 

Now we define for $\alpha,\beta \in V(\Gamma_{\#})$ the operators
\begin{equation} \label{eq:37}
  X^{(\alpha,\beta)} = \bigl(\KB{\alpha}{\beta} - \KB{\beta}{\alpha}\bigr) \otimes \Bbb{1}^{\otimes 2},\quad Y^{(\alpha,\beta)} = \bigl(\kb{\alpha} -
  \kb{\beta}\bigr) \otimes \Bbb{1}^{\otimes 2}
\end{equation}
and
\begin{alignat}{2}
  Z^{(1,2)} &= \KB{1}{2} \otimes a \otimes \Bbb{1} + \KB{2}{1} \otimes a^* \otimes \Bbb{1}, \quad &Z^{(2,1)} &= \KB{2}{1} \otimes a \otimes \Bbb{1} +
  \KB{1}{2} \otimes a^* \otimes \Bbb{1}\\
  Z^{(2,3)} &= \KB{2}{3} \otimes \Bbb{1} \otimes a  + \KB{3}{2} \otimes \Bbb{1} \otimes a^*,\quad &Z^{(3,2)} &= \KB{3}{2} \otimes \Bbb{1} \otimes a  +
  \KB{2}{3} \otimes \Bbb{1} \otimes a^*.
\end{alignat}
This definition allows us to associate to the graph $\Gamma_{\#}$ the control Hamiltonians
\begin{equation}
  X^{(\alpha,\beta)}, Y^{(\alpha,\beta)}, Z^{(\alpha,\beta)},\quad (\alpha,\beta) \in E_+(\Gamma_{\#})
\end{equation}
and the drift Hamiltonian
\begin{equation} \label{eq:40}
  H_D = \omega_{C,1} \Bbb{1} \otimes a^*a \otimes \Bbb{1} + \omega_{C,2} \Bbb{1} \otimes \Bbb{1} \otimes a^*a + \sum_{(\alpha,\beta) \in
    E_+(\Gamma_{\#})} \left( \omega_{A,\alpha,\beta} Y^{(\alpha,\beta)} + \omega_{I,\alpha,\beta} Z^{(\alpha,\beta)} \right)
\end{equation}
As before the operators $Z^{(\alpha,\beta)}$ and $H_D$ are unbounded and essentially selfadjoint on the default domain $D_\Gamma$ which does not
depend on the choice $\Gamma = \Gamma_{\#}$. The control problems connected to the three graphs can therefore be written in a unified way as
\begin{equation} \label{eq:31}
  i \frac{d}{dt} U(0,t) \psi = H_D U(0,t) \psi + \sum_{(\alpha,\beta) \in E_+(\Gamma_{\#})}\left( u_{\alpha,\beta}(t) X^{(\alpha,\beta)} U(0,t) \psi
    + v_{\alpha,\beta}(t) Y^{(\alpha,\beta)} U(0,t) \psi\right)
\end{equation}
and 
\begin{multline} \label{eq:41}
  i \frac{d}{dt} U(0,t) \psi = \\ \sum_{(\alpha,\beta) \in E_+(\Gamma_{\#})}\left(u_{\alpha,\beta}(t) X^{(\alpha,\beta)} U(0,t) \psi +
    v_{\alpha,\beta}(t) Y^{(\alpha,\beta)} U(0,t) + w_{(\alpha,\beta)}(t) Z^{(\alpha,\beta)} U(0,t)\psi\right),
\end{multline}
with piecewise constant control functions $u_{\alpha,\beta}, v_{\alpha,\beta}$ and $w_{\alpha,\beta}$. According to Theorem \ref{thm:2} and
Proposition \ref{prop:1} both problems are strongly controllable. Note that strong controllability means in case of (\ref{eq:41}) that
$\mathcal{G}(X^{(\alpha,\beta)}, Y^{(\alpha,\beta)}, Z^{(\alpha,\beta}; (\alpha,\beta) \in E_+(\Gamma_{\#})) = \mathrm{U}(\mathcal{H})$ holds;
cf. Section \ref{sec:example-1:-two} and \cite{KZSH}.

If we look at the dependency of $X,Y,Z$ on the graph $\Gamma_{\#}$ we see that the $X$ does not depend at all, while the $Y$ only differ by
signs. Hence different structure can only arise from the $Z$ and therefore from the structure of the (extended) path algebra $\PAext(\Gamma)$, which we
will analyze below. Note that parts of the following discussion can be applied to general graphs or at least to general tree graphs. 

\paragraph{Invariant subspaces and the photon game}
Our first step is to construct the minimal invariant subspaces; cf. Theorem \ref{thm:3}. A general
strategy is to reuse the method from the proof of Theorem \ref{thm:3}: Start with a configuration $b \in \ConfP(\Gamma)$ and generate all vectors of the
form $\ket{\gamma \cdot b}$ to get the minimal invariant subspace containing $\ket{b}$:
\begin{equation} \label{eq:30}
  \mathcal{H}^{(b)} = \SP \{ \ket{\gamma \cdot b}\,|\, \gamma\ \text{path in}\ \Gamma\};
\end{equation}
cf. Eqs (\ref{eq:3}) and (\ref{eq:28}). According to Lemma \ref{lem:14} this is equivalent to $\mathcal{H}^{(b)} = \SP \{\ket{c}\,\|\, c \in
\ConfP(\Gamma),\ c \sim b \}$ with $c \sim b$ defined by: $c \sim b :\Leftrightarrow$ $\exists$ path $\gamma$ with $c=\gamma\cdot b$ and 
\begin{equation} \label{eq:29}
  b_0 = i(\gamma)  \quad \text{and}\quad \gamma_k \cdot b \ \text{is regular}\ \forall k = 0, \dots, N,
\end{equation}
holds, where the $\gamma_k$ denote the subpath of $\gamma$; cf. Eq. (\ref{eq:15}). Hence to collect all $c \in \ConfP(\Gamma)$ with
$\ket{c} \in \mathcal{H}^{(b)}$ and $\ket{c}\neq 0$ we can play the following combinatorial game (in the following called the \textbf{photon game}; cf
also Figure \ref{fig:photon-game}): 

\begin{enumerate}
\item \label{item:4}
  Choose a path $\gamma = (e_1, \dots, e_N)$ of $\Gamma$ which starts at $b_0 \in V(\Gamma)$ and ends at $c_0\in V(\Gamma)$. Attach to each positive
  edge $e \in E_+(\Gamma)$ an integer $n_e$ which is initialized to $n_e = \underline{b}(e)$.
\item 
  Walk along $\gamma$ from $b_0$ to its end at $c_0$. In the $j^{\mathrm{th}}$ step we pass edge $e_j$. If $e_j \in E_+(\Gamma)$ increment
  $n_{e_j}$. Otherwise (i.e. if $e_j \in E_-(\Gamma)$) decrement $n_{\overline{e_j}}$.
\item 
  If $e_j \in E_-(\Gamma)$ and after the $j^{\mathrm{th}}$ step the number $n_{\overline{e_j}}$ is negative the process failed and we have to choose
  the next path at item \ref{item:4}.
\item 
  If none of the $n_e$ becomes negative during the whole process we reach a regular configuration $c = \gamma \cdot b$ with $c \sim b$ and therefore
  $\ket{c} \in \mathcal{H}^{(b)}$. Note that it is not sufficient that the $n_e$ are non-negative \emph{at the end}, they have to be non-negative at each
  step. This is exactly the contents of Eq. (\ref{eq:29}). Also note that at the end of the path the numbers $n_e$ become $\underline{c}(e)$.
\item \label{item:5}
  Restart the process at item \ref{item:4} until a basis of $\mathcal{H}^{(b)}$ is reached.
\end{enumerate}

\begin{figure}[b]
  \centering
  \bigskip
  \includegraphics[width=0.8\textwidth]{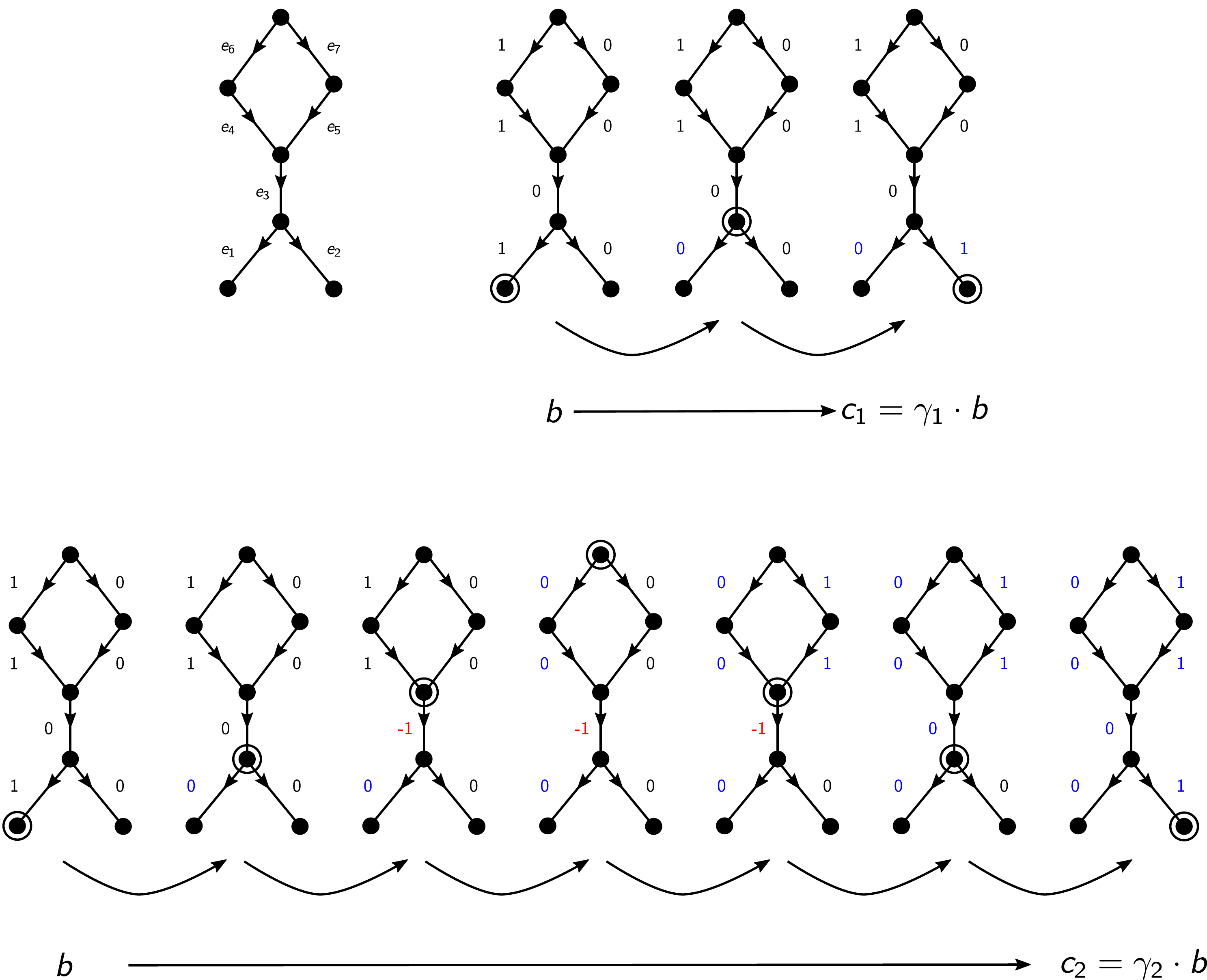}
  \caption{Playing the photon game on the graph shown in the upper left of the figure. Along the path $\gamma=(\overline{e_1}, e_2)$ it is successfully, along
    the path $\gamma_2 = (\overline{e_1}, \overline{e_3}, \overline{e_4}, \overline{e_6}, e_7, e_5, e_3,e_2)$ it is not. In the second case the photon number
    associated to edge $e_3$ becomes negative in between. Hence the configuration $c_1$ is equivalent to $b$ but $c_2$ is not. \label{fig:photon-game}}
\end{figure}

The biggest problem is the lack of an easy condition to check whether the process is finished in item \ref{item:5}. This problem can be solved by
restricting the set of paths from which we have to choose $\gamma$ in item \ref{item:4} to a finite set. For tree graphs this can be done by looking
at \emph{straight} paths. A path is called straight if it does not go back and forth along the same edge, or more precisely if $e_{j+1} \neq
\overline{e_j}$ holds for all $j=1,.\dots, N-1$. To each path $\gamma$ we can associate a unique straight path $\hat{\gamma}$ by subsequently removing
all pairs of edges $e_j, e_{j+1}$ with $e_{j+1} = \overline{e_j}$. It is easy to see that $\ket{\gamma \cdot b} \neq 0$ $\Rightarrow$
$\ket{\hat{\gamma} \cdot b} \neq 0$ holds. Therefore we get

\begin{lem} \label{lem:15}
  For each $b \in \ConfP(\Gamma)$ we have (cf. Eq. (\ref{eq:30})): 
  \begin{equation}
    \mathcal{H}^{(b)} = \SP \{ \ket{\gamma \cdot b}\,|\, \gamma\ \text{straight path in}\ \Gamma\}.
  \end{equation}
\end{lem}

In a tree graph (i.e. a graph without cycles) there is a unique straight path between any pair of vertices $v \neq w$.  (Please check!) Hence with the
observation from Lemma \ref{lem:15} we can supplement the procedure from above as follows: In item \ref{item:4} only choose straight paths, and in
item \ref{item:5} terminate the procedure if the set of straight paths is exhausted. This will finish the game (for a fixed starting configuration
$b$) after finitely many repetitions.   

The last problem to be solved is the labeling of the invariant subspaces. To use configurations as we have done until now is ambiguous since in
general there are $c \in \ConfP(\Gamma)$ with $\ket{c} \in \mathcal{H}^{(b)}$ and $c\neq b$. For tree graphs we can solve this problem in terms of a
map which we will define in the following. Firstly we need an enumeration of $E_+(\Gamma)$, i.e. a bijective map $\{1, \dots, d\} \ni j \mapsto e_j
\in E_+(\Gamma)$, with $d=|E_+(\Gamma)|$. This allows us to rewrite configurations $b \in \ConfP(\Gamma)$ as $d+1$-tuples $b=(b_0,b_1, \dots, b_d)$
with $b_j = b(e_j)$. Secondly we have to choose an arbitrary but fixed vertex $v_0$. For each $b \in \ConfP(\Gamma)$ there is a unique straight path
$\gamma$ starting at $v_0$ and ending at $b_0$, where we include the empty path to allow $v_0$ to be the start and the end at the same time. Now we
define 
\begin{equation}
  \ConfP(\Gamma) \ni b \mapsto \nu(b) = (c_1, \dots, c_d) \in \Bbb{Z}^d,\quad\text{with}\quad b = \gamma \cdot c, \quad c=(v_0, c_1, \dots, c_d).
\end{equation}
Hence, for a given $b \in \ConfP(\Gamma)$ we look for the unique  (and in general non-regular) configuration $c$ such that $b = \gamma \cdot c$ and
$c_0 = v_0$ hold. The photon numbers $(c_1, \dots, c_d)$ of the configuration $c$ are then the result of  the map, which has the following properties:

\begin{prop} \label{prop:5}
  For a tree graph $\Gamma$ the map $\nu$ just defined has the following properties:
  \begin{enumerate}
  \item \label{item:8}
    $\nu(b) = \nu(\tilde{b})$ $\Leftrightarrow$ $\mathcal{H}_b = \mathcal{H}_{\tilde{b}}$.
  \item \label{item:9}
    The range $\Delta$ of $\nu$ satisfies $\Bbb{N}_0^d \subset \Delta \subset [-1,\infty)^d$, where the intervals refer to subsets of $\Bbb{Z}$,
    rather than $\Bbb{R}$.
  \item \label{item:10}
    If $v_0 < w$ holds for all $w \in V(\Gamma)$ we have $\Delta = \Bbb{N}_0^d$.
  \end{enumerate}
\end{prop}

\begin{proof}
  Large parts of the proof relies on the following lemma which simplifies the definition of the equivalence relation $\sim$ for tree graphs.

  \begin{lem} \label{lem:16}
    For a tree graph $\Gamma$ and two regular configurations $b$, $c$ we have: $b \sim c$ $\Leftrightarrow$ $c = \gamma \cdot b$ with a straight
    path $\gamma$.
  \end{lem}

  \begin{proof}
    Lets assume first that $c = \gamma \cdot b$ with a straight path $\gamma$ and write $\gamma=(e_1,\dots, e_N)$. Since $\Gamma$ is a tree and
    $\gamma$ straight each geometric edge appears in $\gamma$ at most once, i.e. $e_j = e_k$ or $e_j = \overline{e_k}$ implies $j=k$. If one edge
    could appear twice the path $\gamma$ would contain either a cycle (which is impossible since $\Gamma$ is a tree) or $\gamma$ would move back and
    force (which is not allowed since $\gamma$ is straight). Hence if we propagate $b \in \ConfP(\Gamma)$ along $\gamma$ (i.e. calculating $\gamma_k
    \cdot b$ for all subpath $\gamma_k$, $k=1,\dots,N$) by playing the photon game, the numbers $n_e \in \Bbb{Z}$ can change in three different ways:
    $n_e$ is incremented exactly once if $c(e) = b(e) + 1$, it is decremented exactly once if $c(e) = b(e) -1$, and if $c(e) = b(e)$ it remains
    unchanged during the whole process. Since $b,c$ are both regular the $n_e$ can never become negative, or in other words all $\gamma_k \cdot b$ are
    regular. Hence $b \sim c$. If on the other hand $b \sim c$ holds the existence of a straight $\gamma$ follows directly from Lemma \ref{lem:15}
    and the definition of the relation $\sim$. 
  \end{proof}

  Let us consider statement \ref{item:8} from the proposition. If $\nu(b) = \nu(\tilde{b})$ there are straight paths $\gamma, \tilde{\gamma}$ and a
  configuration $c \in \ConfP(\Gamma)$ with $c_0 = v_0$ such that $b = \gamma \cdot c$ and $\tilde{b}=\tilde{\gamma} \cdot c$. Hence by concatenating the
  paths $\gamma$ and $\tilde{\gamma}^{-1}$ we get a new path $\gamma_1$ with $\tilde{b} = \gamma_1 \cdot b$. In general $\gamma_1$ is not straight,
  but if we remove subsequently all pairs $e_j = e, e_{j+1} = \overline{e}$ as described above we get another path $\gamma_2$ which is straight and
  satisfies again $\tilde{b} = \gamma_2 \cdot b$. Since $b$ and $\tilde{b}$ are regular Lemma \ref{lem:16} implies that $b \sim \tilde{b}$ holds which
  is equivalent to $\mathcal{H}_b = \mathcal{H}_{\tilde{b}}$.

  Now assume $\mathcal{H}_b = \mathcal{H}_{\tilde{b}}$. There is unique straight path $\gamma_1$ from $b_0$ to $v_0$. Propagating $b$ along $\gamma_1$ 
  leads to a configuration $c = \gamma_1 \cdot b$ (which is in general not regular). Using the inverse path $\gamma = \gamma_1^{-1}$ we get $b = 
  \gamma \cdot c$. Since $\mathcal{H}_b = \mathcal{H}_{\tilde{b}}$ we have $b \sim \tilde{b}$ and therefore there is a path $\gamma_2$ with $\tilde{b}
  = \gamma_2 \cdot b$. Concatenating $\gamma$ and $\gamma_2$ leads to a path $\tilde{\gamma}$ with $\tilde{b} = \tilde{\gamma} \cdot c$. Since we can
  subsequently remove pairs of edges $e,\overline{e}$ we can assume without loss of generality that $\tilde{\gamma}$ is straight. Hence $\nu(b) =
  \nu(\tilde{b})$. 

  Statement \ref{item:9}. For each $(n_1, \dots, n_d) \in \Bbb{N}_0^d$ there is a unique regular configuration $b$ with $b_0=v_0$ and $b(e_j)=n_j$. If
  $\gamma = ()$ is the empty path we have $\gamma \cdot b = b$. Hence $(n_1, \dots, n_d) \in \Delta$. If on the other hand $c = \gamma \cdot b$ with
  $b \in \ConfP(\Gamma)$ and a straight path $\gamma$, we have $c_j \geq b_j - 1$ for all $j=1,\dots,d$, since each $b_j$ can be decremented at most
  once. Since $b$ is regular $b_j \geq 0$. Hence $c_j \geq -1$, as stated. 

  Statement \ref{item:10}. Consider $b \in \ConfP(\Gamma)$. Since $v_0 < b_0$ the unique path from $b_0$ to $v_0$ consists only of positive
  edges. Hence, while playing the photon game along $\gamma$ the numbers $n_j$ are never incremented such that $c=\gamma\cdot b$ satisfies $c_j \geq
  0$. The statement follows with item \ref{item:9}.
\end{proof}

The map $\nu$ provides a labelling of the invariant subspaces in terms of elements of the set $\Delta$. We define
\begin{equation} \label{eq:39}
  \mathcal{H}^{(n_1, \dots, n_d)} = \mathcal{H}_b \quad\text{with}\quad (n_1,\dots,n_d) = \nu(b).
\end{equation}
According to Proposition \ref{prop:5}(\ref{item:8}) we can replace $\mathcal{H}_b$ on the right hand side of this equation by any Hilbert space
$\mathcal{H}_{\tilde{b}}$ with $\mathcal{H}_b = \mathcal{H}_{\tilde{b}}$ without changing $(n_1,\dots,n_d)$. Hence $\mathcal{H}^{(n_1,\dots,n_d)}$ is
well defined. It is also clear that all $\mathcal{H}_b$ are covered by this relabelling such that the system Hilbert space decomposes as
\begin{equation} \label{eq:36}
  \mathcal{H} = \sum_{(n_1,\dots,n_d) \in \Delta} \mathcal{H}^{(n_1,\dots,n_d)}.
\end{equation}
As a byproduct we can also introduce another relabelling -- this time of the basis vectors $\ket{b}$. For a graph $\Gamma$ the Hilbert space
$\mathcal{H}^{(n_1,.\dots,n_d)}$ contains for each vertex $v \in V(\Gamma)$ at most one basis vector $\ket{c} \in \mathcal{H}$ with $c_0 = v$. We
write 
\begin{equation}
  \ket{n_1,\dots,n_d;v} = \ket{c} \iff c_0 = v \ \text{and}\ \ket{c} \in \mathcal{H}^{(n_1,\dots,n_d)},
\end{equation}
and get a basis which is adapted to the decomposition (\ref{eq:36}).

\begin{figure}[b]
  \centering
  \bigskip
  \includegraphics[width=0.8\textwidth]{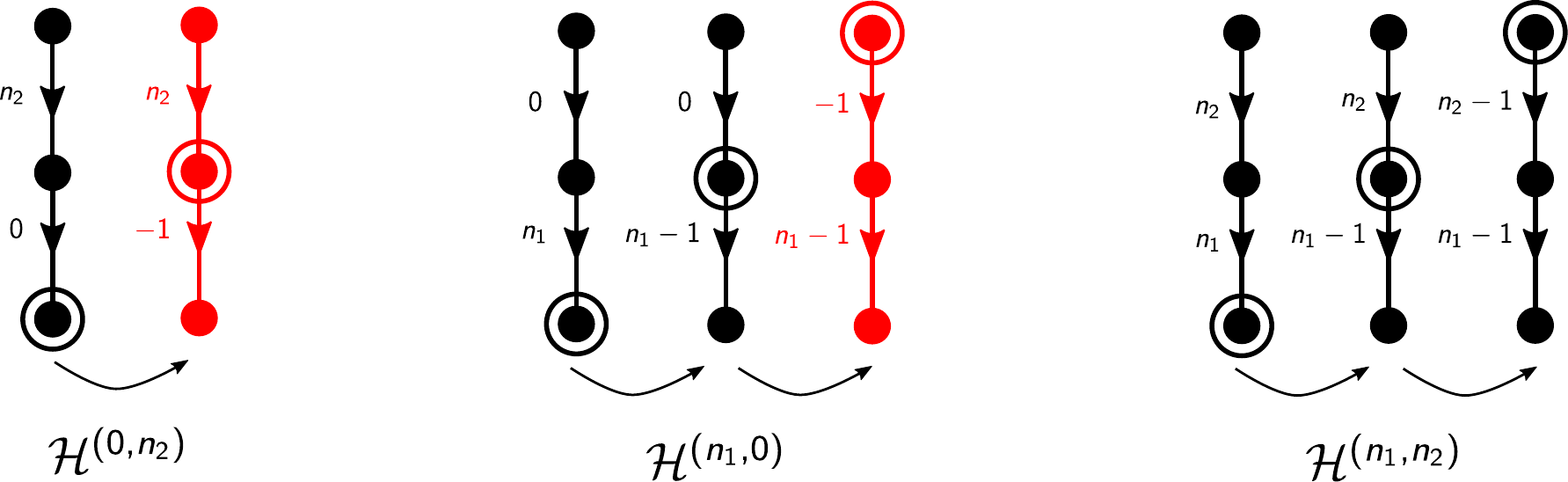}
  \caption{To generate the invariant subspaces $\mathcal{H}^{(n_1,n_2)}$ we play the photon game on the graph $\Gamma_C$ ($\Gamma_V$ and
      $\Gamma_{\Lambda}$ work similarly). If $n_1=0$ or $n_2=0$ the game fails at the second or third step. Hence the corresponding subspaces are only
      one- or two-dimensional. Only for $n_1>0$, $n_2>2$ the game is successful until the end such we get in these cases the (generic) three
      dimensional spaces. \label{fig:game-cascade}}
\end{figure}

Let us come back now to the special case of a 3-graph, i.e. $\Gamma=\Gamma_C$, $\Gamma_V$, or $\Gamma_\Lambda$. In all three cases we can apply the
procedure just introduced. For $\Gamma_C$ and $\Gamma_V$ statement \ref{item:10} of Proposition \ref{prop:5} applies, since in these cases the set of
vertices has a unique minimal element $v_0$ with $v_0 < w$ for all other vertices $w$. Hence the index set $\Delta$ coincides with $\Bbb{N}_0^2$. For
$\Gamma_\Lambda$ this is not the case and therefore the $n_1,n_2$ can become negative. If we choose $1 \in V(\Gamma_\Lambda)$ as the ``reference
vertex'' $v_0$ the index set $\Delta$ is $\Delta = \Bbb{N}_0^2 \cup \{-1\} \times \Bbb{N}$. The results for the construction in all three cases are
summarized in tables \ref{tab:1} to \ref{tab:3}, while table \ref{tab:4} describes the relation of the vectors $\ket{n_1,  n_2; v}$ to the basis
$\ket{v} \otimes \ket{n_1} \otimes \ket{n_2}$ from Equation (\ref{eq:32}). For the cascade ($\Gamma_C$) the whole procedure is demonstrated
graphically in Figure \ref{fig:game-cascade}.
\medskip

\begin{table}[h]
  \centering
  \begin{tabular}{|C{2.5cm}|C{7.5cm}|C{3cm}|}
    \hline
    \rule[-.2cm]{0pt}{.6cm}$n_1, n_2$ & basis of $\mathcal{H}^{(n_1,n_2)}$ &  $\dim(\mathcal{H}^{(n_1,n_2)})$ \\ \hline \hline
    \rule[-.2cm]{0pt}{.6cm}$n_1 > 0, n_2 > 0$ & $\ket{n_1,n_2;1}, \ket{n_1,n_2;2}, \ket{n_1, n_2; 3}$ & $3$ \\ \hline
    \rule[-.2cm]{0pt}{.6cm}$n_1>0, n_2=0$ & $\ket{n_1,n_2;1}, \ket{n_1,n_2;2} \phantom{,\ket{n_1, n_2; 3}}$  & $2$ \\ \hline
    \rule[-.2cm]{0pt}{.6cm}$n_1 =0, n_2 \geq 0$ & $\ket{n_1,n_2;1} \phantom{,\ket{n_1,n_2;2}, \ket{n_1, n_2; 3}}$ & $1$ \\ \hline
  \end{tabular}
  \caption{\label{tab:1} Invariant subspaces of the extended path algebra $\PAext(\Gamma_C)$ for the cascade.}
\end{table}

\begin{table}[h]
  \centering
  \begin{tabular}{|C{2.5cm}|C{7.5cm}|C{3cm}|}
    \hline
    \rule[-.2cm]{0pt}{.6cm}$n_1, n_2$ & basis of $\mathcal{H}^{(n_1,n_2)}$ & $\dim(\mathcal{H}^{(n_1,n_2)})$ \\ \hline \hline
    \rule[-.2cm]{0pt}{.6cm}$n_1 > 0, n_2 > 0$ & $\ket{n_1,n_2;1}, \ket{n_1,n_2;2}, \ket{n_1, n_2; 3}$ & $3$ \\ \hline
    \rule[-.2cm]{0pt}{.6cm}$n_1>0, n_2=0$ & $\ket{n_1,n_2;1}, \ket{n_1,n_2;2} \phantom{, \ket{n_1, n_2; 3}}$ & $2$ \\ \hline
    \rule[-.2cm]{0pt}{.6cm}$n_1=0, n_2>0$ & $\phantom{\ket{n_1,n_2;1},} \ket{n_1,n_2;2}, \ket{n_1, n_2; 3}$ & $2$ \\ \hline 
    \rule[-.2cm]{0pt}{.6cm}$n_1 =0, n_2 = 0$ & $\phantom{\ket{n_1,n_2;1},} \ket{n_1,n_2;2}\phantom{, \ket{n_1, n_2; 3}}$ & $1$ \\ \hline
  \end{tabular}
  \caption{\label{tab:2} Invariant subspaces of the extended path algebra $\PAext(\Gamma_V)$ for the $V$-shaped configuration.}
\end{table}

\begin{table}[h]
  \centering
  \begin{tabular}{|C{2.5cm}|C{7.5cm}|C{3cm}|}
    \hline
    \rule[-.2cm]{0pt}{.6cm}$n_1, n_2$ & basis of $\mathcal{H}^{(n_1,n_2)}$ & $\dim(\mathcal{H}^{(n_1,n_2)})$ \\ \hline \hline
    \rule[-.2cm]{0pt}{.6cm}$n_1 > 0, n_2 \geq 0$ & $\ket{n_1,n_2;1}, \ket{n_1,n_2;2}, \ket{n_1, n_2; 3}$ & $3$ \\ \hline
    \rule[-.2cm]{0pt}{.6cm}$n_1=0, n_2\geq 0$ & $\ket{n_1,n_2;1}\phantom{, \ket{n_1,n_2;2}, \ket{n_1, n_2; 3}}$ & $1$ \\ \hline
    \rule[-.2cm]{0pt}{.6cm}$n_1>0, n_2=-1$ & $\phantom{\ket{n_1,n_2;1}, \ket{n_1,n_2;2},} \ket{n_1, n_2; 3}$ & $1$ \\ \hline 
  \end{tabular}
  \caption{\label{tab:3}Invariant subspaces of the extended path algebra $\PAext(\Gamma_\Lambda)$ for the $\Lambda$-shaped configuration.}
\end{table}

\begin{table}[h]
  \centering
  \begin{tabular}{|C{2cm}|C{3.6cm}|C{3.6cm}|C{3.6cm}|}
    \hline
    \rule[-.2cm]{0pt}{.6cm} & $\Gamma_C$ & $\Gamma_V$ & $\Gamma_\Lambda$ \\ \hline \hline
    \rule[-.2cm]{0pt}{.6cm} $\ket{n_1,n_2; 1}$ & $\ket{1} \otimes \ket{n_1} \otimes \ket{n_2}$ & $\ket{1} \otimes \ket{n_1-1} \otimes \ket{n_2}$ &
    $\ket{1} \otimes \ket{n_1} \otimes \ket{n_2}$ \\ \hline
    \rule[-.2cm]{0pt}{.6cm} $\ket{n_1,n_2; 2}$ & $\ket{2} \otimes \ket{n_1-1} \otimes \ket{n_2}$ & $\ket{2} \otimes \ket{n_1} \otimes \ket{n_2}$ &
    $\ket{2} \otimes \ket{n_1-1} \otimes \ket{n_2}$ \\ \hline 
    \rule[-.2cm]{0pt}{.6cm} $\ket{n_1,n_2; 3}$ & $\ket{3} \otimes \ket{n_1-1} \otimes \ket{n_2-1}$ & $\ket{3} \otimes \ket{n_1} \otimes \ket{n_2-1}$ &
    $\ket{3} \otimes \ket{n_1-1} \otimes \ket{n_2+1}$ \\ \hline 
  \end{tabular}
  \caption{\label{tab:4} Definition of the basis $\ket{n_1, n_2; v}$ for the different graphs. Note that the vectors in the table are zero if $n_1, n_2$ or $n_1
    -1$, $n_2 -1$ are negative.}
\end{table}
\medskip

\paragraph{The path algebra}
Let us come back to the decomposition in Equation (\ref{eq:36}). It gives rise to a double sequence of projections 
\begin{equation}
  E^{(n_1,n_2)} : \mathcal{H} \rightarrow \mathcal{H}^{(n_1,n_2)},\quad \left(E^{(n_1,n_2)}\right)^* = E^{(n_1,n_2)},\quad
  \left(E^{(n_1,n_2)}\right)^2 = E^{(n_1,n_2)}
\end{equation}
Applying Definition \ref{def:2} to this family we can introduce block-diagonal operators and in analogy to Eq. (\ref{eq:33}) the set (recall from
above that $D_\Gamma$ does not depend on the choice $\Gamma=\Gamma_{\#}$):
\begin{equation}
  \PAc(\Gamma_{\#}) = \{ A : D_{\Gamma_{\#}} \rightarrow D_{\Gamma_{\#}}\,|\, A\ \text{is linear and block diagonal}\}.
\end{equation}
By definition all elements of $\PAc(\Gamma_{\#})$ are of the form $A \psi = \sum_{(n_1,n_2) \in \Delta} A^{(n_1,n_2)} \psi$ for $\psi \in
D_\Gamma$. Therefore we can introduce as in Section \ref{sec:example-1:-two} the seminorms 
\begin{equation}
  \|A\|^{(n_1,n_2)} = \left\| A^{(n_1,n_2)} \right\|, \quad (n_1,n_2) \in \Delta,
\end{equation}
where $\left\| A^{(n_1,n_2)} \right\|$ denotes the operator norm of the (bounded!) operator $A^{(n_1,n_2)} \in \mathcal{B}(\mathcal{H}^{(n_1,n_2)})$. Again it is
easy to see that $\PAc$ together with the $\|\,\cdot\,\|^{(n_1,n_2)}$ is a Frechet space and a topological *-algebra. It contains subgroups
$\mathrm{U}(\Gamma_{\#})$, $\mathrm{SU}(\Gamma_{\#})$ and Lie-subalgebras $\mathfrak{u}(\Gamma_{\#})$, $\mathfrak{su}(\Gamma_{\#})$ and
$\mathfrak{sl}(\Gamma_{\#})$ which are defined as in Eqs. (\ref{eq:34}) to (\ref{eq:35}). The relation between $\PAc(\Gamma_{\#})$ and
$\PAext(\Gamma_{\#})$ is now given by the following Proposition (cf. Lemma \ref{lem:13}):

\begin{prop} \label{prop:6}
  The smallest complex Lie-subalgebra $\mathfrak{g}$ of $\PAc(\Gamma_{\#})$ which contains all $Y^{(e)}$, $Z^{(e)}$, $e \in E_+(\Gamma_{\#})$ and is closed in
  $\PAc(\Gamma_{\#})$ is $\mathfrak{sl}(\Gamma_{\#})$. 
\end{prop}

\begin{proof}
  The structure of this proof is very similar to calculations we have already done in Sections \ref{sec:dynamical-group} and
  \ref{sec:full-controlability}. Therefore we will only give a sketch and leave the details as an exercise for the reader (cf. also
  \cite{HofMa}). 
  
  For the rest of the proof let us write $E_+(\Gamma_{\#}) = \{e_1, e_2\}$ with $e_1 = (1,2)$ or $(2,1)$ and $e_2=(2,3)$ or $(3,2)$. Each such $e_j$
  defines a subgraph $K_{2,j}$ isomorphic to $K_2$. Hence there is a corresponding embedding $T_j$ of $\PAc(K_2)$ into $\PAc(\Gamma_{\#})$. For operators
  $A=\KB{\alpha}{\beta} \otimes \KB{n}{m} \in \PAc(K_2)$ with $\alpha,\beta=1,2$ and $n,m \in \Bbb{N}_0$ the map $T_1$ is given by (the strongly
  converging series)
  \begin{equation}
    T_1(A) = A \otimes \Bbb{1} = \sum_{k=0}^\infty \KB{\alpha}{\beta} \otimes \KB{n}{m} \otimes \kb{k} \in \mathcal{B}\left(\Bbb{C}^3 \otimes
    \mathrm{L}^2(\Bbb{R}) \otimes \mathrm{L}^2(\Bbb{R})\right).
  \end{equation}
  Likewise we get
  \begin{equation}
    T_2(B) = \sum_{l=0}^\infty \KB{\gamma}{\delta} \otimes \kb{l} \otimes \KB{p}{q} \in \mathcal{B}\left(\Bbb{C}^3 \otimes
    \mathrm{L}^2(\Bbb{R}) \otimes \mathrm{L}^2(\Bbb{R})\right).
  \end{equation}
  for $B=\KB{\gamma}{\delta} \otimes \KB{p}{q} \in \PAc(K_2)$. Hence the Hamiltonians $Y^{(e_j)}, Z^{(e_j)}$ can be derived from $Y, Z \in \PAc(K_2)$ 
  by $T_j(Y) = Y^{(e_j)}$, $T_j(Z) = Z^{(e_j)}$. Together with Lemma \ref{lem:13} this shows that the closed, complex Liealgebra generated by
  $Y^{(e_j)}, Z^{(e_j)}$ is isomorphic to $\mathfrak{sl}(K_2)$ and therefore contains operators $T_j(A)$, $T_j(B)$ with $A, B$ from above and in
  $\mathfrak{sl}(K_2)$. Calculating commutators of the form
  \begin{equation}
    [T_1(A), T_2(B)] = \bigl(\delta_{\beta\gamma} \KB{\alpha}{\delta} - \delta_{\alpha\delta} \KB{\gamma}{\beta}\bigr) \otimes \KB{n}{m} \otimes
    \KB{p}{q} 
  \end{equation}
  leads to the result.
\end{proof}

Since $Y^{(e)}, Z^{(e)} \in \PAc(\Gamma_{\#})$ for all $e \in E_+(\Gamma_{\#})$ and due to the properties of the exponential map on the (formally) selfadjoint
elements (cf. Proposition \ref{prop:3}) of $\PAc(\Gamma_{\#})$ we immediately get the following two corollaries.

\begin{kor}
  The extended path algebra $\PAext(\Gamma_{\#})$ is dense in $\PAc(\Gamma_{\#})$.
\end{kor}

\begin{kor}
  The dynamical group $\mathcal{G}\bigl(Y^{(e)}, Z^{(e)}; e \in E_+(\Gamma_{\#})\bigr)$ coincides with $\mathrm{SU}(\Gamma_{\#})$.
\end{kor}

These results represent the same level of structure as Propositions \ref{prop:4} and \ref{prop:2} do for two-level systems. Another similarity is that
we can reinterpret the results by introducing operators $Q_j$, $j=1,2$ with the $\mathcal{H}^{(n_1,n_2)}$ as eigenspaces and $n_j$, $j=1,2$ as the
corresponding eigenvalues. These operators define a joint symmetry of the Hamiltonians $Y^{(e)}, Z^{(e)}$, and therefore we can introduce
$\mathrm{U}(\Gamma_{\#})$ alternatively as the group of all unitaries commuting with $Q_1, Q_2$. In analogy to \cite{KZSH} we could write therefore
$\mathrm{U}(Q_1,Q_2)$ rather than $\mathrm{U}(\Gamma_{\#})$. The problem with this point of view is that the enumeration we have used for the Hilbert
spaces $\mathcal{H}^{(n_1,n_2)}$ and the $Q_1, Q_2$ is up to a certain degree arbitrary.  E.g. by using an enumeration
of $\Delta$ in terms of positive integers we could replace $Q_1, Q_2$ by just one operator $\tilde{Q}$. Hence the description of the model in terms
constants of motion like $Q_1, Q_2$ (as introduced in \cite{KZSH}) should be regarded as a description in terms of \emph{coordinates} while the path
algebra delivers the \emph{invariant} picture.

\paragraph{State preparation}
The last topic we want to treat in this section is the transformation of an arbitrary pure state $\psi \in \mathcal{H}$ into the ground state by a
sequence of unitaries $U_j$ which are either from $\mathrm{SU}(\Gamma_{\#})$ or of the form $U_j = \exp(i t X^{(e)})$. Our first step is to introduce
some notation. We write: 
\begin{equation} \label{eq:43}
  \psi^{(n_1,n_2)} = E^{(n_1,n_2)} \psi\quad\text{hence}\quad \psi = \sum_{(n_1,n_2) \in \Delta} \psi^{(n_1,n_2)}.
\end{equation}
The vectors $\psi^{(n_1,n_2)} \in \mathcal{H}^{(n_2,n_2)}$ can be expanded into the basis $\ket{n_1,n_2;v}$, which leads to
\begin{equation}
  \psi^{(n_1,n_2)} = \sum_{v=1}^3 \psi^{(n_1,n_2)}_v \ket{n_1,n_2;v} \quad \text{hence}\quad \psi = \sum_{(n_1,n_2)\in \Delta} \sum_{v=1}^3
  \psi^{(n_1,n_2)}_v \ket{n_1,n_2;v}.
\end{equation}
Note again that some of the vectors $\ket{n_1,n_2;v}$ can be zero if the corresponding Hilbert space $\mathcal{H}^{(n_1,n_2)}$ is not
3-dimensional. This convention saves us from some otherwise cumbersome case distinctions. As another notational convention the state of the system
after each discrete timestep $j$ will be denoted by $\psi$ and not by $\psi_j$. This is another choice we have made to keep the notation simple and
not too confusing. 

Now let us choose an arbitrary $\epsilon>0$ and $N,M \in \Bbb{N}$ such that
\begin{equation} \label{eq:45}
  \left\| \psi^{[N,M]} - \psi \right\| < \epsilon \quad\text{with}\quad \psi^{[N,M]} = \sum_{n \leq N, m \leq M} \psi^{(n,m)}. 
\end{equation}
If we want to transform $\psi$ into any state $\ket{n,m;v}$ with $n \leq N$, $m \leq M$ and if we are happy with an error smaller than $\epsilon$, we
only have to take the components $\psi^{(n_1,n_2)}$ with $n_1 \leq N$, $n_2 \leq M$ into account. Without loss of generality we will therefore assume
that $\psi=\psi^{[N,M]}$ holds.

\begin{figure}[h]
  \centering
  \bigskip
  \includegraphics[width=0.8\textwidth]{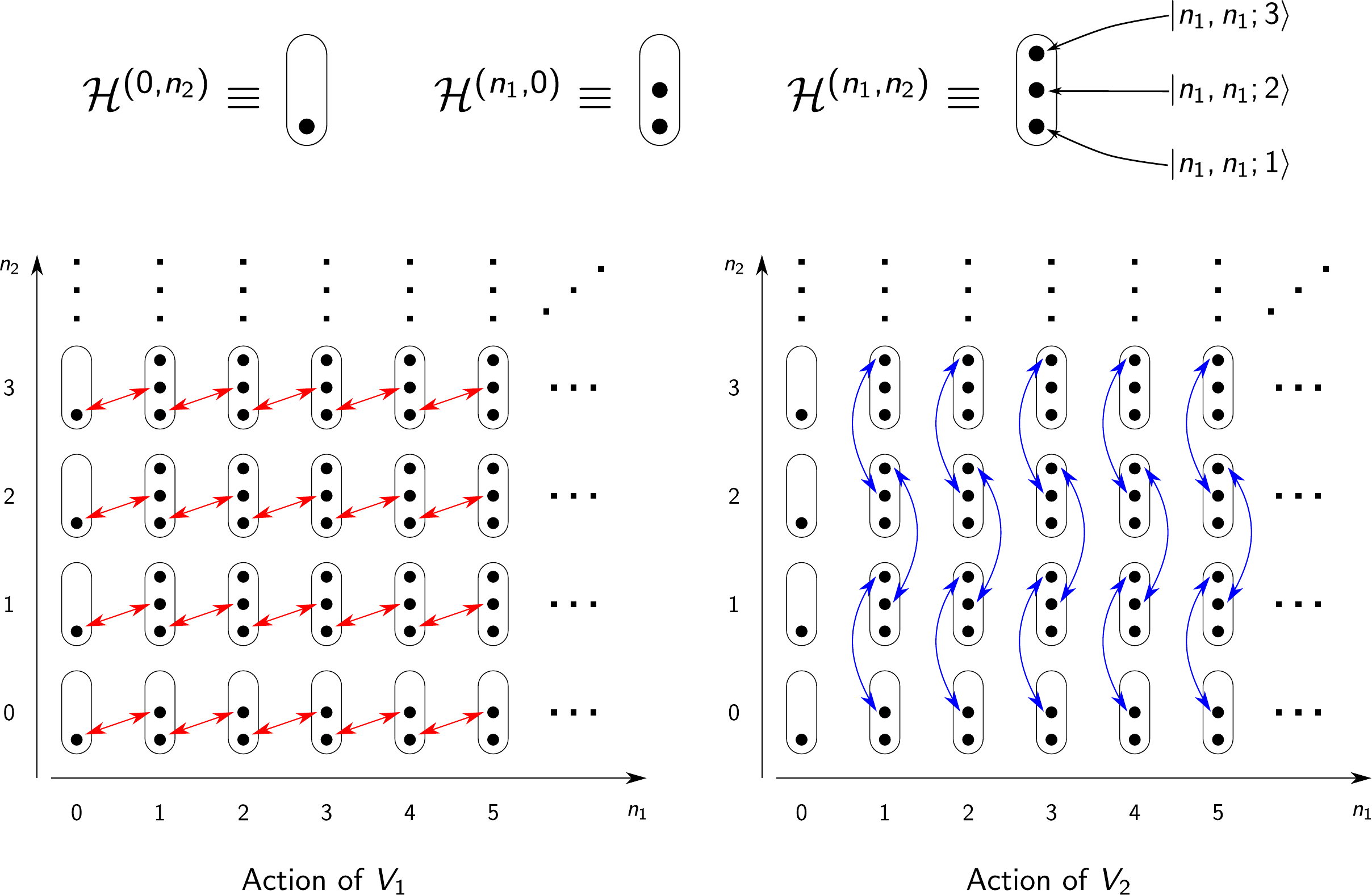}
  \caption{Action of control unitaries. As in Fig. \ref{fig:2ddecomp}, the boxes with one two or three dots stand for the subspaces
    $\mathcal{H}^{(n_1,n_2)}$, while the dots represent the basis vectors. The unitaries $V_1$, $V_2$ exchange pairs of vectors as indicated by the
    red and blue arrows. In both cases all pairs have to be exchanged simultaneously. The unitaries from $\mathrm{SU}(\Gamma_C)$ on the other hand can
  act on each box individually (but they can not leave the boxes).  \label{fig:cascade-action}}
\end{figure}

With this prerequisites we will show for the cascade (i.e. $\Gamma=\Gamma_C$) how $\psi$ can be transformed into $\ket{0,0;1} \in
\mathcal{H}^{(0,0)}$, by using unitaries $U \in \mathrm{SU}(\Gamma_C)$ and
\begin{equation} \label{eq:46}
  V_1 =\exp\left(\frac{\pi}{2} X^{(1,2)}\right) = i X^{(1,2)} \quad V_2 =\exp\left(\frac{\pi}{2} X^{(2,3)}\right) = i X^{(2,3)}.
\end{equation}
Please check yourself that the last equation holds with $X^{(\alpha,\beta)}$ from Eq. (\ref{eq:37}). The changes which are necessary to cover the
cases $\Gamma_V$ and $\Gamma_\Lambda$ are sketched below. To understand the following procedure it is useful to have a look on
Fig. \ref{fig:cascade-action} and to keep the contents of tables \ref{tab:1} and \ref{tab:4} in mind.

\begin{enumerate}
\item \label{item:12}
  The first step is to map the components $\psi^{(0,n_2)}$ in the one-dimensional subspaces to zero. This is done by applying a unitary $U \in
  \mathrm{SU}(\Gamma_C)$ which rotates all components $\psi^{(1,n_2)}$ with $n_2 > 0$ towards $\ket{1,n_2;3}$ such that after the operation
  $\psi_1^{(1,n_2)}$ and $\psi_2^{(1,n_2)}$ are zero. Note that this is possible since $\mathrm{SU}(\Gamma_C)$ contains all block-diagonal unitaries with
  blocks of determinant one. Then we apply $V_1$ which exchanges (up to a factor $i$) the vectors $\ket{0,n_2;1}$ with $\ket{1,n_2;2}$. Hence for
  all $n_2 > 0$ the components $\psi^{(0,n_2)}$ become zero, as stated.
\item \label{item:11}
  For each $n_1 > 0$ we can decrement the biggest index $n_2=M$ with $\psi^{(n_1,n_2)} \neq 0$ by one. Again, we use a two-step procedure. We apply a
  $U \in \mathrm{SU}(\Gamma_C)$ such that all $\psi^{(n_1,M)}$ are rotated towards $\ket{n_1,M;3}$ and all $\psi^{(n_1,M-1)}$ towards
  $\ket{n_1,M-1;1}$. Applying $V_2$ exchanges (again up to a factor $i$) the vectors $\ket{n_1,M;3}$ with $\ket{n_1,M-1;2}$
\item \label{item:13}
  We repeat this procedure until the only nonzero $\psi^{(n_1,n_2)}$ are those with $n_2=0$.
\item \label{item:14}
  According to table \ref{tab:1} the subspaces $\mathcal{H}^{(n_1,0)}$ with $n_1 > 0$ are two-dimensional, while $\mathcal{H}^{(0,0)}$ is
  one-dimensional. This is exactly the scenario studied in the previous section. Hence we can apply the procedure already used in the two level case
  with $X^{(1,2)}$ as the ``symmetry breaking'' Hamiltonian; cf. Section \ref{sec:example-1:-two}. This maps the vector $\psi$ eventually to
  $\ket{0,0,1}$. 
\end{enumerate}

We see that the ``exceptional'' subspaces, i.e. those with dimension one or two, need a special treatment. Therefore the procedure is easily adapted
to $\Gamma_V$ and $\Gamma_\Lambda$ by changing only the treatment of these exceptions. For $\Gamma=\Gamma_V$ the exceptions arise for $n_1=0$
and $n_2=0$. The difference to $\Gamma_C$ is that they both lead to two-dimensional subspaces (cf. table \ref{tab:2}). Hence we can skip step
\ref{item:12} and start immediately with step \ref{item:11}. As a result we map $\psi$ to $\ket{0,0;2}$. Note that for $\Gamma_V$ the vertex $v_0=2$
is -- as the global minimum -- the reference vertex (and not $v_0=1$ as for the cascade). For $\Gamma=\Gamma_\Lambda$ we choose again $v_0 = 1$ as
the reference vertex and map $\psi$ to $\ket{0,0;1}$. The exceptional subspaces are now both one-dimensional; cf. table \ref{tab:3}. Hence the case
$n_1>0$, $n_2=-1$ needs a special treatment as well as $n_1=0$, $n_2 \geq 0$. This is done as in step \ref{item:12} above: We use a combination of $U
\in \mathrm{SU}(\Gamma_\Lambda)$ and $V_2$ to map the components $\psi^{(n_1,-1)}$ to zero. In Step \ref{item:13} we continue until only the
components $\psi^{(n_1,0)}$ are non-zero and in step \ref{item:14}  we take care that $\psi^{(n_1,0)}_1 = 0$ holds all the time. This restricts the
procedure to the two-dimensional subspaces spanned by $\ket{n_1,0;2}$ and $\ket{n_1;0;3}$. They can be treated again in the same way as a two-level
system.

\section{Example 3: The $\Delta$ configuration}
\label{sec:example-3:-delta}

\begin{wrapfigure}{l}{0.3\textwidth}
  \begin{center}
    \includegraphics[width=0.2\textwidth]{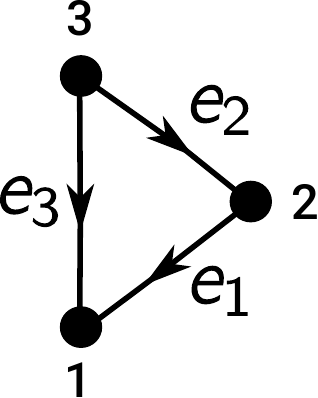}
  \end{center}
  \caption{The graph $\Gamma_\Delta$.\label{fig:Delta}}
\end{wrapfigure}
Finally, let us have a look at the 3-graph we have excluded in the last section: The $\Delta$-configuration $\Gamma_\Delta$ shown in
Fig. \ref{fig:Delta}. The set of vertices is (as before) $V(\Gamma_\Delta) = \{1,2,3\}$ and the edges are given by $E_+(\Gamma_\Delta)=\{ (2,1), (3,2),
(3,1)\}$ and $E(\Gamma_\Delta)$ containing in addition the negative edges $(b,a)$ for $(a,b) \in E_+(\Gamma_\Delta)$. For later reference let us also
introduce the enumeration 
\begin{equation} \label{eq:42}
  e_1 = (2,1),\quad e_2=(3,2),\quad e_3=(3,1).
\end{equation}
In contrast to $\Gamma_C, \Gamma_V$ and $\Gamma_\Lambda$ the graph $\Gamma_\Delta$ is not a tree but contains a cycle. This renders some of the
results from the previous section invalid.

At a first glance the differences between $\Gamma_\Delta$ and the tree graphs are not that visible, since the basic setups look quite similar. For
$\Gamma_\Delta$ the system Hilbert space is $\mathcal{H} = \Bbb{C}^3 \otimes \mathrm{L}^2(\Bbb{R})^{\otimes 3}$ rather than $\Bbb{C}^3 \otimes
\mathrm{L}^2(\Bbb{R})^{\otimes 2}$. This change in the number of tensor factors affects the definition of control Hamiltonians a bit. We have for
$(\alpha,\beta) \in E_+(\Gamma_\Delta)$
\begin{equation} 
  X^{(\alpha,\beta)} = \bigl(\KB{\alpha}{\beta} - \KB{\beta}{\alpha}\bigr) \otimes \Bbb{1}^{\otimes 3},\quad Y^{(\alpha,\beta)} = \bigl(\kb{\alpha} -
  \kb{\beta}\bigr) \otimes \Bbb{1}^{\otimes 3}
\end{equation}
and
\begin{align}
  Z^{(2,1)} &= \KB{2}{1} \otimes a^* \otimes \Bbb{1} \otimes \Bbb{1} + \KB{1}{2} \otimes a \otimes \Bbb{1} \otimes \Bbb{1}\\ 
  Z^{(3,2)} &= \KB{3}{2} \otimes \Bbb{1} \otimes a^* \otimes \Bbb{1}  + \KB{2}{3} \otimes \Bbb{1} \otimes a \otimes \Bbb{1}\\
  Z^{(3,1)} &= \KB{3}{1} \otimes \Bbb{1} \otimes \Bbb{1} \otimes a^* + \KB{2}{3} \otimes \Bbb{1} \otimes \Bbb{1} \otimes a 
\end{align}
With these definition the expressions for the drift Hamiltonian (\ref{eq:40}), and the control problems with (\ref{eq:31}) and without drift
(\ref{eq:41}) can be carried over from the last section without any changes.

Substantial differences arise in the structure of the path algebra $\PA(\Gamma_\Delta)$. As before the task is to determine the minimal invariant
subspaces $\mathcal{H}_\beta \subset \mathcal{H}$, and to label them in an unambiguous way. To do this, we will use again the photon game, introduced
in the previous section. The first step is to identify the straight paths in the graph $\Gamma_\Delta$. Hence, assume we are sitting in the vertex
$1 \in V(\Gamma_\Delta)$. To walk along a straight path on the graph $\Gamma$ we have to decide whether we want to move clockwise or
counter-clockwise. If we choose the latter we reach vertex $2 \in V(\Gamma_\Delta)$. Unless we want to stay here, there is no choice left
where to go: Since the path should be straight we can not go back. The only option is to proceed in counter-clockwise direction to reach vertex
$3$. In this way we have to proceed until we reach the end of our walk. Similar reasoning applies if our first step goes into clockwise direction; cf.
Figure \ref{fig:deltagame}. The example shows that the set of straight path is parametrized by three quantities: the start vertex, the direction
(clockwise or counter clockwise) and the length of the path.

Let us apply this to the photon game. We start with regular configuration $b \in \ConfP(\Gamma_\Delta)$ and rewrite it as a 4-tuple $b=(b_0, n_1, n_2,
n_3)$ with $n_j = \underline{b}(e_j)$; cf. Eq. (\ref{eq:42}). For simplicity also assume that $b_0=1$. The other cases are easily adapted. If $n_1 >0$
and $n_2 > 0$, we can move counter-clockwise and decrement the numbers $n_1$, $n_2$ while we pass the edges $\overline{e_1}$, $\overline{e_2}$. The
last number $n_3$ is incremented since our move along $e_3$ respect the edge's orientation. In this way we can perform $N = \min{n_1, n_2}$ full
cycles. After that either $n_1$ or $n_2$ become zero. If $n_1=0$ our walk ends at vertex $1$. If $n_1 > 0$ and $n_2=0$ we end at vertex $2$. If
initially $n_3 > 0$ holds we can move in clockwise direction, too. Since $e_3$ is passed against it orientation we have to decrement $e_3$ (and
increment $n_1,n_2$). After $n_3$ full cycles we end at vertex $1$ with $n_3=0$; cf. Fig. \ref{fig:deltagame}. This simple reasoning shows that
following statement holds:

\begin{figure}[t]
  \centering
  \includegraphics[width=0.8\textwidth]{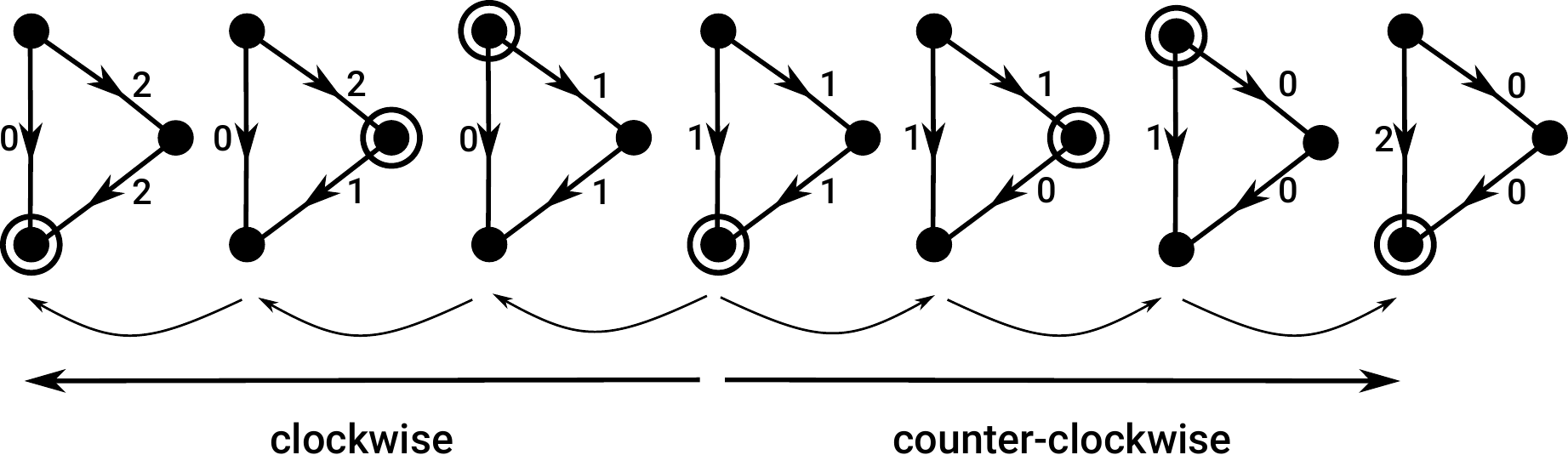}
  \caption{Playing the photon game on the graph $\Gamma_\Delta$. \label{fig:deltagame}} 
\end{figure}

\begin{prop}
  Consider a regular configuration $b \in \ConfP(\Gamma_\Delta)$ and the corresponding minimal, invariant subspace $\mathcal{H}_b$ of
  $\PA(\Gamma_\Delta)$. There is exactly one $c \in \ConfP(\Gamma_\Delta)$ with $c_0=1$ and $\underline{c}(e_3)=0$. 
\end{prop}

Hence, a complete, unambiguous labelling of invariant subspaces of $\PA(\Gamma_\Delta)$ is given by the rule
\begin{equation}
  \mathcal{H}^{(n_1,n_2)} = \mathcal{H}_c\quad \text{with}\quad c=(1,n_1,n_2,0)\quad (n_1,n_2) \in \Bbb{N}_0^2.
\end{equation}
The structure of the $\mathcal{H}^{(n_1,n_2)}$ can also be deduced easily from the discussion of straight paths given above. We just have to start
with the configuration $(1,n_1,n_2,0)$ and move counter-clockwise around the graph until $n_1$ or $n_2$ reach zero. This shows that the dimension of
$\mathcal{H}^{(n_1,n_2)}$ is given by
\begin{equation}
  \dim\left(\mathcal{H}^{(n_1,n_2)}\right) = \begin{cases} 3L+1 & \text{if $n_1 \leq n_2$}\\ 3L+2 & \text{if $n_1 > n_2$} \end{cases} \quad
  \text{with}\quad L=\min(n_1,n_2).
\end{equation}
We can also find a relabelling of the canonical basis, which is adapted to the decomposition of $\mathcal{H}$ into a direct sum of the
$\mathcal{H}^{(n_1,n_2)}$. In the following we write $\ket{b_0,n_1,n_2,n_2}$ for $\ket{b}$, if $b \in \ConfP(\Gamma_\Delta)$ satisfies
$\underline{b}(e_j)=n_j$ for $j=1,2,3$.
\begin{equation} \label{eq:44}
  \ket{n_1,n_2; m, \nu} = \begin{cases} \ket{1,n_1-m, n_2-m, m} & \text{if $\nu=1$}\\ \ket{2, n_1-m-1,n_2-m} & \text{if $\nu=2$}\\ \ket{3,
      n_1-m-1,n_2-m-1} & \text{if $\nu=3$}.
  \end{cases}
\end{equation}
To simplify notations we define $\ket{n_1,n_2,m,\nu}=0$ whenever one of the quantities $n_1-m$, $n_2-m$, $n_1-m-1$ or $n_2-m-1$ becomes negative. This
saves us from giving precise (and cumbersome) index ranges whenever we expand a vector in this basis.

This analysis already reveals the basic difference between $\Gamma_\Delta$ and the tree-graphs from the previous section: In the latter case the
dimension of the Hilbert spaces $\mathcal{H}^{(n_1,n_2)}$ is bounded by $3$, while for $\Gamma_\Delta$ it can be arbitrarily large. Note, however,
that in both cases the chosen labelling of the invariant subspaces only involves a pair $(n_1,n_2) \in \Bbb{N}_0^2$. 

From here on we can proceed as in the last two sections. The Hilbert space decomposes as
\begin{equation}
  \mathcal{H} = \bigoplus_{(n_1,n_2) \in \Bbb{N}_0^2} \mathcal{H}^{(n_1,n_2)}\quad\text{with projections}\ E^{(n_1,n_2)} : \mathcal{H} \rightarrow
  \mathcal{H}^{(n_1,n_2)} 
\end{equation}
and we can define the corresponding algebra of block-diagonal operators,
\begin{equation}
  \PAc(\Gamma_{\Delta}) = \{ A : D_{\Gamma_{\Delta}} \rightarrow D_{\Gamma_{\Delta}}\,|\, A\ \text{is linear and block diagonal}\},
\end{equation}
where $D_{\Gamma_\Delta}$ is the domain we have defined in Eq. (\ref{eq:23}). $\PAc(\Gamma_\Delta)$ becomes a Frechet space if we equip it with the
seminorms
\begin{equation}
  \|A\|^{(n_1,n_2)} = \left\| A^{(n_1,n_2)} \right\|, \quad (n_1,n_2) \in \Bbb{N}_0^2.
\end{equation}
As in Eq. (\ref{eq:34}) to (\ref{eq:35}) we can define the subgroups and Lie-subalgebras $\mathrm{U}(\Gamma_\Delta)$, $\mathrm{SU}(\Gamma_\Delta)$,
$\mathfrak{u}(\Gamma_\Delta)$, $\mathfrak{su}(\Gamma_\Delta)$ and $\mathfrak{sl}(\Gamma_\Delta)$. With all this notations Prop. \ref{prop:6} from
Section \ref{sec:example-2:-three} carries over without any change:

\begin{prop} \label{prop:7}
  The smallest complex Lie-subalgebra $\mathfrak{g}$ of $\PAc(\Gamma_{\Delta})$ which contains all $Y^{(e)}$, $Z^{(e)}$, $e \in E_+(\Gamma_{\Delta})$ and is closed in
  $\PAc(\Gamma_{\Delta})$ is $\mathfrak{sl}(\Gamma_{\Delta})$. 
\end{prop}

The proof is done in the same way as in Prop. \ref{prop:6}: We embed three copies of $\PAc(K_2)$ into $\PAc(\Gamma_\Delta)$ and calculate commutators
of overlapping operators. The details can be found in \cite{HofMa}. From Prop. \ref{prop:7} we immediately get the following corollary:

\begin{kor}
  $\PA(\Gamma_\Delta)$ is dense in $\PAc(\Gamma_\Delta)$.
\end{kor}

The only thing left is the state preparation. At a first glance we expect big differences to the tree graphs in the last section. A little bit
surprisingly, however, we can proceed almost without any change. Compared to the treatment of the cascade $\Gamma_C$ only one extra step is
needed. Let us first adopt the notations from Sect. \ref{sec:example-2:-three}. As in Eq. (\ref{eq:43}) we decompose $\psi \in \mathcal{H}$ as
\begin{equation}
  \psi = \sum_{(n_1,n_2) \in \Bbb{N}_0^2} \psi^{(n_1,n_2)}\quad \text{with}\quad \psi^{(n_1,n_2)} = E^{(n_1,n_2)} \psi
\end{equation}
The vectors $\psi^{(n_1,n_2)}$ can decomposed into the basis $\ket{n_1,n_2;m,\nu}$ as
\begin{equation}
  \psi^{(n_1,n_2)} = \sum_{m =0}^L \sum_{\nu=1}^3 \psi^{n_1,n_2}_{m,\nu} \ket{n_1,n_2;m,\nu},\quad\text{with}\quad L = \min(n_1,n_2).
\end{equation}
The only difference to Sect. \ref{sec:example-2:-three} is the additional parameter $m$. Also recall the remark about index ranges from above: the
$\ket{n_1,n_2;L,\nu}$ are zero, whenever they can not be mapped to a regular configuration via Eq. (\ref{eq:44}). Now we choose $N,M \in \Bbb{N}$,
define the cut-off vector $\psi^{[N,M]}$ as in Eq. (\ref{eq:45}) and assume $\psi = \psi^{[N,M]}$.

Now, the task is to map $\psi$ to the ground state $\ket{0,0;0,0}$ by applying untiaries from $\mathrm{SU}(\Gamma_\Delta)$ and
\begin{equation} 
  V_1 =\exp\left(\frac{\pi}{2} X^{(1,2)}\right) = i X^{(1,2)} \quad V_2 =\exp\left(\frac{\pi}{2} X^{(2,3)}\right) = i X^{(2,3)};
\end{equation}
cf. Eq. (\ref{eq:46}). To do this note first that $\dim(\mathcal{H}^{(0,n_1)})=1$ and $\dim(\mathcal{H}^{(n_1,0)})=2$ as for the cascade $\Gamma_C$. The
only difference is that the generic Hilbert spaces $\mathcal{H}^{(n_1,n_2)}$ are all \emph{exactly} three-dimensional for $\Gamma_C$, while they are
\emph{at least} four-dimensional (and becoming arbitrarily large) for $\Gamma_\Delta$.  Hence if we choose in the first step a unitary $U \in
\mathrm{SU}(\Gamma_\Delta)$ with 
\begin{equation}
  E^{(n_1,n_2)} U\psi  \in \SP \{\ket{n_1,n_2;0,j}\,|\, j=1,2,3\}\quad\text{for}\quad n_1>0,\ n_2>0
\end{equation}
we are exactly in the same situation we have discussed in Sect. \ref{sec:example-2:-three}. Therefore we can proceed with the procedure presented for
$\Gamma_C$.

To summarize our discussion, we can conclude that main difference between $\Gamma_\Delta$ and the tree graphs $\Gamma_C$, $\Gamma_V$ and
$\Gamma_\Lambda$ arise in the treatment of the invariant subspaces $\mathcal{H}_b$, $b \in \ConfP(\Gamma_{\#})$. The most obvious distinction is the
behavior of the dimensions of the $\mathcal{H}_b$. For the tree graphs they are bounded from above by three, while in the case of $\Gamma_\Delta$ they
grow indefinitely. A more subtle point is the method we have used to find a labelling for the $\mathcal{H}_b$. The discussion from the last section is
applicable to arbitrary tree graphs. The scheme developed in this section, however, does not allow an obvious generalization to graphs with more than
one cycle. If such a generalization would be available, the reasoning from the last two sections would be available for arbitrary graphs. In
particular a general formula for the dimension could answer the question whether there are two inequivalent graphs with equivalent path algebras. Our
conjecture is that this is not the case.

From a more practical point of view the model based on the delta configuration has an advantage in efficiency. We can treat three modes (rather than
two) with the same number of levels (three), and we have full controllability over Hilbert spaces of arbitrary high dimension (the
$\mathcal{H}^{(n_2,n_2)}$) by only manipulating relative phases of the atom and using the natural drift of the system.

\section{Outlook}
\label{sec:outlook}

We have studied a $d$-level atom interacting with the light field in a cavity via Hamiltonian (\ref{eq:38}), and shown that the overall system consisting of atom and field is strongly controllable, if
the internal degrees of freedom of the atom can be adequately manipulated. The latter means (as a minimal setup) that we can switch the controls $X^{(e)}, Y^{(e)}$ for all edges $e$ in a spanning tree
of $\Gamma$ individually on and off; cf. Lemma \ref{lem:7}. This is already a very useful result since it opens lots of new possibilities to manipulate electromagnetic radiation in experiments with
light or micro waves. We have, however, gained lots of additional insights into the structure of the control problem at hand.

The most important object in this context is the extended path algebra $\PAext(\Gamma)$ introduced in Section \ref{sec:path-algebra}. It is an important part of the controllability proof which allows us to use
Lie-algebraic methods at least for a subfamily of control Hamiltonian. As such it takes the role of the symmetry arguments used in \cite{KZSH} to solve the two-level case. The latter is also true, if
we look at the state preparation tasks for three level systems in Sections \ref{sec:example-2:-three} and \ref{sec:example-3:-delta}. With a clever
combination of symmetry breaking and respecting unitaries we can (approximately) prepare any state of the overall system. The procedure can be
generalized easily to any tree graph, while graphs containing cycles are more tricky any require a more detailed study.

The developed scheme can be useful in the framework of optimal control. A common strategy to handle infinite dimensional control problems like the one we are discussing, is to cut the Hilbert space
off at, e.g., finite photon numbers. In our case, however, this still would imply that the dimension of the Hilbert space under consideration grows exponentially with the cut-off
parameter. To prepare an arbitrary state of the overall system (approximately) from the ground state, we can, however, use the method from Sections \ref{sec:example-2:-three} and
\ref{sec:example-3:-delta} (and generalizations thereof) and then we only have to find the control functions for unitaries in the path algebra (the ``symmetry breaking'' unitaries are just given by
applying particularly chosen control Hamiltonians for a certain amount of time). Cutting of $\PAext(\Gamma)$ at an invariant subspace $\mathcal{H}^{(n_1,\dots,n_d)}$ (cf. Eq. (\ref{eq:39})) only leads
to a polynomial growth of dimension as a function of $(n_1,\dots,n_d)$.

Another interesting aspect of $\PAext(\Gamma)$ concerns its relation to the structure of the graph $\Gamma$. It is clear that $\PAext(\Gamma)$ contains information about $\Gamma$, but how much? For
graphs with two or three vertices our analysis has shown that the algebras are isomorphic iff the graphs are equivalent. It is an interesting question whether this observation stays true for
arbitrary graphs.

\bibliographystyle{unsrt}
\bibliography{control.bib}

\end{document}